\documentclass[12pt,english,reqno]{smfart}
\usepackage{mathptmx}
\newtheorem{theorem}{Theorem}[section]
\newtheorem{lemma}{Lemma}[section]
\newtheorem{proposition}{Proposition}[section]
\newtheorem{corollary}{Corollary}[section]

\theoremstyle{remark}
\newtheorem{remark}{Remark}[section]
\newtheorem{example}{Example}[section]
\renewcommand{\theequation}{\arabic{section}.\arabic{equation}}
\title
[Haar expectations of ratios]
{Haar expectations of ratios of\\ random characteristic polynomials}
\author{A.\ Huckleberry, A.\ P\"uttmann and M.R.\ Zirnbauer}
\date{July 29, 2015}
\begin{document}
\maketitle
\begin {abstract}
We compute Haar ensemble averages of ratios of random characteristic
polynomials for the classical Lie groups $K = \mathrm{O}_N\,$,
$\mathrm{SO}_N\,$, and $\mathrm{USp}_N\,$. To that end, we start from
the Clifford-Weyl algebra in its canonical realization on the complex
$\mathcal{A}_V$ of holomorphic differential forms for a $\mathbb{C}
$-vector space $V_0\,$. From it we construct the Fock representation
of an orthosymplectic Lie superalgebra $\mathfrak{osp}$ associated to
$V_0\,$. Particular attention is paid to defining Howe's oscillator
semigroup and the representation that partially exponentiates the Lie
algebra representation of $\mathfrak{sp} \subset \mathfrak{osp}$. In
the process, by pushing the semigroup representation to its boundary
and arguing by continuity, we provide a construction of the
Shale-Weil-Segal representation of the metaplectic group.

To deal with a product of $n$ ratios of characteristic polynomials,
we let $V_0 = \mathbb{C}^n \otimes \mathbb{C}^N$ where $\mathbb
{C}^N$ is equipped with the standard $K$-representation, and focus on
the subspace $\mathcal{A}_V^K$ of $K$-equivariant forms. By Howe
duality, this is a highest-weight irreducible representation of the
centralizer $\mathfrak{g}$ of $\mathrm{Lie}(K)$ in $\mathfrak{osp}$.
We identify the $K$-Haar expectation of $n$ ratios with the character
of this $\mathfrak{g} $-representation, which we show to be uniquely
determined by analyticity, Weyl-group invariance, certain weight
constraints and a system of differential equations coming from the
Laplace-Casimir invariants of $\mathfrak{g}\,$. We find an explicit
solution to the problem posed by all these conditions. In this way,
we prove that the said Haar expectations are expressed by a Weyl-type
character formula for \emph{all} integers $N \ge 1$. This completes
earlier work of Conrey, Farmer, and Zirnbauer for the case of
$\mathrm{U}_N\,$.
\end{abstract}
%

%\end {document}
\tableofcontents

\section{Introduction}\label{introduction}\label{varying}

In this article we derive an explicit formula for the average
%\end {document}

\begin{equation}\label{eq:1.1}
    I(t) := \int_K Z(t,k)\, dk
\end{equation}
where $K$ is one of the classical compact Lie groups $\mathrm{O}_N\,$,
$\mathrm{SO}_N\,$, or $\mathrm{USp}_N$ equipped with Haar measure $dk$
of unit mass $\int_K dk = 1$ and
\begin{equation}\label{eq:1.2}
    Z(t,k) := \prod_{j=1}^n \frac{\mathrm{Det}(
    \mathrm{e}^{\frac{\mathrm{i}}{2}\psi_j} \mathrm{Id}_N -
    \mathrm{e}^{-\frac{\mathrm{i}}{2}\psi_j} \,k)}
    {\mathrm{Det}(\mathrm{e}^{\frac{1}{2}\phi_j}\mathrm{Id}_N -
    \mathrm{e}^{- \frac{1}{2}\phi_j} \, k)}
\end{equation}
depends on a set of complex parameters $t := (\mathrm{e}^{
\mathrm{i}\psi_1} , \ldots , \mathrm{e}^{\mathrm{i}\psi_n} ,
\mathrm{e}^{\phi_1} , \ldots, \mathrm{e}^{\phi_n})$, which satisfy
$\mathfrak{Re}\, \phi_j > 0$ for all $j = 1, \ldots, n\,$. The case
of $K = \mathrm{U}_N$ is handled in \cite{CFZ}. Note that
\begin{displaymath}
    Z(t,k) = \mathrm{e}^{\lambda_N} \prod_{j=1}^n \frac{\mathrm{Det}
    (\mathrm{Id}_N-\mathrm{e}^{-\mathrm{i}\psi_j}\,k)} {\mathrm{Det}
    (\mathrm{Id}_N - \mathrm{e}^{-\phi_j}\, k)}
\end{displaymath}
with $\lambda_N = \frac{N}{2} \sum_{j=1}^n (\mathrm{i} \psi_j -
\phi_j)$. This means that $Z(t,k)$ is a product of ratios of
characteristic polynomials, which explains the title of the article.

The Haar average $I(t)$ can be regarded as the (numerical part of
the) character of an irreducible representation of a Lie supergroup
$(\mathfrak{g}\, ,G)$ restricted to a suitable subset of a maximal
torus of $G$. The Lie superalgebra $\mathfrak{g}$ is the Howe dual
partner of the compact group $K$ in an orthosymplectic Lie
superalgebra $\mathfrak{osp} \,$. It is naturally represented on a
certain infinite-dimensional spinor-oscillator module $\mathfrak{a}
(V)$ -- more concretely, the complex of holomorphic differential
forms on the vector space $\mathbb{C}^n \otimes \mathbb{C}^N$ -- and
the irreducible representation is that on the subspace
$\mathfrak{a}(V)^K$ of $K$-equivariant forms.

To even define the character, we must exponentiate the representation
of the Lie algebra part of $\mathfrak{osp}$ on $\mathfrak{a}(V)$.
This requires going to a completion $\mathcal{A}_V$ of $\mathfrak
{a}(V)$, and can only be done partially. Nevertheless, the
represented semigroup contains enough structure to derive
Laplace-Casimir differential equations for its character.

Our explicit formula for $I(t)$ looks exactly like a classical Weyl
formula and is derived in terms of the roots of the Lie superalgebra
$\mathfrak{g}$ and the Weyl group $W$. Let us state this formula for $K = \mathrm{O}_N\,$, $\mathrm{USp}_N$ without going into the details of the $\lambda$-positive even and odd roots $\Delta^+_{\lambda,0}$ and $\Delta ^+_{\lambda ,1}$ and the Weyl group $W$ (see $\S$\ref{sect:WeylGroup} for precise formulas). If $W_\lambda$ is the isotropy subgroup of $W$ fixing the highest weight $\lambda = \lambda_N\,$, then
\begin{equation}\label{eq:WeylCharacter}
    I(t)=\sum_{[w] \in W\!\! / W_\lambda} \mathrm{e}^{ w(\lambda_N)} \frac{\prod_{\beta \in \Delta_{\lambda,1}^+}(1-\mathrm{e}^{ -w(\beta)})} {\prod_{\alpha \in \Delta_{\lambda,0}^+}
    (1 - \mathrm{e}^{-w(\alpha)})} (\ln t) \;.
\end{equation}
%
% In later sections there is a conflict of notations:
% W = V \oplus V^\ast versus W = Weyl group
%
To prove this formula we establish certain properties of $I(t)$ which
uniquely characterize it and are satisfied by the right-hand side.
These are a weight expansion of $I(t)$ (see Corollary
\ref{cor:weightexpansion}), restrictions on the set of weights (see
Corollary \ref{cor:weights}), and the fact that $I(t)$ is annihilated
by certain invariant differential operators (see Corollary
\ref{cor:DlJchi}).

As was stated above, the case $K = \mathrm{U}_N$ is treated in
\cite{CFZ}. Here, we restrict to the compact groups $K = \mathrm{O}_N
\,$, $K = \mathrm{USp}_N\,$, and $K = \mathrm{SO}_N\,$. The cases $K
= \mathrm{O}_N$ and $K = \mathrm{USp}_N$ can be treated
simultaneously. Having established formula (\ref{eq:WeylCharacter})
for $K = \mathrm{O}_N\,$, the following argument gives a similar
result for $K = \mathrm{SO}_N\,$. Let $dk_\mathrm{O}$ and
$dk_\mathrm{SO}$ be the unit mass Haar measures for $\mathrm{O}_N$
and $\mathrm{SO}_N\,$, respectively.  The $\mathrm{O}_N$-measure $(1
+ \mathrm{Det}\, k)\, dk_\mathrm{O}$ has unit mass on $\mathrm{SO}_N$
and zero mass on $\mathrm{O}_N^- = \mathrm{O}_N \setminus
\mathrm{SO}_N\,$. It is $\mathrm{SO}_N$-invariant. Now,
\begin{eqnarray*}
    &&I_{\mathrm{SO}_N}(t) = \int_{\mathrm{SO}_N} Z(t,k)\,dk_\mathrm{SO}
    = \int_{\mathrm{O}_N} Z(t,k) (1+\mathrm{Det}\, k)\, dk_\mathrm{O} \\
    = &&\int_{\mathrm{O}_N} Z(t,k)\, dk_\mathrm{O} + \int_{\mathrm{O}_N}
    Z(t,k) \mathrm{Det}(k) \, dk_\mathrm{O} = I_{\mathrm{O}_N}(t) +
    (-1)^N I_{\mathrm{O}_N}(t')
\end{eqnarray*}
with $t' = (\mathrm{e}^{- \mathrm{i}\psi_1}, \mathrm{e}^{ \mathrm{i}
\psi_2} , \ldots, \mathrm{e}^{\mathrm{i}\psi_n}, \mathrm{e}^{\phi_1},
\ldots, \mathrm{e}^{\phi_n})$, since $\mathrm{Det}(k)\, Z(t,k) =
(-1)^N Z(t',k)$.

\subsection{Comparison with results of other approaches}

To facilitate the comparison with related work, we now present our
final results in the following explicit form. Let $x_j := \mathrm{e}
^{- \mathrm{i}\psi_j}$ and $y_l := \mathrm{e}^{-\phi_l}$. Consider
first the case of the unitary symplectic group $K = \mathrm{USp}_N$
(where $N \in 2\mathbb{N}$). Then for any pair of non-negative
integers $p,q$ in the range $q - p \le N+1$ one directly infers from
(\ref{eq:WeylCharacter}) the formula
\begin{displaymath}
    \int\limits_{\mathrm{USp}_N} \frac{\prod_{k=1}^p
    \mathrm{Det}(1 - x_k \, u)} {\prod_{l=1}^q \mathrm{Det}
    (1 - y_l \, u)} \, du = \sum_{\epsilon \in \{\pm 1\}^p} \,
    \frac{\prod_{k=1}^p x_k^{\frac{N}{2}(1-\epsilon_k)}
    \prod_{l=1}^q (1 - x_k^{\epsilon_k} y_l)}{\prod_{k\le k^\prime}
    (1 - x_k^{\epsilon_k} x_{k^\prime}^{\epsilon_{k^\prime}})
    \prod_{l < l^\prime}(1 - y_l y_{l^\prime})} \;.
\end{displaymath}
The sum on the right-hand side is over sign configurations $\epsilon
\equiv (\epsilon_1, \ldots, \epsilon_p) \in \{ \pm 1 \}^p$. The proof
proceeds by induction in $p\,$, starting from the result
(\ref{eq:WeylCharacter}) for $p = q$ and sending $x_p \to 0$ to pass
from $p$ to $p-1$. In published recent work \cite{cfs,bg}
%
% no longer quite so recent ...
%
the same formula was derived under the more restrictive condition $q \le
N/2\,$. In \cite{bg} this unwanted restriction on the parameter range
came about because the numerator and denominator on the left-hand
side were expanded \emph{separately}, ignoring the supersymmetric
Howe duality (see $\S$\ref{sect:osp-howe} of the present paper) of
the problem at hand.

For $K = \mathrm{SO}_N$ the same induction process starting from
(\ref{eq:WeylCharacter}) yields the result
\begin{displaymath}
    \int\limits_{\mathrm{SO}_N} \frac{\prod_{k=1}^p
    \mathrm{Det}(1 - x_k \, u)} {\prod_{l=1}^q \mathrm{Det}
    (1 - y_l \, u)} \, du = \sum_{\epsilon \in \{ \pm 1\}^p}
    \frac{\prod_{k=1}^p (\epsilon_k x_k)^{\frac{N}{2}(1-\epsilon_k)}
    \prod_{l=1}^q (1 - x_k^{\epsilon_k} y_l)}{\prod_{k < k^\prime}
    (1 - x_k^{\epsilon_k} x_{k^\prime}^{\epsilon_{k^\prime}})
    \prod_{l \le l^\prime}(1 - y_l y_{l^\prime})}
\end{displaymath}
as long as $q-p \le N-1\,$. Please note that this includes even the
case of the trivial group $K = \mathrm{SO}_1 = \{ \mathrm{Id} \}$
with any $p = q > 0\,$. For $K = \mathrm{O}_N$ one has an analogous
result where the sum on the right-hand side is over $\epsilon$ with
an \emph{even} number of sign reversals.

The very same formulas for $\mathrm{SO}_N$ and $\mathrm{O}_N$ were
derived in the recent literature \cite{cfs,bg} but, again, only in
the much narrower range $q \le \mathrm{Int}[N/2]$. There exist a
number of other interesting recent works
%
% no longer quite so recent ...
%
which emphasize the Lie superalgebraic and combinatorial side of the picture (see, e.g., \cite{CW1,CW2,CZ}).

\subsection{Howe duality and weight expansion}

To find an explicit expression for the integral $I(t)$, we first of
all observe that the integrand $Z(t,k)$ is the supertrace of a
representation $\rho$ of a semigroup $(T_1 \times T_+) \times K$ on
the spinor-oscillator module $\mathfrak{a}(V)$ (cf.\ Lemma
\ref{lem:STrAV}). More precisely, we start with the standard
$K$-representation space $\mathbb{C}^N$, the $\mathbb{Z}_2$-graded
vector space $U = U_0 \oplus U_1$ with $U_s \simeq \mathbb{C}^n$, and
the abelian semigroup
\begin{displaymath}
    T_1 \times T_+ := \{ (\mathrm{diag}(\mathrm{e}^{\mathrm{i}\psi_1},
    \ldots, \mathrm{e}^{\mathrm{i} \psi_n}), \mathrm{diag}
    (\mathrm{e}^{\phi_1}, \ldots, \mathrm{e}^{\phi_n}))
    \mid \mathfrak{Re}\, \phi_j > 0 \;, j = 1, \ldots, n \}
\end{displaymath}
of diagonal transformations in $\mathrm{GL}(U_1) \times \mathrm{GL}
(U_0)$. We then consider the vector space $V := U \otimes
\mathbb{C}^N$ which is $\mathbb{Z}_2$-graded by $V_s = U_s \otimes
\mathbb{C}^N$, the infinite-dimensional spinor-oscillator module
$\mathfrak{a}(V) := \wedge (V_1^*) \otimes \mathrm{S} (V_0^*)$, and a
representation $\rho$ of $(T_1 \times T_+) \times K$ on
$\mathfrak{a}(V)$. We also let $V \oplus V^\ast =: W = W_0 \oplus
W_1$ (not the Weyl group).

Averaging the product of ratios $Z(t,k)$ with respect to the compact
group $K$ corresponds to the projection from $\mathfrak{a}(V)$ onto
the vector space $\mathfrak{a}(V)^K$ of $K$-invariants (Corollary
\ref{cor:weightexpansion}). Now, Howe duality (Proposition
\ref{prop:Howeduality}) implies that $\mathfrak{a}(V)^K$ is the
representation space for an irreducible highest-weight representation
$\rho_\ast$ of the Howe dual partner $\mathfrak{g}$ of $K$ in the
orthosymplectic Lie superalgebra $\mathfrak{osp}(W)$. This
representation $\rho_\ast$ is constructed by realizing $\mathfrak{g}
\subset \mathfrak{osp}(W)$ as a subalgebra in the space of degree-two
elements of the Clifford-Weyl algebra $\mathfrak{q}(W)$. Precise
definitions of these objects, their relationships, and the Howe
duality statement can be found in $\S$\ref{sect:osp-howe}.

Using the decomposition $\mathfrak{a}(V)^K = \oplus_{\gamma \in
\Gamma}\, V_\gamma$ into weight spaces, Howe duality leads to the
weight expansion $I(\mathrm{e}^H) = \mathrm{STr}_{\mathfrak{a}(V)^K}
\, \mathrm{e}^{\rho_*(H)} = \sum_{\gamma \in \Gamma} B_\gamma \,
\mathrm{e}^{\gamma(H)}$ for $t = \mathrm{e}^H \in T_1 \times T_+\,$.
Here $\mathrm{STr}$ denotes the supertrace. There are strong
restrictions on the set of weights $\Gamma$. Namely, if $\gamma \in
\Gamma$, then $\gamma = \sum_{j=1}^n (\mathrm{i} m_j \psi_j - n_j
\phi_j)$ and $- \frac{N}{2}\le m_j \le \frac{N}{2} \le n_j$ for all
$j\,$. The coefficients $B_\gamma = \mathrm{STr}_{V_\gamma}
(\mathrm{Id})$ are the dimensions of the weight spaces (multiplied by
parity). Note that the set of weights of the representation $\rho_*$
of $\mathfrak{g}$ on $\mathfrak{a}(V)^K$ is infinite.

\subsection{Group representation and differential equations}

Before outlining the strategy for computing our character in the
infinite-dimensional setting of representations of Lie superalgebras
and groups, we recall the classical situation where $\rho_*$ is an
irreducible finite-dimensional representation of a reductive Lie
algebra $\mathfrak{g}$ and $\rho$ is the corresponding Lie group
representation of the complex reductive group $G$. In that case the
character $\chi $ of $\rho$, which automatically exists, is the trace
$\mathrm{Tr}\, \rho$, which is a radial eigenfunction of every
Laplace-Casimir operator. These differential equations can be
completely understood by their behavior on a maximal torus of $G$.

In our case we must consider the infinite-dimensional irreducible
representation $\rho_*$ of the Lie superalgebra $\mathfrak{g} =
\mathfrak{osp}$ on $\mathfrak{a}(V)^K$. Casimir elements,
Laplace-Casimir operators of $\mathfrak{osp}$, and their radial parts
have been described by Berezin \cite{B}. In the situation $U_0 \simeq
U_1$ at hand we have the additional feature that every
$\mathfrak{osp}$-Casimir element $I$ can be expressed as a bracket $I
= [ \partial , F ]$ where $\partial$ is the holomorphic exterior
derivative when we view $\mathfrak{a} (V)^K$ as the complex of
$K$-equivariant holomorphic differential forms on $V_0\,$.

To benefit from Berezin's theory of radial parts, we construct a
radial superfunction $\chi$ which is defined on an open set
containing the torus $T_1 \times T_+$ such that its numerical part
satisfies $\mathrm{num} \chi(t) = I(t)$ for all $t \in T_1\times
T_+\,$. If we had a representation $(\rho_*,\rho)$ of a Lie
supergroup $(\mathfrak{osp},G)$ at our disposal, we could define
$\chi$ to be its character, i.e.
\begin{displaymath}
    \chi(g) \stackrel{?}{=} \mathrm{STr}_{\mathfrak{a}(V)^K}
    \rho(g)\, \mathrm{e}^{\sum_j \xi_j \, \rho_*(\Xi_j)} .
\end{displaymath}
Since we don't have such a representation, our idea is to define $\chi$ as a character on a totally real submanifold $M$ of maximal dimension which contains a real form of $T_1 \times T_+$ and is invariant with respect to
conjugation by a real form $G_\mathbb{R}$ of $G$, and then to extend $\chi$ by analytic continuation.

Thinking classically we consider the even part of the Lie
superalgebra $\mathfrak{osp}(W_0 \oplus W_1)$, which is the Lie
algebra $\mathfrak{o}(W_1) \oplus \mathfrak{sp}(W_0)$. The real
structures at the Lie supergroup level come from a real form
$W_\mathbb{R}$ of $W$. The associated real forms of $\mathfrak{o}
(W_1)$ and $\mathfrak{sp}(W_0)$ are the real orthogonal Lie algebra
$\mathfrak{o}(W_{1,\mathbb{R}})$ and the real symplectic Lie algebra
$\mathfrak{sp}(W_{0,\mathbb{R}})$. These are defined in such a way
that the elements in $\mathfrak{o}(W_{1,\mathbb{R}}) \oplus
\mathfrak{sp} (W_{0,\mathbb{R}})$ and $\mathrm{i} W_\mathbb{R}$ are
mapped as elements of the Clifford-Weyl algebra via the
spinor-oscillator representation to anti-Hermitian operators on
$\mathfrak{a}(V)$ with respect to a compatibly defined unitary
structure. In this context we frequently use the unitary
representation of the real Heisenberg group $\exp( \mathrm{i} W_{0,
\mathbb{R}}) \times \mathrm{U}_1$ on the completion $\mathcal{A}_V$
of the module $\mathfrak{a}(V)$.

Since $\wedge(V^*_1)$ has finite dimension, exponentiating the spinor
representation of $\mathfrak{o}(W_{1,\mathbb{R}})$ causes no
difficulties. This results in the spinor representation of
$\mathrm{Spin}(W_{1,\mathbb{R}})$, a 2:1 covering of the compact
group $\mathrm{SO}(W_{1,\mathbb{R}})$. So in this case one easily
constructs a representation $R_1 : \, \mathrm{Spin} (W_{1,
\mathbb{R}}) \to \mathrm{U} (\mathfrak{a}(V))$ which is compatible
with $\rho_*|_{\mathfrak{o} (W_{1,\mathbb{R}})}$.

Exponentiating the oscillator representation of $\mathfrak{sp}
(W_{0,\mathbb{R}})$ on the infinite-dimensional vector space
$\mathrm{S}(V_0^*)$ requires more effort. In $\S$\ref{metaplectic
rep}, following Howe \cite{H2}, we construct the Shale-Weil-Segal
representation $R' : \, \mathrm{Mp}(W_{0, \mathbb{R}}) \to \mathrm{U}
(\mathcal{A}_V)$ of the metaplectic group $\mathrm{Mp}(W_{0,
\mathbb{R}})$ which is the 2:1 covering group of the real symplectic
group $\mathrm{Sp} (W_{ 0,\mathbb{R}})$. This is compatible with
$\rho_\ast \vert_{\mathfrak{sp} (W_{0, \mathbb{R}})}$. Altogether we
see that the even part of the Lie superalgebra representation
integrates to $G_\mathbb{R} = \mathrm{Spin} (W_{1 ,\mathbb{R}})
\times_{\mathbb{Z}_2} \mathrm{Mp}(W_{0,\mathbb{R}})\,$.

The construction of $R^\prime$ uses a limiting process coming from
the oscillator semigroup $\widetilde{\mathrm{H}}(W_0^s)$, which is
the double covering of the contraction semigroup $\mathrm{H}(W_0^s)
\subset \mathrm{Sp}(W_0)$ and has $\mathrm{Mp}(W_{0,\mathbb{R}})$ in
its boundary. Furthermore, we have $\widetilde{ \mathrm{H}} (W_0^s) =
\mathrm{Mp}(W_{0,\mathbb{R}}) \times M$ where $M$ is an analytic totally real submanifold of maximal dimension which contains a real form of the torus $T_+$ (see $\S$\ref{oscillator semigroup}). The representation $R_0 : \,\widetilde{\mathrm{H}}(W_0^s)\to \mathrm{End}(\mathcal{A}_V)$ constructed in $\S$\ref{sec:semigrouprep} facilitates the definition of the representation $R^\prime$ and of the character $\chi$ in $\S$\ref{sect:4.2} and $\S$\ref{sect:5.1}. It should be underlined that Proposition \ref{holomorphicity} ensures convergence of the superfunction $\chi(h)$, which is defined as a supertrace and exists for all $h \in \widetilde{\mathrm{H}}(W_0^s)$.

On that basis, the key idea of our approach is to exploit the fact
that every Casimir invariant $I \in \mathsf{U}(\mathfrak{g})$ is
exact in the sense that $I = [ \partial , F]$. By a standard
argument, this exactness property implies that every such invariant
$I$ vanishes in the spinor-oscillator representation. This result in
turn implies for our character $\chi$ the differential equations
$D(I) \chi = 0$ where $D(I)$ is the Laplace-Casimir operator
representing $I$. By drawing on Berezin's theory of radial parts, we
derive a system of differential equations which in combination with
certain other properties ultimately determines $\chi$.

In the case of $K = \mathrm{O}_N$ the Lie group associated to the
even part of the real form of the Howe partner $\mathfrak{g}$ is
embedded in a simple way in the full group $G_\mathbb{R}$ described
above. It is itself just a lower-dimensional group of the same form.
In the case of $K = \mathrm{USp}_N$ a sort of reversing procedure
takes places and the analogous real form is $\mathrm{USp}_{2n} \times
\mathrm{SO}_{2n}^\ast$. Nevertheless, the precise data which are used
as input into the series developments, the uniqueness theorem and the
final calculations of $\chi$ are essentially the same in the two
cases. Therefore there is no difficulty handling them simultaneously.

\section{Howe dual pairs in the orthosymplectic Lie superalgebra}
\label{sect:osp-howe}\setcounter{equation}{0}

In this chapter we collect some foundational information from
representation theory. Basic to our work is the orthosymplectic Lie
superalgebra, $\mathfrak{osp}\,$, in its realization as the space
spanned by supersymmetrized terms of degree two in the Clifford-Weyl
algebra. Representing the latter by its fundamental representation on
the spinor-oscillator module, one gets a representation of
$\mathfrak{osp}$ and of all Howe dual pairs inside of $\mathfrak{osp}
\,$. Roots and weights of the relevant representations are described
in detail.

\subsection{Notion of  Lie superalgebra}

A {\it $\mathbb{Z}_2$-grading} of a vector space $V$ over $\mathbb{K}
= \mathbb{R}$ or $\mathbb{C}$ is a decomposition $V = V_0\oplus V_1$
of $V$ into the direct sum of two $\mathbb{K}$-vector spaces $V_0$
and $V_1\,$. The elements in $(V_0 \cup V_1) \setminus\{ 0\}$ are
called {\it homogeneous}. The {\it parity function} $|\, |:(V_0 \cup
V_1) \setminus \{ 0\} \to \mathbb{Z}_2\,$, $v\in V_s \mapsto |v| =
s\,$, assigns to a homogeneous element its parity. We write $V \simeq
\mathbb{K}^{p| q}$ if $\dim_\mathbb{K} V_0 = p$ and $\dim_\mathbb{K}
V_1 = q\,$.

A {\it Lie superalgebra} over $\mathbb{K}$ is a $\mathbb{Z}_2$-graded
$\mathbb{K}$-vector space $\mathfrak{g} = \mathfrak{g}_0 \oplus
\mathfrak{g}_1$ equipped with a bilinear map $[\, , \,]: \,
\mathfrak{g} \times \mathfrak{g} \to \mathfrak{g}$ satisfying
\begin{enumerate}
\item $[\mathfrak{g}_s , \mathfrak{g}_{s^\prime} ] \subset
\mathfrak{g}_{s + s^\prime}\,$, i.e., $|[X,Y]|= |X| + |Y|$ (mod 2)
for homogeneous elements $X,Y$.
\item Skew symmetry: $[X,Y] = -(-1)^{|X||Y|}[Y,X]$ for homogeneous
$X,Y$.
\item Jacobi identity, which means that $\mathrm{ad}(X) = [X, \; ] : \,
\mathfrak{g} \to \mathfrak{g}$ is a (super-)derivation:
\begin{displaymath}
    \mathrm{ad}(X)\, [Y,Z] = [\mathrm{ad}(X)Y , Z] +
    (-1)^{|X||Y|}[Y,\mathrm{ad}(X)Z] \;.
\end{displaymath}
\end{enumerate}
\begin{example}
[$\mathfrak{gl}(V)$] Let $V = V_0 \oplus V_1$ be a
$\mathbb{Z}_2$-graded $\mathbb{K}$-vector space. There is a canonical
$\mathbb{Z}_2$-grading $\mathrm{End}(V) = \mathrm{End}(V)_0 \oplus
\mathrm{End}(V)_1$ induced by the grading of $V$:
\begin{displaymath}
    \mathrm{End}(V)_s := \{X \in \mathrm{End}(V) \mid \forall
    s^\prime \in \mathbb{Z}_2 : \, X(V_{s^\prime}) \subset
    V_{s + s^\prime}\} \;.
\end{displaymath}
The bilinear extension of $[X,Y] := XY - (-1)^{|X||Y|}YX$ for
homogeneous elements $X , Y \in \mathrm{End}(V)$ to a bilinear map
$[\, , \,] : \, \mathrm{End}(V) \times \mathrm{End}(V) \to
\mathrm{End}(V)$ gives $\mathrm{End}(V)$ the structure of a Lie
superalgebra, namely $\mathfrak{gl}(V)$. The Jacobi identity in this
case is a direct consequence of the associativity $(XY)Z = X (YZ)$
and the definition of $[X,Y]$.

In fact, for every $\mathbb{Z}_2$-graded associative algebra
$\mathcal{A}$ the bracket $[\, , \,] : \, \mathcal{A} \times
\mathcal{A} \to \mathcal{A}$ defined by $[X,Y] = XY - (-1)^{|X| |Y|}
YX$ satisfies the Jacobi identity.
\end{example}
\begin{example}[$\mathfrak{osp}(V \oplus V^\ast)$]\label{exa:osp}
Let $V = V_0 \oplus V_1$ be a $\mathbb{Z}_2$-graded
$\mathbb{K}$-vector space and put $W := V \oplus V^*$. The
$\mathbb{Z}_2$-grading of $V$ induces a $\mathbb{Z}_2$-grading $W =
W_0 \oplus W_1$ in the obvious manner: $W_0^{\vphantom{\ast}} =
V_0^{\vphantom{\ast}} \oplus V_0^*$ and $W_1^{\vphantom{\ast}} =
V_1^{\vphantom{\ast}} \oplus V_1^*\,$. Then consider the canonical
alternating bilinear form $A$ on $W_0\,$,
\begin{displaymath}
    A : \,\, W_0 \times W_0 \to \mathbb{K} \;, \quad (v + \varphi
    \, , v' + \varphi') \mapsto \varphi'(v) - \varphi(v') \;,
\end{displaymath}
and the canonical symmetric bilinear form $S$ on $W_1\,$,
\begin{displaymath}
    S : \,\, W_1 \times W_1 \to \mathbb{K} \;, \quad (v + \varphi
    \, , v' + \varphi') \mapsto \varphi'(v) + \varphi(v') \;.
\end{displaymath}
The \emph{orthosymplectic form} of $W$ is the non-degenerate bilinear
form $Q : \, W \times W \to \mathbb{K}$ defined as the orthogonal sum
$Q = A + S:$
\begin{displaymath}
    Q(w_0^{\vphantom{\prime}} + w_1^{\vphantom{\prime}} , w'_0 + w'_1)
    = A(w_0^{\vphantom{\prime}} , w'_0) + S(w_1^{\vphantom{\prime}} ,
    w'_1) \quad (w_s^{\vphantom{\prime}} \, , w_s^\prime \in
    W_s^{\vphantom{\prime}}) \;.
\end{displaymath}
Note the exchange symmetry $Q(w,w^\prime) = -(-1)^{|w| |w^\prime|}
Q(w^\prime, w)$ for $w, w^\prime \in W_0 \cup W_1\,$.

Given $Q\,$, define a complex linear bijection $\tau : \,
\mathrm{End}(W) \to \mathrm{End}(W)$ by the equation
\begin{equation}\label{eq:def-osp}
    Q(\tau(X) w , w') + (-1)^{|X||w|} Q(w, X w^\prime) = 0
\end{equation}
for all $w, w^\prime \in W_0 \cup W_1\,$. It is easy to check that
$\tau$ has the property
\begin{displaymath}
    \tau(XY) = - (-1)^{|X| |Y|} \tau(Y) \tau(X) \;,
\end{displaymath}
which implies that $\tau$ is an involutory automorphism of the Lie
superalgebra $\mathfrak{gl}(W)$ with bracket $[X,Y] = XY - (-1)^{|X|
|Y|} YX\,$. Hence the subspace $\mathfrak{osp}(W) \subset
\mathrm{End}(W)$ of $\tau$-fixed points is closed w.r.t.\ that
bracket; it is called the (complex) {\it orthosymplectic Lie
superalgebra} of $W$.
\end{example}
\begin{example}[Jordan-Heisenberg algebra]\label{exa:JH}
Using the notation of Example \ref{exa:osp}, consider the vector
space $\widetilde{W} := W \oplus \mathbb{K}$ and take it to be
$\mathbb{Z}_2$-graded by $\widetilde{W}_0 = W_0 \oplus \mathbb{K}$
and $\widetilde{W}_1 = W_1\,$. Define a bilinear mapping $[ \, , \,]
: \, \widetilde{W} \times \widetilde{W} \to \widetilde{W}$ by
\begin{displaymath}
    [\mathbb{K}\, ,\widetilde{W}] = [\widetilde{W},\mathbb{K} ] = 0
    \;, \quad [W , W] \subset \mathbb{K} \;, \quad [w , w^\prime]
    = Q(w,w^\prime) \quad (w,w^\prime \in W) \;.
\end{displaymath}
By the basic properties of the orthosymplectic form $Q\,$, the vector
space $\widetilde{W}$ equipped with this bracket is a Lie
superalgebra -- the so-called Jordan-Heisenberg algebra. Note that
$\widetilde{W}$ is two-step nilpotent, i.e., $[\widetilde{W} , [
\widetilde{W} , \widetilde{W} ] ] = 0\,$.
\end{example}

\subsubsection{Supertrace}

Let $V = V_0\oplus V_1$ be a $\mathbb{Z}_2$-graded $\mathbb K$-vector
space, and recall the decomposition $\mathrm{End}(V) = \bigoplus_{
s,\, t} \mathrm{Hom}(V_s \, , V_t)$. For $X \in \mathrm{End}(V)$, we
denote by $X = \sum_{s,\, t} X_{\, t s}$ the corresponding
decomposition of an operator. The {\it supertrace} on $V$ is the
linear function
\begin{displaymath}
    \mathrm{STr} : \,\, \mathrm{End}(V) \to \mathbb{K} \;,
    \quad X \mapsto \mathrm{Tr}\, X_{00} - \mathrm{Tr}\, X_{11} =
    \sum\nolimits_s (-1)^s \mathrm{Tr}\, X_{s s} \;.
\end{displaymath}
(If $\mathrm{dim}\, V = \infty\,$, then usually the domain of
definition of $\mathrm{STr}$ must be restricted.)

An $\mathrm{ad}$-\emph{invariant bilinear form} on a Lie superalgebra
$\mathfrak{g} = \mathfrak{g}_0 \oplus \mathfrak{g}_1$ over
$\mathbb{K}$ is a bilinear mapping $B : \, \mathfrak{g} \times
\mathfrak{g} \to \mathbb{K}$ with the properties
\begin{enumerate}
\item $\mathfrak{g}_0$ and $\mathfrak{g}_1$ are $B$-orthogonal
to each other ;%
\item $B$ is symmetric on $\mathfrak{g}_0$ and skew on
$\mathfrak{g}_1\,$;%
\item $B([X,Y],Z) = B(X,[Y,Z])$ for all $X, Y, Z \in \mathfrak{g}\,$.
\end{enumerate}
We will repeatedly use the following direct consequences of the
definition of $\mathrm{STr}$.
\begin{lemma}\label{lem:STr}
If $\mathfrak{g}$ is a Lie superalgebra in $\mathrm{End}(V)$, the
trace form $B(X,Y) = \mathrm{STr}\, (XY)$ is an $\mathrm{ad}
$-invariant bilinear form. One has $\mathrm{STr} \, [X,Y] = 0\,$.
\end{lemma}
Recalling the setting of Example \ref{exa:osp}, note that the
supertrace for $W = V \oplus V^\ast$ is odd under the
$\mathfrak{gl}$-automorphism $\tau$ fixing $\mathfrak{osp}(W)$, i.e.,
$\mathrm{STr}_W \circ \tau = - \mathrm{STr}_W\,$. It follows that
$\mathrm{STr}_W X = 0$ for any $X \in \mathfrak{osp}(W)$. Moreover,
$\mathrm{STr}_W (X_1 X_2 \cdots X_{2n+1}) = 0$ for any product of an
odd number of $\mathfrak{osp}$-elements.

\subsubsection{Universal enveloping algebra}\label{sect:UEA}

Let $\mathfrak{g}$ be a Lie superalgebra with bracket $[ \, , \, ]$. The
universal enveloping algebra $\mathsf{U}(\mathfrak{g})$ is defined as the
quotient of the tensor algebra $\mathsf{T}(\mathfrak{g}) = \oplus_{n =
0}^\infty \mathsf{T}^n (\mathfrak{g})$ by the two-sided ideal $\mathsf{J}
(\mathfrak{g})$ generated by all combinations
\begin{displaymath}
    X \otimes Y - (-1)^{|X| |Y|} Y \otimes X - [X,Y]
\end{displaymath}
for homogeneous $X , Y \in \mathsf{T}^1( \mathfrak{g}) \equiv
\mathfrak{g}\,$. If $\mathsf{U}_n(\mathfrak{g})$ is the image of
$\mathsf{T}_n(\mathfrak{g}) := \oplus_{k = 0}^n \mathsf{T}^k
(\mathfrak{g})$ under the projection $\mathsf{T} (\mathfrak{g}) \to
\mathsf{U}(\mathfrak{g}) = \mathsf{T}(\mathfrak{g}) / \mathsf{J}
(\mathfrak{g})$, the algebra $\mathsf{U}(\mathfrak{g})$ is filtered
by $\mathsf{U}( \mathfrak{g}) = \cup_{n=0}^\infty \mathsf{U}_n
(\mathfrak{g})$. The $\mathbb{Z}_2$-grading $\mathfrak{g} =
\mathfrak{g}_0 \oplus \mathfrak{g}_1$ gives rise to a
$\mathbb{Z}_2$-grading of $\mathsf{T}(\mathfrak{g})$ by
\begin{displaymath}
    | X_1 \otimes X_2 \otimes \cdots \otimes X_n | =
    \sum\nolimits_{i = 1}^n | X_i | \quad (\text{for homogenous }
    X_i \in \mathfrak{g})\;,
\end{displaymath}
and this in turn induces a canonical $\mathbb{Z}_2$-grading of
$\mathsf{U} (\mathfrak{g})$.

One might imagine introducing various bracket operations on $\mathsf
{T}(\mathfrak{g})$ and/or $\mathsf{U}(\mathfrak{g})$. However, in
view of the canonical $\mathbb{Z}_2$-grading, the natural bracket
operation to use is the \emph{supercommutator}, which is the bilinear
map $\mathsf{T}(\mathfrak{g}) \times \mathsf{T}(\mathfrak{g}) \to
\mathsf{T}(\mathfrak{g})$ defined by $\{ a , b \} := a b -
(-1)^{|a||b|} b a$ for homogeneous elements $a, \, b \in
\mathsf{T}(\mathfrak{g})$. (For the time being, we use a different
symbol $\{ \, , \, \}$ for better distinction from the bracket $[\, ,
\,]$ on $\mathfrak{g}\,$.) Since by the definition of $\mathsf{J}
(\mathfrak{g})$ one has
\begin{displaymath}
    \{ \mathsf{T}(\mathfrak{g}),\mathsf{J}(\mathfrak{g}) \} =
    \{ \mathsf{J}(\mathfrak{g}), \mathsf{T}(\mathfrak{g}) \}
    \subset \mathsf{J}(\mathfrak{g}) \;,
\end{displaymath}
the supercommutator descends to a well-defined map $\{\,,\,\} : \,
\mathsf{U}(\mathfrak{g})\times\mathsf{U}(\mathfrak{g}) \to
\mathsf{U}(\mathfrak{g})$.
\begin{lemma}\label{lem:2.2}
If $\mathfrak{g}$ is a Lie superalgebra, the supercommutator
$\{\,,\,\}$ gives $\mathsf{U}(\mathfrak{g})$ the structure of another
Lie superalgebra in which $\{ \mathsf{U}_n(\mathfrak{g}) ,
\mathsf{U}_{n^\prime}(\mathfrak{g})\} \subset \mathsf{U}_{n +
n^\prime - 1} (\mathfrak{g})$.
\end{lemma}
\begin{proof}
Compatibility with the $\mathbb{Z}_2$-grading, skew symmetry, and
Jacobi identity are properties of $\{ \, , \, \}$ that are immediate
at the level of the tensor algebra $\mathsf{T}(\mathfrak{g})$. They
descend to the corresponding properties at the level of $\mathsf{U}
(\mathfrak{g})$ by the definition of the two-sided ideal $\mathsf{J}
(\mathfrak{g})$. Thus $\mathsf{U}(\mathfrak{g})$ with the bracket $\{
\, , \, \}$ is a Lie superalgebra.

To see that $\{\mathsf{U}_n(\mathfrak{g}),\mathsf{U}_{n^\prime}
(\mathfrak{g})\}$ is contained in $\mathsf{U}_{n + n^\prime - 1}
(\mathfrak{g})$, notice that this property holds true for $n =
n^\prime = 1$ by the defining relations $\mathsf{J}(\mathfrak{g})
\equiv 0$ of $\mathsf{U} (\mathfrak{g})$. Then use the associative
law for $\mathsf{U}(\mathfrak{g})$ to verify the formula
\begin{displaymath}
    \{ a , bc \} = a b c - (-1)^{|a| (|b| + |c|)} b c a
    = \{ a , b \} c + (-1)^{|a| |b|} b \{ a , c \}
\end{displaymath}
for homogeneous $a, b, c \in \mathsf{U}(\mathfrak{g})$. The claim now
follows by induction on the degree of the filtration $\mathsf{U}
(\mathfrak{g}) = \cup_{n = 0}^\infty \mathsf{U}_n (\mathfrak{g})$.
\end{proof}
By definition, the supercommutator of $\mathsf{U}(\mathfrak{g})$ and
the bracket of $\mathfrak{g}$ agree at the linear level: $\{X,Y\}
\equiv [X,Y]$ for $X, Y \in \mathfrak{g}\,$. It is therefore
reasonable to drop the distinction in notation and simply write $[ \,
, \,]$ for both of these product operations. This we now do.

For future use, note the following variant of the preceding formula:
if $Y_1, \ldots, Y_k\, , X$ are any homogeneous elements of
$\mathfrak{g}\,$, then
\begin{equation}\label{eq:brackets}
    [Y_1 \cdots Y_k \, , \, X] = \sum_{i=1}^k (-1)^{|X|\sum_{j = i+1}^k
    |Y_j|}\, Y_1 \cdots Y_{i-1}\, [Y_i\, , X] \, Y_{i+1} \cdots Y_k \;,
\end{equation}
which expresses the supercommutator in $\mathsf{U}(\mathfrak{g})$ by the
bracket in $\mathfrak{g}\,$.

\subsection{Structure of $\mathfrak{osp}(W)$}\label{subsec:osp}

For a $\mathbb{Z}_2$-graded $\mathbb{K}$-vector space $V = V_0 \oplus
V_1$ let $W = V \oplus V^\ast = W_0 \oplus W_1$ as in Example
\ref{exa:osp}. The {\it orthogonal Lie algebra} $\mathfrak{o}(W_1)$
is the Lie algebra of the Lie group $\mathrm{O}(W_1)$ of
$\mathbb{K}$-linear transformations of $W_1$ that leave the
non-degenerate symmetric bilinear form $S$ invariant. This means that
$X \in \mathrm{End}(W_1)$ is in $\mathfrak{o}(W_1)$ if and only if
\begin{displaymath}
    \forall w, w^\prime \in W_1 : \,\, S(Xw,w') + S(w,Xw') = 0 \;.
\end{displaymath}
Similarly, the {\it symplectic Lie algebra} $\mathfrak{sp}(W_0)$ is
the Lie algebra of the automorphism group $\mathrm{Sp}(W_0)$ of $W_0$
equipped with the non-degenerate alternating bilinear form $A:$
\begin{displaymath}
    \mathfrak{sp}(W_0) = \{ X\in\mathrm{End}(W_0) \mid \forall w ,w'
    \in W_0 : \,\, A(Xw ,w') + A(w,Xw') = 0 \} \;.
\end{displaymath}

For the next statement, recall the definition of the orthosymplectic
Lie superalgebra $\mathfrak{osp}(W)$ and the decomposition
$\mathfrak{osp}(W) = \mathfrak{osp}(W)_0 \oplus
\mathfrak{osp}(W)_1\,$.
\begin{lemma}\label{lem:iso-osp1}
As Lie algebras resp.\ vector spaces,
\begin{displaymath}
    \mathfrak{osp}(W)_0 \simeq \mathfrak{o}(W_1)\oplus\mathfrak{sp}(W_0)
    \;, \quad
    \mathfrak{osp}(W)_1 \simeq W_1^{\vphantom{\ast}} \otimes W_0^* \;.
\end{displaymath}
\end{lemma}
\begin{proof}
The first isomorphism follows directly from the definitions. For the
second isomorphism, decompose $X \in \mathfrak{osp}(W)_1$ as $X =
X_{01}+ X_{10}$ where $X_{01} \in \mathrm{Hom}(W_1,W_0)$ and $X_{10}
\in \mathrm{Hom}(W_0\, ,W_1)$. Then
\begin{displaymath}
    \mathfrak{osp}(W)_1 = \{ X_{10} + X_{01} \in \mathrm{End}(W)_1
    \mid \forall w_s \in W_s : \,
    S(X_{10}w_0\, ,w_1) + A(w_0\, ,X_{01}w_1) = 0 \}\;.
\end{displaymath}
Since both $S$ and $A$ are non-degenerate, the component $X_{01}$ is
determined by the component $X_{10}\,$, and one therefore has
$\mathfrak{osp}(W)_1 \simeq \mathrm{Hom}(W_0\, , W_1) \simeq
W_1^{\vphantom{\ast}} \otimes W_0^*\,$.
\end{proof}
We now review how $\mathfrak{sp}(W_0)$ and $\mathfrak{o}(W_1)$
decompose for our case $W_s = V_s^{\vphantom{\ast}} \oplus V_s^{\ast}
\,$. For that purpose, if $U$ is a vector space with dual vector
space $U^\ast$, let $\mathrm{Sym}(U,U^\ast)$ and $\mathrm{Alt}(U,
U^\ast)$ denote the symmetric resp.\ alternating linear maps from $U$
to $U^\ast$.
\begin{lemma}\label{lem:iso-osp2}
As vector spaces,
\begin{align*}
    \mathfrak{o}(W_1) &\simeq \mathrm{End}(V_1) \oplus \mathrm{Alt}
    (V_1^{\vphantom{\ast}}, V_1^\ast) \oplus \mathrm{Alt} (V_1^\ast
    , V_1^{\vphantom{\ast}}) \;, \\ \mathfrak{sp}(W_0) &\simeq
    \mathrm{End}(V_0) \oplus \mathrm{Sym}(V_0^{\vphantom{\ast}}\, ,
    V_0^\ast)\oplus \mathrm{Sym}(V_0^\ast , V_0^{\vphantom{\ast}})\;.
\end{align*}
\end{lemma}
\begin{proof} There is a canonical decomposition
\begin{displaymath}
    \mathrm{End}(W_s) = \mathrm{End}(V_s^{\vphantom{\ast}}) \oplus
    \mathrm{Hom}(V_s^\ast, V_s^{\vphantom{\ast}}) \oplus \mathrm{Hom}
    (V_s^{\vphantom{*}}\,,V_s^*)\oplus \mathrm{End}(V_s^*)
\end{displaymath}
for $s = 0, 1$. Let $s = 1$ and write the corresponding decomposition
of $X \in \mathrm{End}(W_1)$ as
\begin{displaymath}
    X = \mathsf{A} \oplus \mathsf{B} \oplus \mathsf{C} \oplus
    \mathsf{D} \;.
\end{displaymath}
Substituting $w = v + \varphi$ and $w^\prime = v^\prime +
\varphi^\prime$, the defining condition $S(Xw,w^\prime) = - S(w,X
w^\prime)$ for $X \in \mathfrak{o}(W_1)$ then transcribes to
\begin{displaymath}
    \varphi^\prime(\mathsf{A}v) = - (\mathsf{D}\varphi^\prime)(v)
    \;, \quad (\mathsf{C}v)(v^\prime) = - (\mathsf{C}v^\prime)(v)
    \;, \quad \varphi^\prime(\mathsf{B}\varphi) = -
    \varphi(\mathsf{B}\varphi^\prime) \;,
\end{displaymath}
for all $v, v^\prime \in V_1$ and $\varphi, \varphi^\prime \in
V_1^\ast$. Thus $\mathsf{D} = - \mathsf{A}^\mathrm{t}$, and the maps
$\mathsf{B}, \mathsf{C}$ are alternating. This already proves the
statement for the case of $\mathfrak{o}(W_1)$.

The situation for $\mathfrak{sp}(W_0)$ is identical but for a sign
change: the symmetric form $S$ is replaced by the alternating form
$A\,$, and this causes the parity of $\mathsf{B}, \mathsf{C}$ to be
reversed.
\end{proof}

By adding up dimensions, Lemmas \ref{lem:iso-osp1} and
\ref{lem:iso-osp2} entail the following consequence.
\begin{corollary} \label{cor:dimosp} As a $\mathbb{Z}_2$-graded vector
space, $\mathfrak{osp}(V \oplus V^\ast)$ is isomorphic to
$\mathbb{K}^{p|q}$ where $p = d_0 (2d_0 + 1) + d_1 (2d_1 - 1)$, $q =
4 d_0 d_1\,$, and $d_s = \dim V_s \,$.
\end{corollary}

There exists another way of thinking about $\mathfrak{osp}(W)$, which
will play a key role in the sequel. To define it and keep the sign
factors consistent and transparent, we need to be meticulous about
our ordering conventions. Hence, if $v \in V$ is a vector and
$\varphi \in V^\ast$ is a linear function, we write the value of
$\varphi$ on $v$ as
\begin{displaymath}
    \varphi(v) \equiv \langle v , \varphi \rangle \;.
\end{displaymath}
Based on this notational convention, if $V$ is a $\mathbb
{Z}_2$-graded vector space and $X \in \mathrm{End}(V)$ is a homogeneous
operator, we define the \emph{supertranspose} $X^\mathrm{st} \in
\mathrm{End}(V^\ast)$ of $X$ by
\begin{displaymath}
    \langle v , X^\mathrm{st} \varphi \rangle := (-1)^{|X| |v|}
    \langle Xv , \varphi \rangle \quad (v \in V_0 \cup V_1 ,
    \,\, \varphi \in V^\ast) \;.
\end{displaymath}
This definition differs from the usual transpose by a change of sign
in the case when $X$ has a com\-ponent in $\mathrm{Hom}(V_1 , V_0)$.
From it, it follows directly that the negative supertranspose
$\mathfrak{gl}(V) \to \mathfrak{gl}(V^\ast)$, $X \mapsto
-X^\mathrm{st}$ is an isomorphism of Lie superalgebras:
\begin{displaymath}
    - [X,Y]^\mathrm{st} = [-X^\mathrm{st} , -Y^\mathrm{st}] \;.
\end{displaymath}

The modified notion of transpose goes hand in hand with a modified
notion of what it means for an operator in $\mathrm{Hom}(V,V^\ast)$
or $\mathrm{Hom}(V^\ast , V)$ to be symmetric. Thus, define the
subspace $\mathrm{Sym}(V^\ast,V) \subset \mathrm{Hom}(V^\ast,V)$ to
consist of the elements, say $\mathsf{B}$, which are symmetric in the
$\mathbb{Z}_2$-graded sense:
\begin{equation}\label{eq:sym-B}
    \forall \, \varphi, \varphi^\prime \in V_0^\ast \cup V_1^\ast :
    \quad \langle \mathsf{B} \varphi , \varphi^\prime \rangle = \langle
    \mathsf{B} \varphi^\prime , \varphi \rangle \, (-1)^{|\varphi|
    |\varphi^\prime|} \;.
\end{equation}
By the same principle, define $\mathrm{Sym}(V,V^\ast) \subset
\mathrm{Hom}(V,V^\ast)$ as the set of solutions $\mathsf{C}$ of
\begin{equation} \label{eq:sym-C}
    \forall \, v, v^\prime \in V_0 \cup V_1 : \quad \langle v ,
    \mathsf{C}v^\prime \rangle= \langle v^\prime, \mathsf{C} v\rangle
    \, (-1)^{|v| |v^\prime| + |v| + |v^\prime|} \;.
\end{equation}
To make the connection with the decomposition of Lemma
\ref{lem:iso-osp1} and \ref{lem:iso-osp2}, notice that
\begin{displaymath}
    \mathrm{Sym}(V,V^\ast)\cap \mathrm{Hom}(V_s^{\vphantom{\ast}}
    \,, V_s^\ast) = \left\{\begin{array}{ll}\mathrm{Sym}
    (V_0^{\vphantom{\ast}}\,, V_0^\ast) &s = 0 \;,\\
    \mathrm{Alt}(V_1^{\vphantom{\ast}} , V_1^\ast) &s = 1 \;,
    \end{array} \right.
\end{displaymath}
and similar for the corresponding intersections involving
$\mathrm{Sym}(V^\ast,V)$.

Next, expressing the orthosymplectic form $Q$ of $W = V \oplus
V^\ast$ as
\begin{displaymath}
    Q(v + \varphi , v^\prime + \varphi^\prime) = \langle v ,
    \varphi^\prime \rangle - (-1)^{|v^\prime| |\varphi|}
    \langle v^\prime , \varphi \rangle \;,
\end{displaymath}
and writing out the conditions resulting from $Q(Xw,w^\prime) + (-1)^
{|X| |w|} Q(w,Xw^\prime) = 0$ for the case of $X \equiv \mathsf{B}
\in \mathrm{Hom}(V^\ast,V)$ and $X \equiv \mathsf{C} \in
\mathrm{Hom}(V,V^\ast)$, one sees that
\begin{displaymath}
    \mathfrak{osp}(W) \cap \mathrm{Hom}(V,V^\ast) =
    \mathrm{Sym}(V,V^\ast)\;, \quad \mathfrak{osp}(W) \cap
    \mathrm{Hom}(V^\ast,V) = \mathrm{Sym}(V^\ast,V) \;.
\end{displaymath}
This situation is summarized in the next statement.
\begin{lemma}\label{lem:osp-iso3}
The orthosymplectic Lie superalgebra of $\, W = V \oplus V^\ast$
decomposes as
\begin{displaymath}
    \mathfrak{osp}(W) = \mathfrak{g}^{(-2)}\oplus
    \mathfrak{g}^{(0)}\oplus\mathfrak{g}^{(+2)}\;,
\end{displaymath}
where $\mathfrak{g}^{(+2)} := \mathrm{Sym}(V,V^\ast)$, and
$\mathfrak{g}^{(-2)} := \mathrm{Sym}(V^\ast,V)$, and
\begin{displaymath}
    \mathfrak{g}^{(0)} := (\mathrm{End}(V)\oplus
    \mathrm{End}(V^\ast))\cap \mathfrak{osp}(W)\;.
\end{displaymath}
\end{lemma}
The decomposition of Lemma \ref{lem:osp-iso3} can be regarded as a
$\mathbb{Z}$-grading of $\mathfrak{osp}(W)$. By the `block' structure
inherited from $W = V \oplus V^\ast$, this decomposition is
compatible with the bracket $[\, , \,]:$
\begin{displaymath}
    [ \mathfrak{g}^{(m)} , \, \mathfrak{g}^{(m^\prime)} ] \subset
    \mathfrak{g}^{(m + m^\prime)} \;,
\end{displaymath}
where $\mathfrak{g}^{(m + m^\prime)} \equiv 0$ if $m + m^\prime
\notin \{ \pm 2, 0\}$. It follows that each of the three subspaces
$\mathfrak{g}^{( +2)}$, $\mathfrak{g}^{(-2)}$, and $\mathfrak{g}^{
(0)}$ is a Lie superalgebra, the first two with vanishing bracket.
\begin{lemma}\label{lem:iso-osp4}
The embedding $\mathrm{End}(V)\to \mathrm{End}(V) \oplus \mathrm{End}
(V^\ast)$ by $\mathsf{A} \mapsto \mathsf{A} \oplus (-\mathsf{A}^
\mathrm{st})$ projected to $\mathfrak{osp}(W)$ is an isomorphism of
Lie superalgebras $\mathfrak{gl}(V) \to \mathfrak{g}^{(0)}$.
\end{lemma}
\begin{proof}
Since the negative supertranspose $\mathsf{A} \mapsto - \mathsf{A}
^\mathrm{st}$ is a homomorphism of Lie superalgebras, so is our
embedding $\mathsf{A} \mapsto \mathsf{A} \oplus (- \mathsf{A}
^\mathrm{st})$. This map is clearly injective. To see that it is
surjective, consider any homogeneous $X = \mathsf{A} \oplus
\mathsf{D} \in \mathrm{End}(V) \oplus \mathrm{End}(V^\ast)$ viewed as
an operator in $\mathrm{End}(W)$. The condition for $X$ to be in
$\mathfrak{osp}(W)$ is (\ref{eq:def-osp}). To get a non-trivial
condition, choose $(w,w^\prime) = (v,\varphi)$ or $(w,w^\prime) =
(\varphi,v)$. The first choice gives
\begin{displaymath}
    Q(Xv,\varphi) + (-1)^{|X||v|} Q(v,X\varphi) = \langle
    \mathsf{A} v , \varphi \rangle + (-1)^{|\mathsf{A}||v|}
    \langle v , \mathsf{D}\varphi \rangle = 0 \;.
\end{displaymath}
Valid for all $v \in V_0 \cup V_1$ and $\varphi \in V^\ast$, this
implies that $\mathsf{D} = - \mathsf{A}^\mathrm{st}$. The second
choice leads to the same conclusion. Thus $X = \mathsf{A} \oplus
\mathsf{D}$ is in $\mathfrak{osp}(W)$ if and only if $\mathsf{D} = -
\mathsf{A}^\mathrm{st}$.
\end{proof}
In the following subsections we will often write $\mathfrak{osp}(W)
\equiv \mathfrak{osp}$ for short.

\subsubsection{Roots and root spaces}\label{sect:osp-roots}

A Cartan subalgebra of a Lie algebra $\mathfrak{g}_0$ is a maximal
commutative subalgebra $\mathfrak{h} \subset \mathfrak{g}_0$ such
that $\mathfrak{g}_0$ (or its complexification if $\mathfrak{g}_0$ is
a real Lie algebra) has a basis consisting of eigenvectors of
$\mathrm{ad}(H)$ for all $H \in \mathfrak{h} \,$. Recall that
$|[X,Y]|=|X|+|Y|$ for homogeneous elements $X,Y$ of a Lie
superalgebra $\mathfrak{g}\,$. From the vantage point of decomposing
$\mathfrak{g}$ by eigenvectors or root spaces, it is therefore
reasonable to call a Cartan subalgebra of $\mathfrak{g}_0$ a Cartan
subalgebra of $\mathfrak{g}\,$. We will see that $X \in
\mathfrak{osp}_1$ and $[X,H] = 0$ for all $H \in \mathfrak{h} \subset
\mathfrak{osp}_0$ imply $X = 0\,$, i.e., there exists no commutative
subalgebra of $\mathfrak{osp}$ that properly contains a Cartan
subalgebra. Lie superalgebras with this property are called of type I
in \cite{B}.

Let us determine a Cartan subalgebra and the corresponding root space
decomposition of $\mathfrak{osp}\,$. For $s, t = 0, 1$ choose bases
$\{ e_{s, 1}\, , \ldots, \, e_{s,\,d_s} \}$ of $V_s$ and associated
dual bases $\{ f_{t,1} , \ldots, f_{t,\, d_t} \}$ of $V_t^\ast\,$.
Then for $j = 1, \ldots, d_s$ and $k = 1, \ldots, d_t$ define
rank-one operators $E_{s, j \, ; \, t ,\, k}$ by the equation
$E_{s,j\, ; \, t,\, k} (e_{u,l}) = e_{s,j}\, \delta_{t,u} \,
\delta_{k,l}$ for all $u = 0, 1$ and $l = 1, \ldots, d_u \,$. These
form a basis of $\mathrm{End}(V)$, and by Lemma \ref{lem:iso-osp4}
the operators
\begin{displaymath}
    X_{s j, \, t k}^{(0)} := E_{s , j\,; \, t ,\, k} \oplus
    (- E_{s,j\,;\,t,\,k})^\mathrm{st}
\end{displaymath}
form a basis of $\mathfrak{g}^{(0)}$. Similarly, let bases of
$\mathrm{Hom}(V^\ast,V)$ and $\mathrm{Hom}(V,V^\ast)$ be defined by
\begin{displaymath}
    F_{s,j\, ; \,t,\,k} (f_{u,l}) = e_{s,j} \, \delta_{t,u} \,
    \delta_{k,l} \;, \qquad \tilde{F}_{s,j\,;\,t,\,k} (e_{u,l})
    = f_{s,j} \,\delta_{t,u}\, \delta_{k,l} \;,
\end{displaymath}
for index pairs in the appropriate range. Then by Lemma
\ref{lem:osp-iso3} and equations (\ref{eq:sym-B}, \ref{eq:sym-C}) the
subalgebras $\mathfrak{g}^{(-2)}$ and $\mathfrak{g}^{(2)}$ are
generated by the sets of operators
\begin{eqnarray*}
    &&X_{s j, \, t k}^{(-2)} := F_{s,j\,;\,t,\,k} + F_{t,\,k\,;\,s,j}
    \, (-1)^{|s| |t|} \;, \\ &&X_{s j, \, t k}^{(2)} :=
    \tilde{F}_{s,j\,;\,t,\,k} + \tilde{F}_{t,\,k\,;\,s,j} \,
    (-1)^{|s| |t| + |s| + |t|}\;.
\end{eqnarray*}

Since $\mathfrak{osp}_0 \simeq \mathfrak{o}(W_1) \oplus
\mathfrak{sp}(W_0)$, a Cartan subalgebra of $\mathfrak{osp}$ is the
direct sum of a Cartan subalgebra of $\mathfrak{o}(W_1)$ and a Cartan
subalgebra of $\mathfrak{sp}(W_0)$. Letting $\mathfrak{h}$ be the
span of the diagonal operators
\begin{displaymath}
    H_{s j} := X_{s j ,\, s j}^{(0)}\quad (s=0,1;\, j= 1,\ldots,d_s)\;,
\end{displaymath}
one has that $\mathfrak{h}$ is a Cartan subalgebra of $\mathfrak{osp}
\,$. Indeed, if $\{ \vartheta_{s j} \}$ denotes the basis of
$\mathfrak{h}^\ast$ dual to $\{ H_{s j} \}\,$, inspection of the
adjoint action of $\mathfrak{h}$ on $\mathfrak{osp}$ gives the
following result.
\begin{lemma}\label{lem:rootsosp}
The operators $X_{s j\, , \, t k}^{(m)}$ are eigenvectors of
$\mathrm{ad}(H)$ for all $H \in \mathfrak{h}:$
\begin{displaymath}
    [H\, ,X^{(m)}_{s j\, ,\, t k}] = \begin{cases}
    (\vartheta_{s j} - \vartheta_{t k})(H) \, X^{(m)}_{s j\, ,\, t k}
    & m = 0 \;, \\ (\vartheta_{s j} + \vartheta_{t k})(H) \,
    X^{(m)}_{s j\, , \, t k} & m = -2 \;, \\ (- \vartheta_{s j} -
    \vartheta_{t k})(H)\, X^{(m)}_{s j\,, \,t k} & m = 2\;.\end{cases}
\end{displaymath}
\end{lemma}
A root of a Lie superalgebra $\mathfrak{g}$ is called {\it even} if
its root space is in $\mathfrak{g}_0\,$, it is called {\it odd} if
its root space is in $\mathfrak{g}_1\,$. We denote by $\Delta_0$ and
$\Delta_1$ the set of even roots and the set of odd roots,
respectively. For $\mathfrak{g} = \mathfrak{osp}$ we have
\begin{displaymath}
    \Delta_0 =\{ \pm \vartheta_{1j} \pm \vartheta_{1k}\, , \,
    \pm \vartheta_{0 j} \pm \vartheta_{0 l} \mid j < k , \, j\leq l\}
    , \quad \Delta_1 = \{ \pm \vartheta_{1 j} \pm \vartheta_{0 k} \}.
\end{displaymath}

\subsubsection{Casimir elements}\label{sect:osp-cas}

As before, let $\mathfrak{g} = \mathfrak{g}_0 \oplus \mathfrak{g}_1$
be a Lie superalgebra, and let $\mathsf{U}(\mathfrak{g}) = \cup_{n =
0}^\infty \mathsf{U}_n (\mathfrak{g})$ be its universal enveloping
algebra. Denote the symmetric algebra of $\mathfrak{g}_0$ by
$\mathrm{S}(\mathfrak{g}_0)$ and the exterior algebra of
$\mathfrak{g}_1$ by $\wedge(\mathfrak{g}_1)$. The
Poincar\'e-Birkhoff-Witt theorem for Lie superalgebras states that
for each $n$ there is a bijective correspondence
\begin{displaymath}
    \mathsf{U}_n (\mathfrak{g}) / \mathsf{U}_{n-1} (\mathfrak{g})
    \stackrel{\sim}{\to} \sum\nolimits_{k +l = n}
    \wedge^k( \mathfrak{g}_1 ) \otimes \mathrm{S}^l (\mathfrak{g}_0) \;.
\end{displaymath}
The collection of inverse maps lift to a vector-space isomorphism,
\begin{displaymath}%\label{eq:susy-map}
    \wedge(\mathfrak{g}_1) \otimes \mathrm{S}(\mathfrak{g}_0)
    \stackrel{\sim}{\to} \mathsf{U}(\mathfrak{g}) \;,
\end{displaymath}
called the super-symmetrization mapping. In other words, given a
homogeneous basis $\{ e_1, \ldots, e_d \}$ of $\mathfrak{g} \,$, each
element $x \in \mathsf{U}(\mathfrak{g})$ can be uniquely represented
in the form $x = \sum_n \sum_{i_1, \ldots, \, i_n} x_{i_1, \ldots, \,
i_n}\, e_{i_1} \cdots \, e_{i_n}$ with super-symmetrized
coefficients, i.e.,
\begin{displaymath}
    x_{i_1,\ldots,\, i_l,\, i_{l+1}, \ldots, \, i_n} = (-1)^{|e_{i_l}|
    |e_{i_{l+1}}|} x_{i_1,\, \ldots,\, i_{l+1},\, i_l,\, \ldots,\, i_n}
    \quad (1 \le l < n) \;.
\end{displaymath}
The isomorphism $\wedge(\mathfrak{g}_1) \otimes \mathrm{S}
(\mathfrak{g}_0) \simeq \mathsf{U}(\mathfrak{g})$ gives $\mathsf{U}
(\mathfrak{g})$ a $\mathbb{Z}$-grading (by the degree $n$).

Now recall that $\mathsf{U}(\mathfrak{g})$ comes with a canonical
bracket operation, the supercommutator $[\, , \,] : \, \mathsf{U}
(\mathfrak{g}) \times \mathsf{U}(\mathfrak{g}) \to \mathsf{U}
(\mathfrak{g})$. An element $X \in \mathsf{U}(\mathfrak{g})$ is said
to lie in the center of $\mathsf{U} (\mathfrak{g})$, and is called a
\emph{Casimir element}, iff $[X,Y]= 0$ for all $Y \in \mathsf{U}
(\mathfrak{g})$. By the formula (\ref{eq:brackets}), a necessary and
sufficient condition for that is $[X,Y] = 0$ for all $Y \in
\mathfrak{g}\,$.

In the case of $\mathfrak{g} = \mathfrak{osp}\,$, for every $\ell \in
\mathbb{N}$ there is a Casimir element $I_\ell$ of degree $2\ell$,
which is constructed as follows. Consider the bilinear form $B : \,
\mathfrak{osp} \times \mathfrak{osp} \to \mathbb{K}$ given by the
supertrace (in some representation), $B(X,Y) : = \mathrm{STr}\,
(XY)$. Recall that this form is $\mathrm{ad}$-invariant, which is to
say that $B([X,Y],Z) = B(X,[Y,Z])$ for all $X, Y, Z \in
\mathfrak{g}\,$.

Taking the supertrace in the fundamental representation of
$\mathfrak{osp} \,$, the form $B$ is non-degenerate, and therefore,
if $e_1,\, \ldots, \, e_d$ is a homogeneous basis of
$\mathfrak{osp}\,$, there is another homogeneous basis
$\widetilde{e}_1 , \, \ldots, \, \widetilde{e}_d$ of $\mathfrak{osp}$
so that $B(\widetilde{e}_i \, , \, e_j) = \delta_{ij}\,$. Note
$|\widetilde{e}_i | = |e_i|$ and put
\begin{equation}\label{eq:CasimirElement}
    I_\ell := \sum_{i_1,\,\ldots,\,i_{2\ell}=1}^d \widetilde{e}_
    {i_1}\cdots\, \widetilde{e}_{i_{2\ell}}\,\mathrm{STr}\,
    (e_{i_{2\ell}}\cdots\, e_{i_1})\in \mathsf{U}(\mathfrak{osp}).
\end{equation}
(Notice that, in view of the remark following Lemma \ref{lem:STr},
there is no point in making the same construction with an odd number
of factors.)
\begin{lemma}\label{lem:IlCasimir}
For all $\ell \in \mathbb N$ the element $I_\ell$ is Casimir, and
$|I_\ell| = 0$.
\end{lemma}
\begin{proof}
By specializing the formula (\ref{eq:brackets}) to the present case,
\begin{displaymath}
    [I_\ell\, , \,X] = \sum \sum_{k=1}^{2\ell} (-1)^{|X|
    (|\widetilde{e}_{i_{k+1}}| + \ldots +|\widetilde{e}_{i_{2\ell}}|)}
    \widetilde{e}_{i_1} \cdots \, \widetilde{e}_{i_{k-1}} \, [
    \widetilde{e}_{i_k} , X] \, \widetilde{e}_{i_{k+1}} \cdots \,
    \widetilde{e}_{i_{2\ell}}\,\mathrm{STr}\,(e_{i_{2\ell}}\cdots\,
    e_{i_1})\;.
\end{displaymath}
Now if $[\widetilde{e}_i \, , X] = \sum_j X_{ij}\, \widetilde{e}_j$
then from $\mathrm{ad}$-invariance, $B([\widetilde{e}_i \, , X], e_j)
= X_{ij} = B(\widetilde{e}_i \, , [X,e_j])$, one has $[X,e_j] =
\sum_i e_i\, X_{ij}\,$. Using this relation to transfer the
$\mathrm{ad}(X)$-action from $\widetilde{e}_{i_k}$ to $e_{i_k}\,$,
and reading the formula (\ref{eq:brackets}) backwards, one obtains
\begin{displaymath}
    [I_\ell\,,\,X] = \sum \widetilde{e}_{i_1}\cdots \,\widetilde{e}_{
    i_{2\ell}}\, \mathrm{STr}\,([X,e_{i_{2\ell}}\cdots\, e_{i_1}]) \;.
\end{displaymath}
Since the supertrace of any bracket vanishes, one concludes that
$[I_\ell\, , X ] = 0$.

The other statement, $|I_\ell| = 0\,$, follows from $|\widetilde
{e}_i | = |e_i|$, the additivity of the $\mathbb{Z}_2$-degree and the
fact that $\mathrm{STr}\,(a) = 0$ for any odd element $a \in
\mathsf{U}(\mathfrak{g})$.
\end{proof}
We now describe a useful property enjoyed by the Casimir elements
$I_\ell$ of $\mathfrak{osp}(V \oplus V^\ast)$ in the special case of
isomorphic components $V_0 \simeq V_1\,$. Recalling the notation of
$\S$\ref{sect:osp-roots}, let $\partial := \sum_j X^{(0)}_{0j,1j}$
and $\widetilde{\partial} := - \sum_j X_{1j ,0j}^{(0)}\,$. These are
odd elements of $\mathfrak{osp}\,$.
%
% (The reason for using the symbols $\partial, \widetilde{\partial}$ will % become clear later).
%
Notice that the bracket ${C} := [\partial , \widetilde{\partial}] = -
\sum_{s,j} H_{s j}$ is in the Cartan algebra of $\mathfrak{osp}\,$.
From $[\partial , \partial] = 2 \, \partial^2 = 0$ and the Jacobi
identity one infers that
\begin{displaymath}
    [\partial,{C}] = [[\partial,\partial],\widetilde{\partial}]
    - [\partial , [\partial , \widetilde{\partial}]]
    = - [\partial , {C}] = 0 \;.
\end{displaymath}
By the same argument, $[\widetilde{\partial},{C}] = 0\,$. One
also sees that ${C}^2 = \mathrm{Id}\,$.

Now define $F_\ell$ to be the following odd element of
$\mathsf{U}(\mathfrak{osp})$:
\begin{displaymath}
    F_\ell = - \sum_{i_1,\,\ldots,\,i_{2\ell}=1}^d \widetilde{e}_{i_1}
    \cdots \, \widetilde{e}_{i_{2\ell}}\, \mathrm{STr}\,(e_{i_{2\ell}}
    \cdots \, e_{i_1} \widetilde{\partial} {C}) \;.
\end{displaymath}
\begin{lemma}\label{lem:Cas-exact}
Let $\mathfrak{osp}(V \oplus V^\ast)$ be the orthosymplectic Lie
superalgebra for a $\mathbb{Z}_2$-graded vector space $V$ with
isomorphic components $V_0 \simeq V_1\,$. Then for all $\ell \in
\mathbb{N}$ the Casimir element $I_\ell$ is expressible as a bracket:
$I_\ell = [\partial , F_\ell ]\,$.
\end{lemma}
\begin{proof}
By the same argument as in the proof of Lemma \ref{lem:IlCasimir},
\begin{displaymath}
    [\partial , F_\ell ] = - \sum \widetilde{e}_{i_1}\cdots \,
    \widetilde{e}_{i_{2\ell}}\,\mathrm{STr}\,([\partial ,e_{i_{2\ell}}
    \cdots e_{i_1}] \, \widetilde{\partial} {C}) \;.
\end{displaymath}
Using the relations $[\partial , {C}] = 0$ and ${C}^2 =
\mathrm{Id}\,$, one has for any $a\in \mathsf{U} (\mathfrak{osp})$ that
\begin{displaymath}
    - \mathrm{STr}\, ([\partial, \, a]\, \widetilde{\partial} {C}) =
    \mathrm{STr}\, (\widetilde{\partial} {C} \, [\partial , \, a]) =
    \mathrm{STr}\, ([\widetilde{\partial} , \partial] {C} \, a)
    = \mathrm{STr}\,({C}^2 a) = \mathrm{STr}\, (a) \;,
\end{displaymath}
where the second equality sign is from $\mathrm{STr}\, (c,[b,a]) =
\mathrm{STr}\, ([c,b]\, a)$. The statement of the lemma now follows
on setting $a = e_{i_{2\ell}} \cdots\, e_{i_1}\,$.
\end{proof}
As we shall see in $\S$\ref{sect:CasEigenvalues}, Lemma \ref{lem:Cas-exact} has the drastic consequence that all $\mathfrak{osp}$-Casimir elements $I_\ell$ are zero in a certain class of representations of $\mathfrak{osp}(V \oplus V^\ast)$ for $V_0 \simeq V_1\,$.

\subsection{Howe pairs in $\mathfrak{osp}(W)$}\label{sect:howe-pairs}

In the present context, a pair $(\mathfrak{h},\mathfrak{h}')$ of
subalgebras $\mathfrak{h}, \mathfrak{h}' \subset \mathfrak{g}$ of a
Lie superalgebra $\mathfrak{g}$ is called a \emph{dual pair} whenever
$\mathfrak{h}'$ is the centralizer of $\mathfrak{h}$ in
$\mathfrak{g}$ and vice versa. In this subsection, let $\mathbb{K} =
\mathbb{C}\,$.

Given a $\mathbb{Z}_2$-graded complex vector space $U = U_0 \oplus
U_1$ we let $V := U \otimes \mathbb{C}^N$, where $\mathbb{C}^N$ is
equipped with the standard representation of $\mathrm{GL}
(\mathbb{C}^N)$, $\mathrm{O}(\mathbb{C}^N)\,$, or $\mathrm{Sp}
(\mathbb{C}^N)\,$, as the case may be. As a result, the Lie algebra
$\mathfrak{k}$ of whichever group is represented on $\mathbb{C}^N$ is
embedded in $\mathfrak{osp}(V \oplus V^\ast)$.  We will now describe
the dual pairs $(\mathfrak{h},\mathfrak{k})$ in $\mathfrak{osp} (W)$
for $W = V \oplus V^*$. These are known as dual pairs in the sense of
R.\ Howe.

Let us begin by recalling that for any representation $\rho :\, K \to
\mathrm{GL}(E)$ of a group $K$ on a vector space $E\,$, the dual
representation $\rho^\ast : \, K \to \mathrm{GL}(E^\ast)$ on the
linear forms on $E$ is given by $(\rho(k)\varphi)(x) =
\varphi(\rho(k)^{-1}x)$. By this token, every representation $\rho
:\, K \to \mathrm{GL}(\mathbb{C}^N)$ induces a representation $\rho_W
= (\mathrm{Id} \otimes \rho) \times (\mathrm{Id} \otimes \rho^\ast)$
of $K$ on $W = V \oplus V^\ast$.
\begin{lemma}
Let $\rho: \, K \to \mathrm{GL}(\mathbb{C}^N)$ be any representation
of a Lie group $K\,$. If $V = U \otimes \mathbb{C}^N$ for a
$\mathbb{Z}_2$-graded complex vector space $U = U_0 \oplus U_1\,$,
the induced repre\-sentation ${\rho_W}_\ast (\mathfrak{k})$ of the
Lie algebra $\mathfrak{k}$ of $K$ on $W = V \oplus V^\ast$ is a
subalgebra of $\mathfrak{osp}(W)_0\,$.
\end{lemma}
\begin{proof}
The $K$-action on $\mathbb{C}^N \otimes (\mathbb{C}^N)^\ast$ by $z
\otimes \zeta \mapsto \rho(k)z \otimes \rho^\ast(k)\zeta$ preserves
the canonical pairing $z \otimes \zeta \mapsto \zeta(z)$ between
$\mathbb{C}^N$ and $(\mathbb{C}^N)^\ast$. Consequently, the
$K$-action on $V \otimes V^\ast$ by $(\mathrm{Id} \otimes \rho)
\otimes (\mathrm{Id} \otimes \rho^\ast)$ preserves the canonical
pairing $V \otimes V^\ast \to \mathbb{C}\,$. Since the
orthosymplectic form $Q : \, W \times W \to \mathbb{C}$ uses nothing
but that pairing, it follows that
\begin{displaymath}
    Q( \rho_W(k) w \, , \, \rho_W(k) w^\prime) = Q(w,w^\prime) \quad
    (\text{for all } w,w^\prime \in W)\;.
\end{displaymath}
Passing to the Lie algebra level one obtains ${\rho_W}_\ast(
\mathfrak{k}) \subset \mathfrak{osp}(W)$. The operator $\rho_W(k)$
preserves the $\mathbb{Z}_2$-grading of $W$; therefore one actually
has ${\rho_W}_\ast( \mathfrak{k}) \subset \mathfrak{osp}(W)_0\,$.
\end{proof}
Let us now assume that the complex Lie group $K$ is defined by a
non-degenerate bilinear form $B :\, \mathbb{C}^N \times \mathbb{C}^N
\to \mathbb{C}$ in the sense that
\begin{displaymath}
    K = \{ k \in \mathrm{GL}(\mathbb{C}^N) \mid \forall z\, , z^\prime
    \in \mathbb{C}^N : \, B(k z \, ,\, k z^\prime) = B(z\,,z^\prime)\}\;.
\end{displaymath}
We then have a canonical isomorphism $\psi : \, \mathbb{C}^N \to
(\mathbb{C}^N)^\ast$ by $z \mapsto B(z\, ,\,)$, and an isomorphism
$\Psi : \, (U \oplus U^\ast) \otimes \mathbb{C}^N \to W$ by $(u +
\varphi) \otimes z \mapsto u \otimes z + \varphi \otimes \psi(z)$.
\begin{lemma}\label{lem:2.11}
$\rho_W(k) = \Psi \circ (\mathrm{Id} \otimes k) \circ \Psi^{-1}$ for
all $k \in K\,$.
\end{lemma}
\begin{proof}
If $u \in U$, $\varphi \in U^\ast$, and $z \in \mathbb{C}^N$, then by
the definition of $\Psi$ and $\rho_W(k)$,
\begin{displaymath}
    \rho_W(k) \Psi((u+\varphi) \otimes z) = u \otimes k z +
    \varphi \otimes \psi(z) k^{-1} \;.
\end{displaymath}
Since $B$ is $K$-invariant, one has $\psi(z) k^{-1} = \psi(kz)$, and
therefore
\begin{displaymath}
    u \otimes k z + \varphi \otimes \psi(z) k^{-1}
    = u \otimes k z + \varphi \otimes \psi(k z)
    = \Psi ((\mathrm{Id} \otimes k)((u+\varphi)\otimes z)) \;.
\end{displaymath}
Thus $\rho_W(k) \circ \Psi = \Psi \circ (\mathrm{Id} \otimes k)$.
\end{proof}
Let us now examine what happens to the orthosymplectic form $Q$ on
$W$ when it is pulled back by the isomorphism $\Psi$ to a bilinear
form $\Psi^*Q$ on $(U\oplus U^*) \otimes \mathbb{C}^N:$
\begin{displaymath}
    \Psi^*Q((u + \varphi) \otimes z\, , (u' + \varphi') \otimes z')
    = \varphi'(u)\, \psi(z')(z) - (-1)^{|u'||\varphi|} \varphi(u')
    \, \psi(z)(z') \;.
\end{displaymath}
By definition, $\psi(z)(z^\prime) = B(z\, ,z^\prime)$, and writing
$B(z\, , z^\prime) = (-1)^\delta B(z^\prime,z)$ where $\delta = 0$ if
$B$ is symmetric and $\delta = 1$ if $B$ is alternating, we obtain
\begin{equation}\label{eq:pull-back}
    \Psi^*Q((u + \varphi) \otimes z \, , (u^\prime + \varphi^\prime)
    \otimes z^\prime) = (\varphi^\prime (u) - (-1)^{|u'||\varphi| +
    \delta} \varphi(u')) B(z',z) \;.
\end{equation}
In view of this, let $\widetilde{U}$ denote the vector space $U = U_0
\oplus U_1$ with the twisted $\mathbb{Z}_2$-grading, i.e.\
$\widetilde{U}_s := U_{s+1}$ ($s \in \mathbb{Z}_2$). Moreover, notice
that $\Psi$ determines an embedding
\begin{displaymath}
    \mathrm{End}(U \oplus U^\ast) \otimes \mathrm{End}(\mathbb{C}^N)
    \to \mathrm{End}(W) \;, \quad
    X \otimes k \mapsto \Psi \circ (X \otimes k) \circ \Psi^{-1} \;,
\end{displaymath}
whose restriction to $\mathrm{End}(U \oplus U^*) \otimes \{
\mathrm{Id} \} \to \mathrm{End}(W)$ is an injective homomorphism.

In the following we often write $\mathrm{O}(\mathbb{C}^N) \equiv
\mathrm{O}_N$ and $\mathrm{Sp}(\mathbb{C}^N) \equiv \mathrm{Sp}_N$
for short.
\begin{corollary}\label{cor:2.2}
For $K = \mathrm{O}_N$ and $K = \mathrm{Sp}_N\,$, the map $X \mapsto
\Psi \circ (X \otimes \mathrm{Id}) \circ \Psi^{-1}$ defines a Lie
superalgebra embedding into $\mathfrak{osp}(W)$ of $\,
\mathfrak{osp}(U \oplus U^*)$ resp.\ $\mathfrak{osp} (\widetilde{U}
\oplus \widetilde{U}^\ast)$.
\end{corollary}
\begin{proof}
For $K = \mathrm{O}_N$ the bilinear form $B$ of $\mathbb{C}^N$ is
symmetric and the bilinear form $Q$ of $W$ pulls back -- see
(\ref{eq:pull-back}) -- to the standard orthosymplectic form of $U
\oplus U^\ast$.

For $K = \mathrm{Sp}_N\,$, the form $B$ is alternating. Its pullback,
the orthosymplectic form of $U \oplus U^\ast$ twisted by the sign
factor $(-1)^\delta$, is restored to standard form by switching to
the $\mathbb{Z}_2$-graded vector space $\widetilde{U} \oplus
\widetilde{U}^\ast$ with the twisted $\mathbb{Z}_2$-grading.
\end{proof}
To go further, we need a statement concerning $\mathrm{Hom}_G(V_1
,V_2)$, the space of $G$-equivariant homomorphisms between two
modules $V_1$ and $V_2$ for a group $G\,$.
\begin{lemma}
\label{lem:HomG} Let $X_1\, , X_2 \, , Y_1\, , Y_2$ be
finite-dimensional vector spaces all of which are representation
spaces for a group $G\,$. If the $G$-action on $X_1$ and $X_2$ is
trivial, then
\begin{displaymath}
    \mathrm{Hom}_G(X_1 \otimes Y_1\, , X_2 \otimes Y_2) \simeq
    \mathrm{Hom}(X_1\, ,X_2) \otimes \mathrm{Hom}_G(Y_1\, ,Y_2) \;.
\end{displaymath}
\end{lemma}
\begin{proof}
$\mathrm{Hom}(X_1 \otimes Y_1\, , X_2 \otimes Y_2)$ is canonically
isomorphic to $X_1^\ast \otimes Y_1^\ast \otimes X_2 \otimes Y_2$ as
a $G$-representation space, with $G$-equivariant maps corresponding
to $G$-invariant tensors. Since the $G$-action on $X_1^\ast \otimes
X_2^{\vphantom{\ast}}$ is trivial, one sees that $\mathrm{Hom}_G(X_1
\otimes Y_1\, , X_2 \otimes Y_2)$ is isomorphic to the tensor product
of $X_1^\ast\otimes X_2^{\vphantom{\ast}} \simeq \mathrm{Hom}(X_1,
X_2)$ with the space of $G$-invariants in $Y_1^\ast\otimes
Y_2^{\vphantom{\ast}} \,$. The latter in turn is isomorphic to
$\mathrm{Hom}_G(Y_1, Y_2)$.
\end{proof}
\begin{proposition}\label{prop:dualpairs}
Writing $\mathfrak{g}_N \equiv \mathfrak{g}(\mathbb{C}^N)$ for
$\mathfrak{g} = \mathfrak{gl}\,$, $\mathfrak{o}\,$,
$\mathfrak{sp}\,$, the following pairs are dual pairs in
$\mathfrak{osp}(W):$ $(\mathfrak{gl}(U),\mathfrak{gl}_N)$,
$(\mathfrak{osp}(U\oplus U^*), \mathfrak{o}_N)$,
$(\mathfrak{osp}(\widetilde{U}\oplus \widetilde{U}^*),
\mathfrak{sp}_N)$.
\end{proposition}
\begin{proof}
Here we calculate the centralizer of $\mathfrak{k}$ in
$\mathfrak{osp}(W)$ for each of the three cases $\mathfrak{k} =
\mathfrak{gl}_N\,$, $\mathfrak{o}_N\,$, $\mathfrak{sp}_N$ and refer
the reader to \cite{H1} for the remaining details.

Since both $V \subset W$ and $V^* \subset W$ are $K$-invariant
subspaces, $\mathrm{End}_K(W)$ decomposes as
\begin{displaymath}
    \mathrm{End}_K(W) = \mathrm{End}_K(V) \oplus \mathrm{Hom}_K(V^*,V)
    \oplus \mathrm{Hom}_K(V,V^*) \oplus \mathrm{End}_K(V^*) \;.
\end{displaymath}
By Schur's lemma, $\mathrm{End}_K(\mathbb{C}^N) \simeq \mathbb{C}\,$,
and therefore Lemma \ref{lem:HomG} implies
\begin{displaymath}
    \mathrm{End}_K(V) = \mathrm{End}_K(U \otimes \mathbb{C}^N) \simeq
    \mathrm{End}(U) \otimes \mathrm{End}_K(\mathbb{C}^N) = \mathrm{End}(U)\;.
\end{displaymath}
By the same reasoning, $\mathrm{End}_K(V^*) = \mathrm{End}(U^*)$.
Applying Lemma \ref{lem:HomG} to the two remaining summands, we
obtain
\begin{displaymath}
    \mathrm{Hom}_K(V,V^*) \simeq \mathrm{Hom}(U,U^*) \otimes
    \mathrm{Hom}_K(\mathbb{C}^N,{\mathbb{C}^N}^*),
\end{displaymath}
plus the same statement where each vector space is replaced by its
dual.

If $K = \mathrm{GL}(\mathbb{C}^N) \equiv \mathrm{GL}_N\,$, then
$\mathrm{Hom}_K(\mathbb{C}^N ,{\mathbb{C}^N}^*) = \mathrm{Hom}_K
({\mathbb{C}^N}^* , \mathbb{C}^N) = \{ 0\}$. Hence,
\begin{displaymath}
    \Phi : \,\, \mathrm{End}(U) \oplus \mathrm{End}(U^\ast) \to
    \mathrm{End}_{\mathrm{GL}_N}(W)\;, \quad X \oplus Y \mapsto
    (X \otimes \mathrm{Id}) \times (Y \otimes
\mathrm{Id})\;,
\end{displaymath}
for $W = U \otimes \mathbb{C}^N \oplus U^\ast \otimes
{\mathbb{C}^N}^\ast$ is an isomorphism. This means that the
centralizer of $\mathfrak{gl}_N$ in $\mathfrak{osp}(W)$ is the
intersection $\Phi(\mathrm{End}(U)\oplus\mathrm{End}(U^*)) \cap
\mathfrak{osp}(W)$, which can be identified with $\mathrm{End}(U) =
\mathfrak{gl}(U)$. Thus we have the first dual pair, $(\mathfrak{gl}
(U), \mathfrak{gl}_N)$.

In the case of $K = \mathrm{O}_N\,$, the discussion is shortened by
recalling Lemma \ref{lem:2.11} and the $K$-equivariant isomorphism
$\Psi : \, (U \oplus U^\ast) \otimes \mathbb{C}^N \to W$. By Schur's
lemma, these imply $\mathrm{End}_K(W) \simeq \mathrm{End}(U \oplus
U^\ast)$. From Corollary \ref{cor:2.2} it then follows that the
intersection $\mathfrak{osp}(W) \cap \mathrm{End}_K(W)$ is isomorphic
as a Lie superalgebra to $\mathfrak{osp}(U \oplus U^\ast)$. Passing
to the Lie algebra level for $K$, we get the second dual pair,
$(\mathfrak{osp}(U \oplus U^\ast), \mathfrak{o}_N)$.

Finally, if $K = \mathrm{Sp}_N\,$, the situation is identical except
that Corollary \ref{cor:2.2} compels us to switch to the
$\mathbb{Z}_2$-twisted structure of orthosymplectic Lie superalgebra
in $\mathrm{End}_K(W) \simeq \mathrm{End}(U \oplus U^\ast)$. This
gives us the third dual pair, $(\mathfrak{osp}(\widetilde{U} \oplus
\widetilde{U}^\ast), \mathfrak{sp}_N)$.
\end{proof}

\subsection{Clifford-Weyl algebra $\mathfrak{q}(W)$}

Let $\mathbb{K} = \mathbb{C}$ or $\mathbb{K} = \mathbb{R}$ (in this
subsection the choice of number field again is immaterial) and recall
from Example \ref{exa:JH} the definition of the Jordan-Heisenberg Lie
superalgebra $\widetilde{W} = W \oplus \mathbb{K}\,$, where $W = W_0
\oplus W_1$ is a $\mathbb{Z}_2$-graded vector space with components
$W_1^{\vphantom{ \ast}} = V_1^{\vphantom{\ast}} \oplus V_1^\ast$ and
$W_0^{\vphantom{ \ast}} = V_0^{\vphantom{\ast}} \oplus V_0^\ast$. The
universal enveloping algebra of the Jordan-Heisenberg Lie
superalgebra is called the \emph{Clifford-Weyl algebra} (or quantum
algebra). We denote it by $\mathfrak{q}(W) \equiv
\mathsf{U}(\widetilde{W})$.

Equivalently, one defines the Clifford-Weyl algebra $\mathfrak{q}(W)$
as the associative algebra generated by $\widetilde{W} = W \oplus
\mathbb{K}$ subject to the following relations for all $w, w^\prime
\in W_0 \cup W_1:$
\begin{displaymath}
    w w' - (-1)^{|w||w'|} w' w = Q(w , w') \;.
\end{displaymath}
In particular, $w_0\, w_1 = w_1\, w_0$ for all $w_0 \in W_0$ and $w_1
\in W_1\,$. Reordering by this commutation relation defines an
isomorphism of associative algebras $\mathfrak{q}(W) \simeq
\mathfrak{c}(W_1) \otimes \mathfrak{w}(W_0)$, where the Clifford
algebra $\mathfrak{c}(W_1)$ is generated by $W_1 \oplus \mathbb{K}$
with the relations $w w^\prime + w^\prime w = S(w,w^\prime)$ for
$w,w^\prime \in W_1\,$, and the Weyl algebra $\mathfrak{w}(W_0)$ is
generated by $W_0 \oplus \mathbb{K}$ with the relations $w w^\prime -
w^\prime w = A(w,w^\prime)$ for $w, w^\prime \in W_0\,$.

As a universal enveloping algebra the Clifford-Weyl algebra
$\mathfrak{q} (W)$ is filtered,
\begin{displaymath}
    \mathfrak{q}_0(W) := \mathbb{K} \subset \mathfrak{q}_1(W) :=
    W \oplus \mathbb{K} \subset \ldots \subset \mathfrak{q}_n(W)
    \ldots \;,
\end{displaymath}
and it inherits from the Jordan-Heisenberg algebra $\widetilde{W}$ a
canonical $\mathbb{Z}_2$-grading and a canonical structure of Lie
superalgebra by the supercommutator -- see $\S$\ref{sect:UEA} for the
definitions. The next statement is a sharpened version of Lemma
\ref{lem:2.2}.
\begin{lemma}\label{lem:n+n'-2}
    $[\mathfrak{q}_n(W) , \mathfrak{q}_{n^\prime}(W)] \subset
    \mathfrak{q}_{n + n^\prime - 2}(W)$.
\end{lemma}
\begin{proof}
Lemma \ref{lem:2.2} asserts the commutation relation $[\mathsf{U}_n
(\mathfrak{g}), \mathsf{U}_{n^\prime} (\mathfrak{g}) ] \subset
\mathsf{U}_{n + n^\prime-1} (\mathfrak{g})$ for the general case of a
Lie superalgebra $\mathfrak{g}$ with bracket $[ \mathfrak{g},
\mathfrak{g}] \subset \mathfrak{g}\,$. For the specific case at hand,
where the fundamental bracket $[W , W] \subset \mathbb{K}$ has zero
component in $W$, the degree $n + n^\prime - 1$ is lowered to $n +
n^\prime -2$ by the very argument proving that lemma.
\end{proof}
It now follows that each of the subspaces $\mathfrak{q}_n(W)$ for $n
\le 2$ is a Lie superalgebra. Since $[\mathfrak{q}_2(W),
\mathfrak{q}_1(W)] \subset \mathfrak{q}_1(W)$, the quotient space
$\mathfrak{q}_2(W)/\mathfrak{q}_1(W)$ is also a Lie superalgebra. By
the Poincar\'e-Birkhoff-Witt theorem, there exists a vector-space
isomorphism
\begin{displaymath}
    \mathfrak{q}_2(W) / \mathfrak{q}_1(W) \stackrel{\sim}{\to}
    \mathfrak{s} \;,
\end{displaymath}
sending $\mathfrak{q}_2(W) / \mathfrak{q}_1(W)$ to $\mathfrak{s}\,$,
the space of super-symmetrized degree-two elements in $\mathfrak{q}_2
(W)$. Hence $\mathfrak{q}_2(W)$ has a direct-sum decomposition
$\mathfrak{q}_2(W) = \mathfrak{q}_1(W) \oplus \mathfrak{s}\,$.

If $\{ e_i \}$ is a homogeneous basis of $W$, every $a \in
\mathfrak{s}$ is uniquely expressed as
\begin{equation}\label{eq:2.8}
    a = \sum\nolimits_{i,j} a_{ij}\, e_i \, e_j \;, \quad a_{ij} =
    (-1)^{|e_i| |e_j|} a_{ji} \;.
\end{equation}
By adding and subtracting terms,
\begin{displaymath}
    2 w w^\prime = (w w^\prime + (-1)^{|w||w^\prime|} w^\prime w)
    + (w w^\prime - (-1)^{|w||w^\prime|} w^\prime w) \;,
\end{displaymath}
one sees that the product $w w^\prime$ for $w,w^\prime \in W$ has
scalar part $(w w^\prime)_\mathbb{K} = \frac{1}{2}[w,w^\prime] =
\frac{1}{2}Q(w,w^\prime)$ with respect to the decomposition
$\mathfrak{q}_2(W) = \mathbb{K} \oplus W \oplus \mathfrak{s}\,$.
\begin{lemma}
$[\mathfrak{s} , \mathfrak{s}] \subset \mathfrak{s}\,$.
\end{lemma}
\begin{proof}
From the definition of $\mathfrak{s}$ and $[W,W] \subset \mathbb{K}$
it is clear that $[\mathfrak{s}, \mathfrak{s}] \subset \mathbb{K}
\oplus \mathfrak{s}\,$. The statement to be proved, then, is that
$[a,b]$ for $a,\, b \in \mathfrak{s}$ has zero scalar part.

By the linearity of the supercommutator, it suffices to consider a
single term of the sum (\ref{eq:2.8}). Thus we put $a = w w^\prime +
(-1)^{|w| |w^\prime|} w^\prime w$, and have
\begin{displaymath}
    {\textstyle{\frac{1}{2}}} [a,b] = [w w^\prime,b] = w\,
    [w^\prime,b] + [w,b] \, w^\prime (-1)^{|w^\prime| |b|} \;.
\end{displaymath}
Now we compute the scalar part of the right-hand side. Using the
Jacobi identity for the Lie superalgebra $\mathfrak{q}(W)$ we obtain
\begin{displaymath}
    [a,b]_\mathbb{K} = [w,[w^\prime,b]] + [[w,b],w^\prime]
    (-1)^{|w^\prime| |b|} = [[w,w^\prime],b] \;.
\end{displaymath}
The last expression vanishes because $[w,w^\prime] \subset
\mathbb{K}$ lies in the center of $\mathfrak{q}(W)$.
\end{proof}

\subsection{$\mathfrak{osp}(W)$ inside $\mathfrak{q}(W)$}

As a subspace of $\mathfrak{q}(W)$ which closes w.r.t.\ the
supercommutator $[\, , \, ]$, $\mathfrak{s}$ is a Lie superalgebra.
Now from Lemma \ref{lem:n+n'-2} and the Jacobi identity for
$\mathfrak{q}(W)$, one sees that $\mathfrak{s} \subset
\mathfrak{q}_2(W)$ acts on each of the quotient spaces
$\mathfrak{q}_n(W) / \mathfrak{q}_{n - 1}(W)$ for $n \ge 1$ by $a
\mapsto [a,\,\,]\,$. In particular, $\mathfrak{s}$ acts on
$\mathfrak{q}_1(W) / \mathfrak{q}_0(W) = W$ by $a \mapsto [a,\,\,]$,
which defines a homomorphism of Lie superalgebras
\begin{displaymath}
    \tau : \, \mathfrak{s} \to \mathfrak{gl}(W) \;, \quad
    a \mapsto \tau(a) = [a , \,\, ] \;.
\end{displaymath}
The mapping $\tau$ is actually into $\mathfrak{osp}(W) \subset
\mathfrak{gl}(W)$. Indeed, for $w,w^\prime \in W$ one has
\begin{eqnarray*}
    Q(\tau(a)w,w')+(-1)^{|\tau(a)||w|}Q(w,\tau(a)w')
    = [[a,w],w']+(-1)^{|a||w|}[w,[a,w']] \;,
\end{eqnarray*}
and since $[a,[w,w^\prime]] = 0\,$, this vanishes by the Jacobi
identity.
\begin{lemma}\label{lem:tau-iso}
The map $\tau :\, \mathfrak{s} \to \mathfrak{osp}(W)$ is an
isomorphism of Lie superalgebras.
\end{lemma}
\begin{proof}
Being a homomorphism of Lie superalgebras, the linear mapping $\tau$
is an isomorphism of such algebras if it is bijective. We first show
that $\tau$ is injective. So, let $a \in \mathfrak{s}$ be any element
of the kernel of $\tau\,$. The equation $\tau (a) = 0$ means that
$[[a,w],w^\prime] = [\tau(a)w,w^\prime]$ vanishes for all $w,
w^\prime \in W$. To fathom the consequences of this, let $\{e_i\}$
and $\{ \widetilde{e}_i \}$ be two homogeneous bases of $W$ so that
$Q(e_i\,,\widetilde{e}_j) = \delta_{ij} \,$. Using that $a \in
\mathfrak{s}$ has a uniquely determined expansion $a = \sum a_{ij}\,
e_i \, e_j$ with supersymmetric coefficients $a_{ij} = (-1)^{|e_i|
|e_j|} a_{j\,i}\,$, one computes
\begin{displaymath}
    [ [ a , \widetilde{e}_j ] , \widetilde{e}_i ] = a_{ij} +
    (-1)^{|e_i| |e_j|} a_{j\,i} = 2 a_{ij} \;.
\end{displaymath}
Thus the condition $[[a,w],w^\prime] = 0$ for all $w, w^\prime \in W$
implies $a = 0\,$. Hence $\tau$ is injective.

By the Poincar\'e-Birkhoff-Witt isomorphism
\begin{displaymath}
    \mathfrak{s} \simeq \mathfrak{q}_2(W) / \mathfrak{q}_1(W) \simeq
    \sum\nolimits_{k+l = 2}\wedge^k (W_1) \otimes \mathrm{S}^l(W_0) \;,
\end{displaymath}
the dimensions of the $\mathbb{Z}_2$-graded vector space
$\mathfrak{s} = \mathfrak{s}_0 \oplus \mathfrak{s}_1$ are
\begin{displaymath}
    \dim\, \mathfrak{s}_0 = \dim\, \wedge^2 (W_1) + \dim \mathrm{S}^2
    (W_0) \;, \quad \dim\, \mathfrak{s}_1 = \dim W_1 \, \dim W_0 \;.
\end{displaymath}
These agree with those of $\mathfrak{osp}(W)$ as recorded in Corollary
\ref{cor:dimosp}. Hence our injective linear map $\tau : \, \mathfrak{s}
\to \mathfrak{osp}(W)$ is in fact a bijection.
\end{proof}
\begin{remark}
By the isomorphism $\tau$ every representation $\rho$ of
$\mathfrak{s} \subset \mathfrak{q}(W)$ induces a representation $\rho
\circ \tau^{-1}$ of $\mathfrak{osp}(W)$.
\end{remark}
Let us conclude this subsection by writing down an explicit formula
for $\tau^{-1}$. To do so, let $\{ e_i \}$ and $\{ \widetilde{e}_j
\}$ be homogeneous bases of $W$ with $Q(e_i \, , \, \widetilde{e}_j)
= \delta_{ij}$ as before. For $X \in \mathfrak{osp}(W)$ notice that the
coefficients $a_{ij} := Q(e_i \, , X e_j) (-1)^{|e_j|}$ are
supersymmetric:
\begin{displaymath}
    a_{ij} = (-1)^{|X||e_i| + 1 + |e_j|} Q(X e_i \, ,\, e_j)
    = (-1)^{|X||e_i|} Q(e_j \, , X e_i) = (-1)^{|e_i| |e_j|}
    a_{j\,i} \;,
\end{displaymath}
where the last equality sign uses $(-1)^{|e_i| |e_j|}\, a_{j\,i} =
(-1)^{|e_i| |X e_i|}\, a_{j\,i} = (-1)^{|e_i| |X| + |e_i|}\,
a_{j\,i}\,$.

The inverse map $\tau^{-1} : \, \mathfrak{osp}(W) \to \mathfrak{s}$
is now expressed as
\begin{equation}\label{eq:tau-inv}
    \tau^{-1}(X) = {\textstyle{\frac{1}{2}}} \sum\nolimits_{i,j}
    Q(e_i \, , X e_j) (-1)^{|e_j| + 1} \, \widetilde{e}_i \,
    \widetilde{e}_j \;.
\end{equation}
To verify this formula, one calculates the double supercommutator $[
e_i\, , [\tau^{-1}(X) , e_j] ]$ and shows that the result is equal to
$[e_i \, , X e_j] = Q(e_i\, , X e_j)$, which is precisely what is
required from the definition of $\tau$ by $[\tau^{-1}(X),\, e_j] = X
e_j\,$.

\subsection{Spinor-oscillator representation}
\label{sect:osp-repn}

As before, starting from a $\mathbb{Z}_2$-graded $\mathbb{K}$-vector
space $V = V_0 \oplus V_1\,$, let the direct sum $W = V\oplus V^*$ be
equipped with the orthosymplectic form $Q$ and denote by
$\mathfrak{q}(W)$ the Clifford-Weyl algebra of $W$.

Consider now the following tensor product of exterior and symmetric
algebras:
\begin{displaymath}
     \mathfrak{a}(V) := \wedge(V_1^*) \otimes \mathrm{S}(V_0^*) \;.
\end{displaymath}
Following R.\ Howe we call it the \emph{spinor-oscillator module} of
$\mathfrak{q}(W)$. Notice that $\mathfrak{a}(V)$ can be identified
with the graded-commutative subalgebra in $\mathfrak{q}(W)$ which is
generated by $V \oplus \mathbb{K}\,$. As such, $\mathfrak{a}(V)$
comes with a canonical $\mathbb{Z}_2$-grading and its space of
endomorphisms carries a structure of Lie superalgebra, $\mathfrak{gl}
(\mathfrak{a}(V)) \equiv \mathrm{End}(\mathfrak{a}(V))$.

The algebra $\mathfrak{a}(V)$ now is to become a representation space
for $\mathfrak{q}(W)$. Four operations are needed for this: the
operator $\varepsilon(\varphi_1) : \, \wedge^k(V_1^\ast) \to
\wedge^{k+1} (V_1^\ast)$ of exterior multiplication by a linear form
$\varphi_1 \in V_1^\ast\,$; the operator $\iota(v_1) : \, \wedge^k(
V_1^\ast) \to \wedge^{k-1}(V_1^\ast)$ of alternating contraction with
a vector $v_1\in V_1\,$; the operator $\mu(\varphi_0):\, \mathrm{S}^l
(V_0^\ast) \to \mathrm{S}^{l+1}(V_0^\ast)$ of multiplication with a
linear function $\varphi_0 \in V_0^\ast\,$; and the operator
$\delta(v_0) : \, \mathrm{S}^l (V_0^\ast) \to \mathrm{S}^{l-1}
(V_0^\ast)$ of taking the directional derivative by a vector $v_0 \in
V_0\,$.

The operators $\varepsilon$ and $\iota$ obey the \emph{canonical
anti-commutation relations} (CAR), which is to say that $\varepsilon
(\varphi)$ and $\varepsilon(\varphi^\prime)$ anti-commute, $\iota(v)$
and $\iota( v^\prime)$ do as well, and one has
\begin{displaymath}
    \iota(v) \varepsilon(\varphi) + \varepsilon(\varphi) \iota(v)
    = \varphi(v) \, \mathrm{Id}_{\wedge(V_1^\ast)} \;.
\end{displaymath}
The operators $\mu$ and $\delta$ obey the \emph{canonical commutation
relations} (CCR), i.e., $\mu(\varphi)$ and $\mu(\varphi^\prime)$
commute, so do $\delta(v)$ and $\delta(v^\prime)$, and one has
\begin{displaymath}
    \delta(v) \mu(\varphi) - \mu(\varphi) \delta(v)
    = \varphi(v) \, \mathrm{Id}_{\mathrm{S}(V_0^\ast)} \;.
\end{displaymath}
Given all these operations, one defines a linear mapping $q : \, W
\to \mathrm{End}(\mathfrak{a}(V))$ by
\begin{displaymath}
    q(v_1 + \varphi_1 + v_0 + \varphi_0) = \iota(v_1) +
    \varepsilon(\varphi_1) + \delta(v_0) + \mu(\varphi_0) \quad
    (v_s \in V_s\, , \, \varphi_s \in V^*_s)\;,
\end{displaymath}
with $\iota(v_1)$, $\varepsilon(\varphi_1)$ operating on the first
factor of the tensor product $\wedge(V_1^\ast) \otimes \mathrm{S}
(V_0^\ast)$, and $\delta(v_0)$, $\mu(\varphi_0)$ on the second
factor. Of course the two sets $\varepsilon, \iota$ and $\mu, \delta$
commute with each other. In terms of $q\,$, the relations CAR and CCR
are succinctly summarized as
\begin{equation}\label{eq:2.9}
    [q(w)\,,\,q(w^\prime)] = Q(w,w^\prime)\,
    \mathrm{Id}_{\mathfrak{a}(V)}\;,
\end{equation}
where $[\, , \,]$ denotes the usual supercommutator of the Lie
superalgebra $\mathfrak{gl}(\mathfrak{a}(V))$. By the relation
(\ref{eq:2.9}) the linear map $q$ extends to a representation of the
Jordan-Heisenberg Lie super\-algebra $\widetilde{W} = W \oplus
\mathbb{K}\,$, with the constants of $\widetilde{W}$ acting as
multiples of $\mathrm{Id}_{\mathfrak{a}(V)}\,$.

Moreover, being a representation of $\widetilde{W}$, the map $q$
yields a representation of the universal enveloping algebra
$\mathsf{U}( \widetilde{W}) \equiv \mathfrak{q}(W)$. This
representation is referred to as the \emph{spinor-oscillator
representation} of $\mathfrak{q}(W)$. In the sequel we will be
interested in the $\mathfrak{osp} (W)$-representation induced from it
by the isomorphism $\tau^{-1}$.

There is a natural $\mathbb{Z}$-grading $\mathfrak{a}(V) =
\bigoplus_{m \ge 0} \, \mathfrak{a}^m(V)\,$,
\begin{displaymath}
    \mathfrak{a}^m(V) = {\textstyle\bigoplus\nolimits}_{k + l = m}
    \wedge^k(V_1^*)\otimes \mathrm{S}^l(V_0^\ast) \;.
\end{displaymath}
Note that the operators $\varepsilon(\varphi_1)$ and $\mu(\varphi_0)$
increase the $\mathbb{Z}$-degree of $\mathfrak{a}(V)$ by one, while
the operators $\iota(v_1)$ and $\delta(v_0)$ decrease it by one. Note
also if ${C} = (- \mathrm{Id}_V) \oplus \mathrm{Id}_{V^\ast}$ is the $\mathfrak{osp}$-element introduced in $\S$\ref{sect:osp-cas}, then a direct computation using the formula (\ref{eq:tau-inv}) shows that $\mathfrak{a}^m(V)$ is an eigenspace of the operator $(q \circ\tau^{-1}) ({C})$ with eigenvalue $m\,$. Thus ${C} \in \mathfrak{osp}$ is represented on the spinor-oscillator module $\mathfrak{a}(V)$ by the \emph{degree}.

\subsubsection{Weight constraints}\label{subsubsec:WoHdP}

We now specialize to the situation of $V = U \otimes \mathbb{C}^N$
with $U = U_0 \oplus U_1$ a $\mathbb{Z}_2$-graded vector space as in
$\S$\ref{sect:howe-pairs}, and we require $U_0$ and $U_1$ to be
isomorphic with dimension $\dim U_0 = \dim U_1 = n\,$. Recall that
\begin{displaymath}
    (\mathfrak{osp}(U\oplus U^\ast),\mathfrak{o}_N)\;, \quad
    (\mathfrak{osp}
    (\widetilde{U} \oplus \widetilde{U}^*),\mathfrak{sp}_N)\;,
\end{displaymath}
are Howe dual pairs in $\mathfrak{osp}(W)$. As of now we denote these by
$(\mathfrak{g}, \mathfrak{k})$. There is a decomposition
\begin{eqnarray*}
    &&\mathfrak{g} = \mathfrak{g}^{(-2)} \oplus \mathfrak{g}^{(0)}
    \oplus \mathfrak{g}^{(2)}\;, \quad \mathfrak{g}^{(0)} =
    \mathfrak{g} \cap (\mathrm{End}(U)\oplus\mathrm{End}(U^*)) \;, \\
    &&\mathfrak{g}^{(-2)} = \mathfrak{g} \cap \mathrm{Hom}(U^*,U)\;,
    \quad \mathfrak{g}^{(2)} = \mathfrak{g} \cap \mathrm{Hom}(U,U^*) \;,
\end{eqnarray*}
in both cases. The notation highlights the fact that the operators in
$\mathfrak{g}^{(m)} \hookrightarrow \mathfrak{osp}(V \oplus V^\ast)$
change the degree of elements in $\mathfrak{a}(V)$ by $m\,$. Note
that the Cartan subalgebra $\mathfrak{h}$ of diagonal operators in
$\mathfrak{g}$ is contained in $\mathfrak{g}^{(0)}$ but
$\mathfrak{h}\ne\mathfrak{g}^{(0)} $.

Since the Lie algebra $\mathfrak{k}$ is defined on $\mathbb{C}^N$,
the $\mathfrak{k}$-action on $\mathfrak{a}(V)$ preserves the degree.
This action exponentiates to an action of the complex Lie group $K$
on $\mathfrak{a}(V)$.
\begin{proposition}\label{prop:Howeduality}
The subalgebra $\mathfrak{a}(V)^K$ of $K$-invariants in
$\mathfrak{a}(V)$ is an irreducible module for $\mathfrak{g}\,$. The
vacuum $1 \in \mathfrak{a}(V)^K$ is contained in it as a cyclic
vector such that
\begin{displaymath}
    \mathfrak{g}^{(-2)}.1 = 0 \;, \quad \mathfrak{g}^{(0)}.1 =
\langle 1 \rangle_{\mathbb{C}}
    \;, \quad \langle \mathfrak{g}^{(2)}.1 \rangle_{\mathbb{C}}
= \mathfrak{a}(V)^K \;.
\end{displaymath}
\end{proposition}
\begin{proof}
This is a restatement of Theorems 8 and 9 of \cite{H1}.
\end{proof}
\begin{remark}
In the case of $(\mathfrak{g}, \mathfrak{k}) = (\mathfrak{osp}(U
\oplus U^\ast), \mathfrak{o}_N )$ it matters that $K =
\mathrm{O}_N\,$, as the connected Lie group $K = \mathrm{SO}_N$ has
invariants in $\mathfrak{a}(V)$ not contained in $\langle
\mathfrak{g}^{(2)}.1 \rangle_\mathbb{C} \,$.
\end{remark}
Proposition \ref{prop:Howeduality} has immediate consequences for the
weights of the $\mathfrak{g}$-representation on $\mathfrak{a}(V)^K$.
Using the notation of $\S$\ref{sect:osp-roots}, let $\{ H_{sj} \}$ be
a standard basis of $\mathfrak{h}$ and $\{ \vartheta_{sj} \}$ the
corresponding dual basis. We now write $\vartheta_{0j} =: \phi_j\,$
and $\vartheta_{1j} =: \mathrm{i}\psi_j$ ($j = 1, \ldots, n$).
\begin{corollary}\label{cor:weights}
The representations of $\mathfrak{osp}(U\oplus U^*)$ on
$\mathfrak{a}(V)^{\mathrm{O}_N}$ and $\mathfrak{osp}(\widetilde{U}
\oplus \widetilde {U}^*)$ on $\mathfrak{a}(V)^{ \mathrm{Sp}_N}$ each
have highest weight $\lambda_N = \frac{N}{2}\sum_{j=1}^n
(\mathrm{i}\psi_j - \phi_j)$. Every weight of these representations
is of the form $\sum_{j=1}^n (\mathrm{i} m_j \psi_j - n_j \phi_j)$
with $-\frac{N}{2}\le m_j \le \frac{N}{2} \le n_j\,$.
\end{corollary}
\begin{proof}
Recall from $\S$\ref{sect:howe-pairs} the embedding of
$\mathfrak{osp} (U\oplus U^*)$ and $\mathfrak{osp}(\widetilde{U}
\oplus \widetilde{U}^*)$ in $\mathfrak{osp}(W)$, and from Lemma
\ref{lem:tau-iso} the isomorphism $\tau^{-1} : \, \mathfrak{osp}(W)
\to \mathfrak{s}$ where $\mathfrak{s}$ is the Lie superalgebra of
supersymmetrized degree-two elements in $\mathfrak{q}(W)$.
Specializing formula (\ref{eq:tau-inv}) to the case of a Cartan
algebra generator $H_{sj} \in \mathfrak{h} \subset \mathfrak{g}$ one
gets
\begin{displaymath}
    \tau^{-1}(H_{s j}) = - \tfrac{1}{2} \sum\nolimits_{a = 1}^N
    ((f_{s,j}\otimes f_a)(e_{s,j} \otimes e_a) + (-1)^s
    (e_{s,j}\otimes e_a )(f_{s,j}\otimes f_a) ) \;,
\end{displaymath}
where $\{e_a \}$ is a basis of $\mathbb{C}^N$ and $\{ f_a \}$ the
dual basis of $(\mathbb{C}^N)^\ast$.

Now let $\tau^{-1}(H_{sj}) \in \mathfrak{s}$ act by the corresponding
operator, say $\hat{H}_{sj} := (q \circ \tau^{-1})(H_{sj})$, in the
spinor-oscillator representation $q$ of $\mathfrak{s} \subset
\mathfrak{q} (W)$. Application of that operator to the highest-weight
vector $1 \in \mathbb{C} \equiv \wedge^0(V_1^\ast) \otimes
\mathrm{S}^0 (V_0^\ast) \subset \mathfrak{a}(V)^K$ yields
\begin{align*}
    \hat{H}_{1 j}\, 1 &= \tfrac{1}{2}\sum\nolimits_a \iota(e_{1,j}
    \otimes e_a)\varepsilon(f_{1,j}\otimes f_a) 1 = \tfrac{N}{2}\;,\\
    \hat{H}_{0 j}\, 1 &= - \tfrac{1}{2}\sum\nolimits_a \delta(e_{0,j}
    \otimes e_a) \mu(f_{0,j}\otimes f_a) 1 = - \tfrac{N}{2} \;.
\end{align*}
Altogether this means that $\hat{H}\, 1 = \lambda_N(H) 1$ where
$\lambda_N(H)= \frac{N}{2} \sum_j (\mathrm{i}\psi_j(H) - \phi_j(H))$.

From Lemma \ref{lem:rootsosp} the roots $\alpha$ corresponding to
root spaces $\mathfrak{g}_\alpha \subset \mathfrak{g}^{(2)}$ are of
the form
\begin{displaymath}
    -\phi_j -\phi_{j^{\,\prime}} \;, \quad -\mathrm{i}\psi_j -\mathrm{i}
    \psi_{j^{\,\prime}}\;,\quad -\phi_j-\mathrm{i}\psi_{j^{\,\prime}}\;,
\end{displaymath}
where the indices $j, j^{\,\prime}$ are subject to restrictions that
depend on $\mathfrak{g}$ being $\mathfrak{osp} (U\oplus U^*)$ or
$\mathfrak{osp}( \widetilde{U}\oplus \widetilde{U}^*)$. From
$\mathfrak{a}(V)^K = \mathfrak{g}^{ (2)}.1$ one then has $m_j \le
\frac{N}{2} \le n_j$ for every weight $\gamma = \sum (\mathrm{i} m_j
\psi_j - n_j \phi_j)$ of the $\mathfrak{g}$-representation on
$\mathfrak{a}(V)^K$.

The restriction $m_j \ge \frac{N}{2} - N$ results from $\wedge(
V_1^\ast) = \wedge(U_1^\ast \otimes (\mathbb{C}^N)^\ast)$ being
isomorphic to $\otimes_{j=1}^n \wedge(\mathbb{C}^N)^\ast$ and the
vanishing of $\wedge^k (\mathbb{C}^N)^\ast = 0$ for $k > N\,$.
\end{proof}
\begin{corollary}\label{cor:degree}
For each of our two cases $\mathfrak{g} = \mathfrak{osp}(U \oplus
U^\ast)$ and $\mathfrak{g} = \mathfrak{osp}(\widetilde{U}\oplus
\widetilde{U}^\ast)$ the element ${C} = - \sum_{s,j} H_{s j}
\subset \mathfrak{g}$ is represented on $\mathfrak{a}(V)^K$ by the
degree operator.
\end{corollary}
\begin{proof}
Since the $K$-action on $\mathfrak{a}(V)$ preserves the degree, the
subalgebra $\mathfrak{a}(V)^K$ is still $\mathbb{Z}$-graded by the
same degree. Summing the above expressions for $(q \circ \tau^{-1})
(H_{s j})$ over $s,j$ and using CAR and CCR to combine terms, we
obtain
\begin{displaymath}
    (q \circ \tau^{-1})({C}) = \sum_{j=1}^n \sum_{a=1}^N
    \big( \mu(f_{0,j}\otimes f_a) \delta(e_{0,j} \otimes e_a) +
    \varepsilon(f_{1,j}\otimes f_a) \iota(e_{1,j}\otimes e_a) \big)
    \;,
\end{displaymath}
which is in fact the operator for the degree of the
$\mathbb{Z}$-graded module $\mathfrak{a}(V)^K$.
\end{proof}

\subsubsection{Positive and simple roots}
\label{sect:simple-roots}

We here record the systems of simple positive roots that we will use
later (in $\S$\ref{unicity theorem}). In the case of
$\mathfrak{osp}(U \oplus U^\ast)$ this will be
\begin{displaymath}
    \phi_1 - \phi_2 \;, \ldots, \phi_{n-1} - \phi_n \;,
    \phi_n - \mathrm{i}\psi_1 \;, \mathrm{i}\psi_1 - \mathrm{i}
    \psi_2 \;, \ldots, \mathrm{i}\psi_{n-1} - \mathrm{i}\psi_n \;,
    \mathrm{i} \psi_{n-1} + \mathrm{i} \psi_n \;.
\end{displaymath}
The corresponding system of positive roots for $\mathfrak{osp}(U \oplus
U^\ast)$ is
\begin{displaymath}
    \phi_j \pm \phi_k \;, \, \mathrm{i}\psi_j \pm \mathrm{i}\psi_k
    \,\, (j < k) \;, \,\, 2\phi_j \,, \,\, \phi_j \pm \mathrm{i}\psi_k
    \,\, (j\, ,k =1, \ldots, n) \;.
\end{displaymath}
In the case of $\mathfrak{osp}(\widetilde{U} \oplus
\widetilde{U}^\ast)$ we choose the system of simple positive roots
\begin{displaymath}
    \phi_1 - \phi_2 \;, \ldots, \phi_{n-1} - \phi_n \;,
    \phi_n - \mathrm{i}\psi_1 \;, \mathrm{i}\psi_1 -
    \mathrm{i}\psi_2 \;, \ldots,  \mathrm{i}\psi_{n-1}
    - \mathrm{i}\psi_n \;, 2\mathrm{i}\psi_n \;.
\end{displaymath}
The corresponding positive root system then is
\begin{displaymath}
    \phi_j \pm \phi_k \;, \, \mathrm{i}\psi_j \pm \mathrm{i}\psi_k
    \,\, (j < k) \;, \,\, 2\mathrm{i}\psi_j \,, \,\, \phi_j \pm
    \mathrm{i}\psi_k \,\, (j\, ,k =1, \ldots, n) \;.
\end{displaymath}
In both cases the roots
\begin{displaymath}
    \phi_j - \phi_k \;, \, \mathrm{i}\psi_j - \mathrm{i}\psi_k
    \,\, (j < k) \;, \quad \phi_j - \mathrm{i}\psi_k \,\,
    (j\, , k = 1, \ldots, n) \;,
\end{displaymath}
form a system of positive roots for $\mathfrak{gl}(U) \simeq
\mathfrak{g}^{(0)} \subset \mathfrak{osp}$.

\subsubsection{Unitary structure}\label{sect:2.6.3}

We now equip the spinor-oscillator module $\mathfrak{a}(V)$ for $V =
V_0 \oplus V_1$ with a unitary structure. The idea is to think of the
algebra $\mathfrak{a}(V)$ as a subset of $\mathcal{O}(V_0\, , \,
\wedge V_1^\ast)$, the holomorphic functions $V_0 \to
\wedge(V_1^\ast)$. For such functions a Hermitian scalar product is
defined via Berezin's notion of superintegration as follows.

For present purposes, it is imperative that $V$ be defined over
$\mathbb{R}$, i.e., $V = V_\mathbb{R} \otimes \mathbb{C}$, and that
$V$ be re-interpreted as a \emph{real} vector space $V^\prime :=
V_\mathbb{R} \oplus J\, V_\mathbb{R}$ with complex structure $J
\simeq \mathrm{i}\,$. Needless to say, this is done in a manner
consistent with the $\mathbb{Z}_2$-grading, so that $V^\prime =
V_0^\prime \oplus V_1^\prime$ and $V_s^\prime = V_{s,\mathbb{R}}
\oplus J\, V_{s,\mathbb{R}} \simeq V_{s,\mathbb{R}} \otimes
\mathbb{C} = V_s\,$.

% Please do not remove the parentheses in {}From
%
{}From $V_s = U_s \otimes \mathbb{C}^N$ and $U_1 \simeq U_0$ we are
given an isomorphism $V_1 \simeq V_0\,$. This induces a canonical
isomorphism $\wedge ({V_1^\prime}^\ast) \simeq \wedge(
{V_0^\prime}^\ast)$, which gives rise to a bundle isomorphism
$\Omega$ sending $\Gamma(V_0^\prime \, , \, \wedge
{V_1^\prime}^\ast)$, the algebra of real-analytic functions on
$V_0^\prime$ with values in $\wedge ({V_1^\prime}^\ast)$, to
$\Gamma({V_0^\prime} \,,\, \wedge T^\ast V_0^\prime)$, the complex of
real-analytic differential forms on $V_0^\prime\,$. Fixing some
orientation of $V_0^\prime\,$, the Berezin (super-)integral for the
$\mathbb{Z}_2$-graded vector space $V^\prime = V_0^\prime \oplus
V_1^\prime$ is then defined as the composite map
\begin{displaymath}
    \Gamma(V_0^\prime \,,\, \wedge {V_1^\prime}^\ast) \stackrel{\Omega}
    {\longrightarrow} \Gamma(V_0^\prime \,,\,\wedge T^\ast V_0^\prime)
    \ \stackrel{\int}{\longrightarrow} \mathbb{C}\;,\quad \Phi\mapsto
    \Omega[\Phi]\mapsto \int_{V_0^\prime}\Omega[\Phi]\;,
\end{displaymath}
whenever the integral over $V_0^\prime$ exists. Thus the Berezin
integral is a two-step process: first the section $\Phi$ is converted
into a differential form, then the form $\Omega[\Phi]$ is integrated
in the usual sense to produce a complex number. Of course, by the
rules of integration of differential forms only the top-degree
component of $\Omega[\Phi]$ contributes to the integral.

The subspace $V_\mathbb{R} \subset V^\prime$ has played no role so
far, but now we use it to decompose the complexification $V^\prime
\otimes \mathbb{C}$ into holomorphic and anti-holo\-morphic parts:
$V^\prime \otimes \mathbb{C} = V \oplus \overline{V}$ and determine
an operation of complex conjugation $V^\ast \to \overline{V^\ast}$.
We also fix on $V = V_0 \oplus V_1$ a Hermitian scalar product
(a.k.a.\ unitary structure) $\langle \, , \, \rangle$ so that $V_0
\perp V_1\,$. This scalar product determines a parity-preserving
complex anti-linear bijection $c : \, V \to V^\ast$ by $v \mapsto cv
= \langle v , \, \rangle$. Composing $c$ with complex conjugation
$V^\ast \to \overline{V^\ast}$ we get a $\mathbb{C}$-linear
isomorphism $V \to \overline{V^\ast}$, $v \mapsto \overline{cv}\,$.

In this setting there is a distinguished Gaussian section $\gamma \in
\Gamma(V_0^\prime \, , \wedge {V_1^\prime}^\ast \otimes \mathbb{C})$
singled out by the conditions
\begin{equation}\label{eq:def-gamma}
    \forall v_0 \in V_0\, , v_1 \in V_1 : \quad
    \delta(v_0) \gamma = - \mu(\overline{c v_0}) \gamma \;, \quad
    \iota(v_1) \gamma = - \varepsilon(\overline{c v_1}) \gamma \;.
\end{equation}
To get a close-up view of $\gamma\,$, let $\{ e_{0,j }\}$ and $\{
e_{1,j} \}$ be ortho\-normal bases of $V_0$ resp.\ $V_1\,$, and let
$z_j = c e_{0,j}$ and $\zeta_j = c e_{1,j}$ be the corresponding
coordinate functions, with complex conjugates $\overline{z}_j$ and
$\overline{\zeta}_j\,$. Viewing $\zeta_j\,$, $\overline{\zeta}_j$ as
generators of $\wedge({V_1^\prime}^\ast \otimes \mathbb{C})$, our
section $\gamma \in \Gamma(V_0^\prime \,,\, \wedge {V_1^\prime}^\ast
\otimes \mathbb{C})$ is the standard Gaussian
\begin{displaymath}
    \gamma = \mathrm{const} \times \mathrm{e}^{-\sum_j
    (z_j \overline{z}_j + \zeta_j \overline{\zeta}_j )} \;.
\end{displaymath}
We fix the normalization of $\gamma$ by the condition
$\int_{V_0^\prime} \Omega[\gamma] = 1$.

A unitary structure on the spinor-oscillator module $\mathfrak{a}(V)$
is now defined as follows. Let complex conjugation $V^\ast \to
\overline{V^\ast}$ be extended to an algebra anti-homomorphism
$\mathfrak{a}(V) \to \mathfrak{a}(\overline{V})$ by the convention
$\overline{\Phi_1 \Phi_2} = \overline{\Phi}_2 \, \overline{\Phi}_1$
(without any sign factors). Then, if $\Phi_1\, , \Phi_2$ are any two
elements of $\mathfrak{a}(V)$, we view them as holomorphic maps $V_0
\to \wedge (V_1^\ast)$, multiply $\overline{\Phi}_1$ with $\Phi_2$ to
form $\overline{\Phi}_1 \Phi_2 \in \Gamma(V_0^\prime \,, \wedge
{V_1^\prime}^\ast \otimes \mathbb{C})$, and define their Hermitian
scalar product by
\begin{equation}\label{eq:HSP}
    \left\langle \Phi_1 \, , \Phi_2 \right\rangle_{\mathfrak{a}(V)} :=
    \int_{V_0^\prime} \Omega[ \gamma \, \overline{\Phi}_1 \Phi_2 ]\;.
\end{equation}

Let us mention in passing that (\ref{eq:HSP}) coincides with the
unitary structure of $\mathfrak{a}(V)$ used in the Hamiltonian
formulation of quantum field theories and in the Fock space
description of many-particle systems composed of fermions and bosons.
The elements
\begin{equation}\label{eq:occ-no}
    \bigwedge\nolimits_j \zeta_j^{m_j} \otimes \prod\nolimits_j
    z_j^{n_j} / \sqrt{n_j \, !}
\end{equation}
for $m_j \in \{ 0, 1 \}$ and $n_j \in \{ 0, 1, \ldots \}$ form an
orthonormal set in $\mathfrak{a}(V)$, which in physics is called the
occupation number basis of $\mathfrak{a}(V)$.
\begin{lemma}\label{lem:h.c.}
For all $v_0 \in V_0$ and $v_1 \in V_1$ the pairs of operators
$\delta (v_0)$, $\mu(cv_0)$ and $\iota(v_1)$, $\varepsilon(cv_1)$ in
$\mathrm{End} (\mathfrak{a}(V))$ obey the relations
\begin{displaymath}
    \delta(v_0)^\dagger = \mu(c v_0) \;, \quad
    \iota(v_1)^\dagger = \varepsilon(c v_1) \;,
\end{displaymath}
i.e., they are mutual adjoints with respect to the unitary structure
of $\mathfrak{a} (V)$.
\end{lemma}
\begin{proof}
Let $v \in V_0\,$. Since $\overline{\Phi}_1 \in \mathfrak{a}
(\overline{V})$ is anti-holomorphic, we have $\delta(v)
\overline{\Phi}_1 = 0\,$. By the first defining property of $\gamma$
in (\ref{eq:def-gamma}) and the fact that $\delta(v)$ is a
derivation,
\begin{displaymath}
    \gamma\, \overline{\Phi}_1 \delta(v) \Phi_2 =
    \delta(v) \left( \gamma\, \overline{\Phi}_1 \Phi_2 \right)
    + \mu(\overline{cv}) \gamma\, \overline{\Phi}_1 \Phi_2 \;,
\end{displaymath}
and passing to the Hermitian scalar product by the Berezin integral
we obtain
\begin{displaymath}
    \left\langle \Phi_1 , \delta(v) \Phi_2
    \right\rangle_{\mathfrak{a}(V)} = \int_{V_0}
    \Omega[\gamma\,\overline{\Phi}_1 \overline{\mu(cv)}\Phi_2] =
    \left\langle \mu(cv)\Phi_1,\Phi_2\right\rangle_{\mathfrak{a}(V)}\;.
\end{displaymath}
By the definition of the $\dagger$-operation this means that
$\delta(v)^\dagger = \mu(cv)$.

In the case of $v \in V_1$ the argument is similar but for a few sign
changes. Our starting relation changes to
\begin{displaymath}
    \gamma\, \overline{\Phi}_1 \iota(v) \Phi_2 = (-1)^{|\Phi_1|}
    \iota(v) \left( \gamma\, \overline{\Phi}_1 \Phi_2 \right) +
    (-1)^{|\Phi_1|} \varepsilon(\overline{cv})\gamma\,
    \overline{\Phi}_1 \Phi_2\;,
\end{displaymath}
since the operator $\iota(v)$ is an \emph{anti}-derivation. If $v
\mapsto \tilde{v}$ denotes the isomorphism $V_1 \to V_0\,$, then
$\Omega \circ \iota(v) = \iota(\tilde{v}) \circ \Omega$ and the first
term on the right-hand side Berezin-integrates to zero because
$\iota(\tilde{v})$ lowers the degree in $\wedge T^\ast V_0^\prime\,$.
Therefore,
\begin{displaymath}
    \left\langle \Phi_1 , \iota(v) \Phi_2 \right\rangle_{\mathfrak{a}(V)}
    = \int_{V_0^\prime} \Omega[\gamma\,\overline{\Phi}_1 \overline{
    \varepsilon(cv)} \Phi_2] = \left\langle \varepsilon(cv) \Phi_1 ,
    \Phi_2 \right\rangle_{\mathfrak{a}(V)}\;,
\end{displaymath}
which is the statement $\iota(v)^\dagger = \varepsilon(cv)$.
\end{proof}

By the Hermitian scalar product (\ref{eq:HSP}) and the corresponding
$L^2$-norm, the spinor-oscillator module $\mathfrak{a}(V)$ is
completed to a Hilbert space, $\mathcal{A}_V\,$. A nice feature here
is that, as an immediate consequence of the factors $1 / \sqrt{n_j \,
!}$ in the orthonormal basis (\ref{eq:occ-no}), the $L^2$-condition
$\langle \Phi , \Phi \rangle_{\mathfrak{a}(V)} < \infty$ implies absolute
convergence of the power series for $\Phi \in \mathcal{A}_V\,$. Hence
$\mathcal{A}_V$ can be viewed as a subspace of $\mathcal{O}(V_0\, ,
\wedge V_1^\ast)$:
\begin{displaymath}
    \mathcal{A}_V = \{ \Phi \in \mathcal{O}(V_0 \, , \wedge V_1^\ast)
    \mid \langle \Phi , \Phi \rangle_{\mathfrak{a}(V)} < \infty \} \;.
\end{displaymath}
In the important case of isomorphic components $V_0 \simeq V_1\,$, we
may regard $\mathcal{A}_V$ as the Hilbert space of square-integrable
holomorphic differential forms on $V_0\,$.

Note that although $\delta(v)$ and $\mu(\varphi)$ do not exist as
operators on the Hilbert space $\mathcal{A}_V\,$, they do extend to
linear operators on $\mathcal{O}(V_0 \, , \wedge V_1^\ast)$ for all
$v \in V_0$ and $\varphi \in V_0^\ast\,$.

\subsection{Real structures}\label{sect:2.7}

In this subsection we define a real structure for the complex vector
space $W = V\oplus V^\ast$ and describe, in particular, the resulting
real forms of the ($\mathbb{Z}_2$-even components of the) Howe dual partners introduced above.

Recalling the map $c:\, V \to V^\ast$, $v \mapsto \langle v, \,
\rangle\,$, let $W_\mathbb{R} \simeq V$ be the vector space
\begin{displaymath}
    W_{\mathbb{R}} = \{ v+cv \mid v \in V \} \subset V \oplus V^*
    = W \;.
\end{displaymath}
Note that $W_\mathbb{R}$ can be viewed as the fixed point set
$W_\mathbb{R} = \mathrm{Fix}(C)$ of the involution
\begin{displaymath}
    C : \,\, W \to W \;, \quad
    v + \varphi \mapsto c^{-1}\varphi + cv \;.
\end{displaymath}
By the orthogonality assumption, $W_\mathbb{R} = W_{0,\mathbb{R}}
\oplus W_{1,\mathbb{R}}$ where $W_{s,\mathbb{R}} = W_s \cap
W_\mathbb{R}\,$.

The symmetric bilinear form $S$ on $W_1^{\vphantom{\ast}} =
V_1^{\vphantom{\ast}} \oplus V_1^\ast$ restricts to a Euclidean
structure
\begin{displaymath}
    S\, :\,\, W_{1,\mathbb{R}}\times W_{1,\mathbb{R}}\to \mathbb{R}
    \;, \quad (v + cv \, , v'+cv' ) \mapsto 2\, \mathfrak{Re}\langle
    v,v'\rangle \;,
\end{displaymath}
whereas the alternating form $A$ on $W_0^{\vphantom{\ast}} =
V_0^{\vphantom{\ast}} \oplus V_0^\ast$ induces a real-valued
symplectic form
\begin{displaymath}
    \omega = \mathrm{i}A \, : \,\, W_{0,\mathbb{R}} \times
    W_{0,\mathbb{R}} \to \mathbb{R} \;, \quad (v+cv\,, v'+cv')
    \mapsto 2\, \mathfrak{Im}\langle v,v'\rangle \;.
\end{displaymath}
Please be warned that, since $Q = S + A$ fails to be real-valued on
$W_\mathbb{R}\,$, the intersection $\mathfrak{osp}(W) \cap
\mathrm{End}(W_\mathbb{R})$ is \emph{not} a real form of the complex
Lie superalgebra $\mathfrak{osp}(W)$.

The connected classical real Lie groups associated to the bilinear
forms $S$ and $\omega$ are
\begin{align*}
    \mathrm{SO}(W_{1,\mathbb{R}}) & :=\{ g \in
    \mathrm{SL}(W_{1,\mathbb{R}}) \mid \forall w, w^\prime \in
    W_{1,\mathbb{R}} : \, S(gw,gw') = S(w,w')\} \;, \\
    \mathrm{Sp}(W_{0,\mathbb{R}})& := \{
    g\in\mathrm{GL}(W_{0,\mathbb{R}}) \mid \forall w,w^\prime
    \in W_{0,\mathbb{R}} : \, \omega(gw,gw') = \omega(w,w') \} \;.
\end{align*}
They have Lie algebras denoted by $\mathfrak{o}(W_{1,\mathbb{R}})$
and $\mathfrak{sp} (W_{0,\mathbb{R}})$. By construction we have
$\mathfrak{osp}(W)_0 \cap \mathrm{End}(W_\mathbb{R}) \simeq
\mathfrak{o}(W_{1,\mathbb{R}}) \oplus \mathfrak{sp}
(W_{0,\mathbb{R}})$, and this in fact is a real form of the complex
Lie algebra $\mathfrak{osp}(W)_0 \simeq \mathfrak{o} (W_1) \oplus
\mathfrak{sp}(W_0)$.
\begin{proposition}
The elements of $\mathfrak{o}(W_{1,\mathbb{R}}) \oplus
\mathfrak{sp}(W_{0,\mathbb{R}}) \subset \mathfrak{osp}(W)$ are mapped
via $\tau^{-1}$ and the spinor-oscillator representation to
anti-Hermitian operators in $\mathrm{End}(\mathfrak{a}(V))$.
\end{proposition}
\begin{proof}
Let $X \in \mathfrak{o}(W_{1,\mathbb{R}}) \oplus \mathfrak{sp}
(W_{0,\mathbb{R}})$. We know from Lemma \ref{lem:tau-iso} that
$\tau^{-1}(X)$ is a super-symmetrized element of degree two in the
Clifford-Weyl algebra $\mathfrak{q}(W)$. To see the explicit form of
such an element, recall the definition $\tau(a) w = [a , w]\,$. Since
$Q = S + A\,$, and $A$ restricts to $\mathrm{i}\omega\,$, the
fundamental bracket $[\, , \,] : \, W_\mathbb{R} \times W_\mathbb{R}
\to \mathbb{C}$ given by $[w , w^\prime] = Q(w,w^\prime)$ is
real-valued on $W_{1,\mathbb{R}}$ but imaginary-valued on
$W_{0,\mathbb{R}}$. Therefore,
\begin{align*}
    \tau^{-1}(\mathfrak{o}(W_{1,\mathbb{R}})) &=
    \mathrm{span}_\mathbb{R} \{ w w^\prime - w^\prime w \}
    \quad (w, w^\prime \in W_{1,\mathbb{R}})\;, \\ \tau^{-1}
    (\mathfrak{sp}(W_{0,\mathbb{R}})) &= \mathrm{span}_\mathbb{R}
    \{\mathrm{i} w w^\prime + \mathrm{i} w^\prime w \}
    \quad (w, w^\prime \in W_{0,\mathbb{R}}) \;.
\end{align*}
The proposed statement $X^\dagger = -X$ now follows under the
assumption that the spinor-oscillator representation maps every $w
\in W_\mathbb{R}$ to a self-adjoint operator in $\mathrm{End}(
\mathfrak{a}(V))$. But every element $w \in W_\mathbb{R}$ is of the
form $v_1 + c v_1 + v_0 + c v_0$ and this maps to the operator
$\iota(v_1) + \varepsilon(c v_1) + \delta(v_0) + \mu(c v_0)$, which
is self-adjoint by Lemma \ref{lem:h.c.}.
\end{proof}

Given the real structure $W_\mathbb{R}$ of $W$, we now ask how
$\mathrm{End}( W_\mathbb{R})$ intersects with the Howe pairs
$(\mathfrak{osp}(U \oplus U^\ast), \mathfrak{o}_N)$ and
$(\mathfrak{osp}(\widetilde{U} \oplus \widetilde {U}^\ast),
\mathfrak{sp}_N)$ embedded in $\mathfrak{osp}(W)$. By the observation
that $Q$ restricted to $W_\mathbb{R}$ is not real-valued,
$\mathfrak{osp}(U \oplus U^\ast) \cap \mathrm{End}( W_\mathbb{R})$
fails to be a real form of the complex Lie superalgebra
$\mathfrak{osp}(U \oplus U^\ast)$, and the same goes for $\mathfrak{
osp} (\widetilde{U} \oplus \widetilde{U}^\ast)$. Nevertheless, it is
still true that the even components of these intersections are real
forms of the complex Lie algebras $\mathfrak{osp}(U\oplus U^\ast)_0$
and $\mathfrak{osp} (\widetilde{U}\oplus \widetilde{U}^\ast)_0\,$.

The real forms of interest are best understood by expressing them via
blocks with respect to the decomposition $W = V \oplus V^\ast$. Since
$W_\mathbb{R} = \mathrm{Fix}(C)$, the complex linear endomorphisms of
$W$ stabilizing $W_\mathbb{R}$ are given by
\begin{displaymath}
    \mathrm{End}(W_\mathbb{R}) \simeq \{ X \in \mathrm{End}(W)
    \mid X = C X C^{-1} \} \;.
\end{displaymath}
Writing $X$ in block-decomposed form
\begin{displaymath}
    X = \mathsf{A} \oplus \mathsf{B} \oplus \mathsf{C} \oplus
    \mathsf{D} \equiv \begin{pmatrix} \mathsf{A} &\mathsf{B}\\
    \mathsf{C} &\mathsf{D} \end{pmatrix} \;,
\end{displaymath}
where $\mathsf{A} \in \mathrm{End}(V)$, $\mathsf{B} \in
\mathrm{Hom}(V^\ast,V)$, $\mathsf{C} \in \mathrm{Hom}(V,V^\ast)$, and
$\mathsf{D} \in \mathrm{End}(V^\ast)$, the condition $X = C X C^{-1}$
becomes
\begin{displaymath}
    \mathsf{C} = \overline{\mathsf{B}} \;, \quad
    \mathsf{D} = \overline{\mathsf{A}} \;.
\end{displaymath}
The bar here is a short-hand notation for the complex anti-linear
maps
\begin{align*}
    \mathrm{Hom}(V^\ast,V) \to \mathrm{Hom}(V,V^\ast)\;,
    \quad &\mathsf{B} \mapsto \overline{\mathsf{B}} :=c \mathsf{B} c
    \;,\\ \mathrm{End}(V) \to \mathrm{End}(V^\ast)\;, \quad &\mathsf{A}
    \mapsto \overline{\mathsf{A}} := c \mathsf{A} c^{-1} \;.
\end{align*}
When expressed with respect to compatible bases of $V$ and $V^\ast$,
these maps are just the standard operation of taking the complex
conjugate of the matrices of $\mathsf{A}$ and $\mathsf{B}$.

Now, to get an understanding of the intersections $\mathfrak{o}_N
\cap \mathrm{End}(W_\mathbb{R})$ and $\mathfrak{sp}_N \cap
\mathrm{End}(W_\mathbb{R})$, recall the relation $\mathsf{D} = -
\mathsf{A}^\mathrm{t}$ for $X \in \mathfrak{osp}(W)_0$ and the fact
that the action of the complex Lie algebras $\mathfrak{o}_N =
\mathfrak{o}(\mathbb{C}^N)$ and $\mathfrak{sp}_N = \mathfrak{sp}
(\mathbb{C}^N)$ on $W$ stabilizes the decomposition $W = V \oplus
V^\ast$, with the implication that $\mathsf{B} = \mathsf{C} = 0$ in
both cases. Combining $\mathsf{D} = - \mathsf{A}^ \mathrm{t}$ with
$\mathsf{D} = \overline{\mathsf{A}}$ one gets the anti-Hermitian
property $\mathsf{A} = - \overline{\mathsf{A}}^ \mathrm{t}$, which
means that $\mathfrak{o}_N \cap \mathrm{End}(W_\mathbb{R})$ and
$\mathfrak{sp}_N \cap \mathrm{End}(W_\mathbb{R})$ are compact real
forms of $\mathfrak{o}_N$ and $\mathfrak{sp}_N\,$.

Turning to the Howe dual partners of $\mathfrak{o}_N$ and
$\mathfrak{sp}_N\,$, recall from $\S$\ref{sect:howe-pairs} the
isomorphism $\psi : \, \mathbb{C}^N \to (\mathbb{C}^N)^\ast$ and
arrange for it to be an isometry, $\overline{\psi^{-1}} =
\psi^\mathrm{t}$, of the unitary structures of $\mathbb{C}^N$ and
$(\mathbb{C}^N)^\ast$. Recall also the embedding of the two Lie
superalgebras $\mathfrak{osp}(U \oplus U^\ast)$ and $\mathfrak{osp}
(\widetilde{U} \oplus \widetilde{U}^\ast)$ into $\mathfrak{osp}(W) =
\mathfrak{osp}(U \otimes \mathbb{C}^N \oplus U^\ast \otimes
(\mathbb{C}^N)^\ast)$ by
\begin{displaymath}
    \begin{pmatrix} \mathsf{a} &\mathsf{b} \\ \mathsf{c} &\mathsf{d}
    \end{pmatrix} \mapsto \begin{pmatrix} \mathsf{a} \otimes
    \mathrm{Id} &\mathsf{b} \otimes \psi^{-1} \\ \mathsf{c} \otimes
    \psi &\mathsf{d} \otimes \mathrm{Id} \end{pmatrix} =
    \begin{pmatrix} \mathsf{A} &\mathsf{B} \\ \mathsf{C} &\mathsf{D}
    \end{pmatrix} \;.
\end{displaymath}
Here the notation still means the same, i.e., $\mathsf{a} \in
\mathrm{End}(U)$, $\mathsf{b} \in \mathrm{Hom}(U^\ast,U)$, and so on.

Let a real structure $(U \oplus U^\ast)_\mathbb{R}$ of $U \oplus
U^\ast$ be defined in the same way as the real structure
$W_\mathbb{R} = (V \oplus V^\ast)_\mathbb{R}$ of $W = V \oplus
V^\ast$.
\begin{proposition}\label{prop:2.4}
$\mathfrak{osp}(U \oplus U^\ast)_0 \cap \mathrm{End}(W_\mathbb{R})
\simeq \mathfrak{o}( (U_1^{\vphantom{\ast}} \oplus
U_1^\ast)_\mathbb{R} ) \oplus \mathfrak{sp} ((U_0^{\vphantom{\ast}}
\oplus U_0^\ast)_\mathbb{R})$.
\end{proposition}
\begin{proof}
The intersection is computed by transferring the conditions
$\mathsf{D} = \overline{\mathsf{A}}$ and $\mathsf{C} = \overline{
\mathsf{B}}$ to the level of $\mathfrak{osp}(U \oplus U^\ast)_0\,$.
Of course $\mathsf{D} = \overline{\mathsf{A}}$ just reduces to the
corresponding condition $\mathsf{d} = \overline{\mathsf{a}}\,$.
Because the isometry $\psi : \, \mathbb{C}^N \to (\mathbb{C}^N)^\ast$
in the present case is symmetric one has $\overline{\psi^{-1}} =
\psi^\mathrm{t} = + \psi$, so the condition $\mathsf{C} =
\overline{\mathsf{B}}$ transfers to $\mathsf{c} =
\overline{\mathsf{b}}\,$. For the same reason, the parity of the maps
$\mathsf{b}, \mathsf{c}$ is identical to that of $\mathsf{B},
\mathsf{C}$, i.e., $\mathsf{b} \vert_{U_0^\ast \to
U_0^{\vphantom{\ast}}}$ is symmetric, $\mathsf{b} \vert_{U_1^\ast \to
U_1^{\vphantom{\ast}}}$ is skew, and similar for $\mathsf{c}\,$.
Hence, computing the intersection $\mathfrak{osp}(U \oplus U^\ast)_0
\cap \mathrm{End}(W_\mathbb{R})$ amounts to the same as computing
$\mathfrak{osp}(V \oplus V^\ast)_0 \cap \mathrm{End}(W_\mathbb{R})$,
and the statement follows from our previous discussion of the latter
case.
\end{proof}
In the case of the Howe pair $(\mathfrak{osp}(\widetilde{U} \oplus
\widetilde{U}^\ast), \mathfrak{sp}_N)$ the isometry $\psi : \,
\mathbb{C}^N \to (\mathbb{C}^N)^\ast$ is skew, so that $\overline{
\psi^{-1}} = \psi^\mathrm{t} = - \psi\,$. At the same time, the
parity of $\mathsf{b}, \mathsf{c}$ is reversed as compared to
$\mathsf{B}, \mathsf{C}$: now the map $\mathsf{b} \vert_{U_0^\ast \to
U_0^{\vphantom{\ast}}}$ is skew and $\mathsf{b} \vert_{U_1^\ast \to
U_1^{\vphantom{\ast}}}$ is symmetric (and similar for $\mathsf{c})$.
Therefore,
\begin{displaymath}
    \mathfrak{osp}(\widetilde{U}\oplus\widetilde{U}^\ast)_0\cap
    \mathrm{End}(W_{1,\mathbb{R}})\simeq \{ \begin{pmatrix}
    \mathsf{a} &\mathsf{b}\\ - \overline{\mathsf{b}}
    &-\mathsf{a}^\mathrm{t} \end{pmatrix} \in \mathrm{End}
    (U_1^{\vphantom{\ast}} \oplus U_1^\ast) \mid \mathsf{a} = -
    \overline{\mathsf{a}}^\mathrm{t} \;, \, \mathsf{b} = +
    \mathsf{b}^\mathrm{t} \} \;,
\end{displaymath}
which is a compact real form $\mathfrak{usp}(U_1^{\vphantom{\ast}}
\oplus U_1^\ast)$ of $\mathfrak{sp}(U_1 \oplus U_1^\ast)$; and
\begin{displaymath}
    \mathfrak{osp}(\widetilde{U} \oplus \widetilde{U}^\ast)_0 \cap
    \mathrm{End}(W_{0,\mathbb{R}}) \simeq \{ \begin{pmatrix} \mathsf{a}
    &\mathsf{b}\\ - \overline{\mathsf{b}} &-\mathsf{a}^\mathrm{t}
    \end{pmatrix} \in \mathrm{End}(U_0^{\vphantom{\ast}} \oplus U_0^\ast)
    \mid \mathsf{a} = -\overline{\mathsf{a}}^\mathrm{t} \;, \, \mathsf{b}
    = -\mathsf{b}^\mathrm{t} \} \;,
\end{displaymath}
which is a non-compact real form of $\mathfrak{o}
(U_0^{\vphantom{\ast}} \oplus U_0^\ast)$ known as
$\mathfrak{so}^\ast(U_0^{\vphantom{\ast}} \oplus U_0^\ast)$.

Let us summarize this result.
\begin{proposition}\label{prop:2.5}
$\mathfrak{osp}(\widetilde{U} \oplus \widetilde{U}^\ast)_0 \cap
\mathrm{End}(W_\mathbb{R}) \simeq \mathfrak{usp}
(U_1^{\vphantom{\ast}} \oplus U_1^\ast) \oplus
\mathfrak{so}^\ast(U_0^{\vphantom{\ast}} \oplus U_0^\ast)$.
\end{proposition}

\section{Semigroup representation}\label{integrating}
\setcounter{equation}{0}

As before, we identify the complex Lie superalgebra $\mathfrak{g} :=
\mathfrak{osp}(W)$ with the space of super-symmetrized degree-two
elements in $\mathfrak{q}_2(W)$, so that
\begin{displaymath}
    \mathfrak{q}_2(W) = \mathfrak{g} \oplus \mathfrak{q}_1(W) \;,
    \quad \mathfrak{q}_1(W) = W \oplus \mathbb{C}\;.
\end{displaymath}
The adjoint representation of $\mathfrak{g}$ on $\mathfrak{q}(W)$
restricts to the Lie algebra representation of $\mathfrak{g}_0 =
\mathfrak{o}(W_1) \oplus \mathfrak{sp}(W_0)$ on $W = W_1\oplus W_0$
which is just the direct sum of the fundamental representations of
$\mathfrak{o}(W_1)$ and $\mathfrak{sp}(W_0)$. These are integrated by
the fundamental representations of the complex Lie groups
$\mathrm{SO}(W_1)$ and $\mathrm{Sp}(W_0)$ respectively.

Since the Clifford-Weyl algebra $\mathfrak{q}(W)$ is an associative
algebra, one can ask if, given $x \in \mathfrak{g}_0 \subset
\mathfrak{q}(W)$, the exponential series $\mathrm{e}^x$ makes sense.
The existence of a one-parameter group $\mathrm{e}^{tx}$ for $x \in
\mathfrak{g}_0$ would of course imply that
\begin{equation}\label{conjugation}
    \frac{d}{dt} \, \mathrm{e}^{\, tx} \, w \, \mathrm{e}^{-tx}
    \, \Big\vert_{t=0} = \mathrm{ad}(x)\, w \qquad (w \in W)\,.
\end{equation}

Now $\mathfrak{q}(W) = \mathfrak{c}(W_1) \otimes \mathfrak{w}(W_0)$.
Since the Clifford algebra $\mathfrak{c}(W_1)$ is finite-dimensional,
the series $\mathrm{e}^x$ for $x \in \mathfrak{o}(W_1)
\hookrightarrow \mathfrak{g}_0$ does make immediate sense. In this
way one is able to exponentiate the Lie algebra $\mathfrak{o}(W_1)$
in $\mathfrak{c}(W_1)$. The associated complex Lie group, which is
then embedded in $\mathfrak{c}(W_1)$, is $\mathrm{Spin}(W_1)$. This
is a 2:1 cover of the complex orthogonal group $\mathrm{SO}(W_1)$.
Its conjugation representation on $W_1$ as in (\ref{conjugation})
realizes the covering map as a homomorphism $\mathrm{Spin}(W_1) \to
\mathrm{SO}(W_1)$.

Viewing the other summand $\mathfrak{sp}(W_0)$ of $\mathfrak{g}_0$ as
being in the infinite-dimensional Weyl algebra $\mathfrak{w}(W_0)$,
it is definitely not possibly to exponentiate it in such a naive way.
This is in particular due to the fact that for most $x \in \mathfrak
{sp}(W_0)$ the formal series $\mathrm{e}^x$ is not contained in any
space $\mathfrak{w}_n (W_0)$ of the filtration of $\mathfrak{w}
(W_0)$.

As a first step toward remedying this situation, we consider
$\mathfrak{q}(W)$ as a space of densely defined operators on the
completion $\mathcal{A}_V$ (cf.\ $\S$\ref{sect:2.6.3}) of the
spinor-oscillator module $\mathfrak{a}(V)$. Since all difficulties
are on the $W_0$ side, for the remainder of this chapter we simplify
the notation by letting $W := W_0$ and discussing only the oscillator
representation of $\mathfrak{w}(W)$. Recall that this representation
on $\mathfrak{a}(V)$ is defined by multiplication $\mu( \varphi)$ for
$\varphi \in V^*$ and the directional derivative $\delta(v)$ for $v
\in V$.

For $x \in \mathfrak{w}(W)$ there is at least no formal obstruction
to the exponential series of $x$ existing in $\mathrm{End}(\mathcal
A_V)$. However, direct inspection shows that convergence cannot be
expected unless some restrictions are imposed on $x\,$. This is done
by introducing a notion of unitarity and an associated semigroup of
contraction operators.

\subsection {The oscillator semigroup}

Here we introduce the basic semigroup in the complex symplectic
group. Various structures are lifted to its canonical 2:1 covering.
Actions of the real symplectic and metaplectic groups are discussed
along with the role played by the cone of elliptic elements.

\subsubsection{Contraction semigroup: definitions, basic properties}

Letting $\langle \, , \, \rangle$ be the unitary structure on $V$
which was fixed in the previous chapter, we recall the complex
anti-linear bijection $c : \, V \to V^\ast$, $v \mapsto \langle v \,
, \, \rangle\,$. There is an induced map $C : \, W \to W$ on $W = V
\oplus V^\ast$ by $C(v + \varphi) = c^{-1}\varphi + cv\,$. As before,
we put $W_\mathbb{R} := \mathrm{Fix}(C) \subset W$.

Since we have restricted our attention to the symplectic side, the
vector spaces $W$ and $W_\mathbb{R}$ are now equipped with the
standard complex symplectic structure $A$ and real symplectic form
$\omega = \mathrm{i}A$ respectively. From here on in this chapter we
abbreviate the notation by writing $\mathrm{Sp} := \mathrm{Sp}(W)$
and $\mathfrak{sp} := \mathfrak{sp} (W)$. Let an anti-unitary
involution $\sigma : \, \mathrm{Sp} \to \mathrm{Sp}$ be defined by $g
\mapsto C g \, C^{-1}$. Its fixed point group $\mathrm{Fix} (\sigma)$
is the real form $\mathrm{Sp} (W_\mathbb{R}$) of main interest. We
here denote it by $\mathrm{Sp}_\mathbb{R}$ and let
$\mathfrak{sp}_\mathbb{R}$ stand for its Lie algebra.

Given $A$ and $C$, consider the mixed-signature Hermitian structure
\begin{displaymath}
    W \times W \to \mathbb{C} \;, \quad (w\, ,w^\prime) \mapsto
    A(Cw\, ,w^\prime) \;,
\end{displaymath}
which we denote by $A(Cw\, ,w^\prime) =: \langle w \, , w^\prime
\rangle_s\,$, with subscript $s$ to distinguish it from the canonical
Hermitian structure of $W$ given by $\langle v + \varphi , v^\prime +
\varphi^\prime \rangle := \langle v \, , v^\prime \rangle + \langle
c^{-1} \varphi^\prime , c^{-1} \varphi \rangle$. The relation between
the two is
\begin{displaymath}
    \langle w \, , w^\prime \rangle_s = \langle w \, , s w^\prime
    \rangle \;, \qquad s = (- \mathrm{Id}_V) \oplus \mathrm{Id}_{V^*} \;.
\end{displaymath}
Note also the relation
\begin{displaymath}
    \sigma(g) = C g\, C^{-1} = s\, (g^{-1})^\dagger s
    \qquad (g \in \mathrm{Sp}) \;.
\end{displaymath}

Now observe that the real form $\mathrm{Sp}_\mathbb{R}$ is the
subgroup of $\langle \, , \, \rangle_s$-isometries in $\mathrm{Sp}:$
\begin{displaymath}
    \mathrm{Sp}_\mathbb{R} = \{g \in \mathrm{Sp} \mid \forall w \in W :
    \; \langle g w , g w \rangle_s = \langle w , w \rangle_s \} \;.
\end{displaymath}
Then define a semigroup $\mathrm{H}(W^s)$ in $\mathrm{Sp}$ by
\begin{displaymath}
    \mathrm{H}(W^s) := \{ h \in \mathrm{Sp} \mid \forall w \in W ,\;
    w \not= 0 : \;
    \langle h w , h w \rangle_s < \langle w , w \rangle_s \} \,.
\end{displaymath}
Note that the operation $g \mapsto g^\dagger$ of Hermitian
conjugation with respect to $\langle \, ,\, \rangle $ stabilizes
$\mathrm{Sp}$ and that $\mathrm{Sp}_\mathbb{R}$ is defined by the
condition $g^\dagger s g = s\,$. The semigroup $\mathrm{H}(W^s)$ is
defined by $h^\dagger s h < s\,$, or equivalently, $s - h^\dagger s
h$ is positive definite. We will see later that $\mathrm{H}(W^s)$
(or, rather, a $2:1$ cover thereof) acts by contraction operators on
the Hilbert space $\mathcal{A}_V$.

It is immediate that $\mathrm{H}(W^s)$ is an open semigroup in
$\mathrm{Sp}$ with $\mathrm{Sp}_\mathbb{R}$ on its boundary.
Furthermore, $\mathrm{H}(W^s)$ is stabilized by the action of
$\mathrm{Sp}_\mathbb{R} \times \mathrm{Sp}_\mathbb{R}$ by $h\mapsto
g_1 h \, g_2^{-1}$.

The map $\pi :\, \mathrm{Sp} \to \mathrm{Sp}\,$, $h \mapsto h
\sigma(h^{-1})$, will play an important role in our considerations.
It is invariant under the $\mathrm{Sp}_\mathbb{R}$-action by right
multiplication, $\pi(hg^{-1}) = \pi(h)$, and is equivariant with
respect to the action defined by left multiplication on its domain of
definition and conjugation on its image space, $\pi(gh) = g \pi(h)
g^{-1}$. Direct calculation shows that in fact the $\pi$-fibers are
exactly the orbits of the $\mathrm{Sp}_\mathbb{R}$-action by right
multiplication. Observe that if $h = \exp(\mathrm{i}X)$ for $X \in
\mathfrak{sp}_\mathbb{R}\,$, then $\sigma(h) = h^{-1}$ and $\pi(h) =
h^2$. In particular, if $\mathfrak{t}$ is a Cartan subalgebra of
$\mathfrak{sp}$ which is defined over $\mathbb{R}\,$, then $\pi
\vert_{\exp(\mathrm{i} \mathfrak{t}_\mathbb{R})}$ is just the
squaring map $t \mapsto t^2$.

\subsubsection{Actions of $\mathrm{Sp}_\mathbb{R}\,$}

We now fix a Cartan subalgebra $\mathfrak{t}$ having the property
that $T_\mathbb{R} = \exp (\mathfrak{t}_\mathbb{R})$ is contained in
the unitary maximal compact subgroup defined by $\langle \, ,\,
\rangle$ of $\mathrm{Sp} (W_\mathbb{R})$. This means that $T$ acts
diagonally on the decomposition $W = V \oplus V^*$ and there is a
(unique up to order) orthogonal decomposition $V = E_1 \oplus \ldots
\oplus E_d$ into one-dimensional subspaces so that if $F_j :=
c(E_j)$, then $T$ acts via characters $\chi_j$ on the vector spaces
$P_j := E_j \oplus F_j$ by $t(e_j\, ,f_j) = (\chi_j(t) e_j \, ,
\chi_j(t)^{-1} f_j)$.

In other words, we may choose $\{ e_j \}_{j = 1, \ldots, d}$ to be an
orthonormal basis of $V$ and equip $V^*$ with the dual basis so that
the elements $t \in T$ are in diagonal form:
\begin{displaymath}
    t = \mathrm{diag}(\lambda_1^{\vphantom{-1}},\ldots,\lambda_d^{
    \vphantom{-1}}\, ,\lambda_1^{-1}, \ldots , \lambda_d^{-1})\;, \
    \qquad \lambda_j = \chi_j(t) \;.
\end{displaymath}

Observe that, conversely, the elements of $\mathrm{Sp}$ that
stabilize the decomposition $W = P_1 \oplus \ldots \oplus P_d$ and
act diagonally in the above sense, are exactly the elements of $T$.
Moreover, $\exp (\mathrm{i} \mathfrak{t}_\mathbb{R})$ is the subgroup
of elements $t \in T$ with $\chi_j(t) \in \mathbb{R}_+$ for all
$j\,$. Note that the complex symplectic planes $P_j$ are
$A$-orthogonal and defined over $\mathbb{R}$.

We now wish to analyze $\mathrm{H}(W^s)$ via the map $\pi :\, h
\mapsto h \sigma(h^{-1})$. However, for a technical reason related to
the proof of Proposition \ref{M-slice} below, we must begin with the
opposite map, $\pi^\prime : \, h \mapsto \sigma(h^{-1}) h\,$. Thus
let $M := \pi^\prime (\mathrm{H}(W^s))$ and write $\pi^\prime:\,
\mathrm{H} (W^s)\to M\,$. The toral semigroup $T_+ := \exp(\mathrm{i}
\mathfrak{t}_\mathbb{R})\,\cap \, \mathrm{H}(W^s)$ consists of those
elements $t \in \exp( \mathrm{i} \mathfrak{t}_\mathbb{R})$ that act
as contractions on $V^\ast$, i.e., $0 < \chi_j(t)^{-1} < 1$ for all
$j\,$. The restriction $\pi^\prime \vert_{T_+} = \pi \vert_{T_+}$ is,
as indicated above, the squaring map $t \mapsto t^2$; in particular
we have $T_+ \subset M$ and the set $\{ g\, t g^{-1} \mid t \in T_+
\, , \, g \in \mathrm{Sp}_\mathbb{R} \}$ is likewise contained in
$M\,$.

In the sequel, we will often encounter the action of $\mathrm{Sp}_
\mathbb{R}$ on $T_+$ and $M$ by conjugation. We therefore denote this
action by a special name, $\mathrm{Int}(g)\, t := g\, t g^{-1}$.
\begin{proposition}\label{M-slice}
$M = \mathrm{Int}(\mathrm{Sp}_\mathbb{R}) T_+ \,$.
\end{proposition}
\begin{proof}
For $g \in \mathrm{Sp}$ one has $\sigma(g^{-1}) = C g^{-1} C^{-1} = s
g^\dagger s\,$. Hence if $M \ni m = \sigma(h^{-1})h$ with $h \in
\mathrm{H}(W^s)$, then $m = s h^\dagger s h\,$. Consequently,
$\langle w, m w \rangle_s = \langle h w , h w \rangle_s < \langle w ,
w \rangle_s$ for all $w \in W \setminus \{ 0 \}$. In particular,
$\langle w , m w \rangle_s\in \mathbb{R}$ and if $w \not= 0$ is an
$m$-eigenvector with eigenvalue $\lambda$, it follows that $\lambda
\in \mathbb{R}$ and $\lambda \langle w , w \rangle_s < \langle w , w
\rangle_s \not = 0\,$.

Now we have $C h C^{-1} = \sigma(h)$ and hence $C m C^{-1} = m^{-1}$.
As a result, if $w \not= 0$ is an $m$-eigenvector with eigenvalue
$\lambda$, then so is $C w$ with eigenvalue $\lambda^{-1}$. Since $C
s C^{-1} = - s$, the product of $\langle w , w \rangle_s$ with
$\langle Cw , Cw \rangle_s$ is always negative. If $\langle w , w
\rangle_s > 0$ it follows that $\lambda < 1$ and $\lambda^{-1} > 1$;
if $\langle Cw , Cw \rangle_s > 0$ then $\lambda^{-1} < 1$ and
$\lambda > 1$. In both cases $0 < \lambda \not= 1$.

Since $m$ does indeed have at least one eigenvector, we have
constructed a complex 2-plane $Q_1$ as the span of the linearly
independent vectors $w$ and $Cw$. The plane $Q_1$ is defined over
$\mathbb{R}$ and, because $0 \not= \langle w , w \rangle_s =
A(Cw,w)$, it is $A$-nondegenerate. Its $A$-orthogonal complement
$Q_1^\perp$ is therefore also non-degenerate and defined over
$\mathbb{R}$.

The transformation $m \in \mathrm{Sp}$ stabilizes the decomposition
$W = Q_1 \oplus Q_1^\perp\,$. Hence, proceeding by induction we
obtain an $A$-orthogonal decomposition $W = Q_1 \oplus \ldots \oplus
Q_d\,$. Since the $Q_j$ are $m$-invariant symplectic planes defined
over $\mathbb{R}$, there exists $g \in \mathrm{Sp}_\mathbb{R}$ so
that $t := g m g^{-1}$ stabilizes the above $T$-invariant
decomposition $W = P_1 \oplus \ldots \oplus P_d\,$. Exchanging $w$
with $C w$ if necessary, we may assume that $t$ acts diagonally on
$P_j = E_j \oplus F_j$ by $(e_j\, ,f_j) \mapsto (\lambda_j \, e_j \,
, \lambda_j^{-1} f_j)$ with $\lambda_j > 1$. In other words, $t \in
T_+\,$.
\end{proof}
If we let
\begin{displaymath}
    \mathrm{Sp}_\mathbb{R} T_+ \mathrm{Sp}_\mathbb{R} := \{g_1 t g_2^{-1}
    \mid g_1 , g_2 \in \mathrm{Sp}_\mathbb{R} \,, \, t \in T_+ \}\;,
\end{displaymath}
then we now have the following analog of the $KAK$-decomposition.
\begin{corollary}\label{left/right transitive}
The semigroup $\mathrm{H}(W^s)$ decomposes as $\mathrm{H}(W^s) =
\mathrm{Sp}_\mathbb{R} T_+ \mathrm{Sp}_\mathbb{R}\,$. In particular,
$\mathrm{H}(W^s)$ is connected.
\end{corollary}
\begin{proof}
By definition, $\mathrm{H}(W^s) = {\pi^\prime}^{-1}(M)$, and from
Proposition \ref{M-slice} one has $\mathrm{H}(W^s)= {\pi^\prime}^{-1}
(\mathrm{Int}(\mathrm{Sp}_\mathbb{R}) T_+)$. Now the map $\pi^\prime
\vert_{T_+} : \, T_+ \to T_+ \,$, $t\mapsto t^2$ is surjective.
Therefore ${\pi^\prime}^{-1}(T_+) = \mathrm{Sp}_\mathbb{R} T_+\,$,
which is to say that each point in the fiber of $\pi^\prime$ over $t
\in T_+$ lies in the orbit of $\sqrt{t} \in T_+$ generated by left
multiplication with $\mathrm{Sp}_\mathbb{R}\,$. On the other hand, by
the property $g \pi^\prime(h) g^{-1}= \pi^\prime(hg^{-1})$ of
$\mathrm{Sp}_\mathbb{R}$-equivariance we have
\begin{displaymath}
    \mathrm{Int}(\mathrm{Sp}_\mathbb{R}) T_+ = \mathrm{Int}
    (\mathrm{Sp}_\mathbb{R}) \pi^\prime ( {\pi^\prime}^{-1}(T_+) )
    = \pi^\prime ( {\pi^\prime}^{-1}(T_+)\mathrm{Sp}_\mathbb{R} )
\end{displaymath}
and hence $\mathrm{H}(W^s) = {\pi^\prime}^{-1}( \mathrm{Int}
(\mathrm{Sp}_\mathbb{R}) T_+) = {\pi^\prime}^{-1} (T_+) \mathrm
{Sp}_\mathbb{R} = \mathrm{Sp}_\mathbb{R} T_+ \mathrm{Sp}_\mathbb{R}
\,$.

Because $\mathrm{Sp}_\mathbb{R}$ and $T_+$ are connected, so is
$\mathrm{H}(W^s) = \mathrm{Sp}_\mathbb{R} T_+ \mathrm{Sp}_\mathbb{R}
\,$.
\end{proof}
It is clear that $M = \mathrm{Int}( \mathrm{Sp}_\mathbb{R} ) T_+
\subset \mathrm{H}(W^s)$. Furthermore, since both $T_+$ and
$\mathrm{Sp}_\mathbb{R}$ are invariant under the operation of
Hermitian conjugation $h\mapsto h^\dagger$ and under the involution
$h \mapsto s h s$, we have the following consequences.
\begin{corollary}\label{cor:3.2}
$\mathrm{H}(W^s)$ is invariant under $h \mapsto h^\dagger$ and also
under $h \mapsto s h s$. In particular, $\mathrm{H}(W^s)$ is
stabilized by the map $h \mapsto \sigma(h^{-1}) = s h^\dagger s$.
\end{corollary}
\begin{remark}
Letting $h^\prime := \sigma(h)^{-1}$ one has $\pi^\prime(h) = \sigma
(h)^{-1} h = h^\prime \sigma(h^\prime)^{-1} = \pi(h^\prime)$ and
hence $M = \pi^\prime(\mathrm{H}(W^s)) = \pi(\mathrm{H}(W^s))$. The
stability of $\mathrm{H}(W^s)$ under $h \mapsto \sigma(h)^{-1}$ was
not immediate from our definition of $\mathrm{H}(W^s)$, which is why
we have been working from the viewpoint of $\mathrm{H}(W^s) =
{\pi^\prime}^{-1}(M)$ so far. Now that we have it, we may regard
$\mathrm{H}(W^s)$ as the total space of an $\mathrm{Sp}_\mathbb
{R}$-principal bundle $\pi:\, \mathrm{H}(W^s) \to M$. We are going to
see in Corollary \ref{bijective} that this principal bundle is
trivial.
\end{remark}
Next observe that, since $\sigma(m) = m^{-1}$ for $m = \sigma(h)^{-1}
h = h^\prime \sigma(h^\prime)^{-1} \in M\,$, the maps $\pi : \, M \to
M$ and $\pi^\prime : \, M \to M$ coincide and are just the square $m
\mapsto m^2$. Thus the claim that the elements of $M$ have a unique
square root in $M$ can be formulated as follows.
\begin{proposition}\label{squaring}
The squaring map $\pi = \pi^\prime : \, M \to M$ is bijective.
\end{proposition}
\begin{proof}
Recall from Proposition \ref{M-slice} that every $m \in M$ is
diagonalizable in the sense that $M = \mathrm{Int}(
\mathrm{Sp}_\mathbb{R}) T_+\,$. Since $\pi : \, T_+\to T_+$ is
surjective, the surjectivity of $\pi : \, M \to M$ is immediate. For
the injectivity of $\pi$ we note that the $m$-eigenspace with
eigenvalue $\lambda$ is contained in the $m^2$-eigenspace with
eigenvalue $\lambda^2$. The result then follows from the fact that
all eigenvalues of $m$ are positive real numbers.
\end{proof}
\begin{corollary}\label{bijective}
Each of the two $\mathrm{Sp}_\mathbb{R}$-equivariant maps
$\mathrm{Sp}_\mathbb{R} \times M \to \mathrm{H} (W^s)$ defined by
$(g,m) \mapsto gm$ and by $(g,m) \mapsto m g^{-1}$, is a bijection.
\end{corollary}
\begin{proof}
Consider the map $(g,m) \mapsto gm\,$. Surjectivity is evident from
Corollary \ref{left/right transitive} and $M = \mathrm{Int}(\mathrm
{Sp}_\mathbb{R}) T_+\,$. For the injectivity it suffices to prove
that if $m_1 , m_2 \in M$ and $g \in \mathrm{Sp}_\mathbb{R}$ with $g
m_1 = m_2\,$, then $m_1 = m_2\,$. But this follows directly from
$\pi^\prime(m_1) = \pi^\prime(g m_1) = \pi^\prime (m_2)$ and the fact
that $\pi^\prime\vert_M$ is the bijective squaring map.

The proof for the map $(g,m) \mapsto m g^{-1}$ is similar, with $\pi$
replacing $\pi^\prime$.
\end{proof}

\subsubsection{Cone realization of $M$}

Let us look more carefully at $M$ as a geometric object. First, as we
have seen, the elements $m$ of $M$ satisfy the condition $m = \sigma
(m^{-1})$. We regard $\psi :\, \mathrm{H}(W^s) \to \mathrm{H}(W^s)
\,$, $h \mapsto \sigma(h^{-1})$, as an anti-holomorphic involution
and reformulate this condition as $M \subset \mathrm{Fix} (\psi)$. In
the present section we are going to show that $M$ is a closed,
connected, real-analytic submanifold of $\mathrm{H}(W^s)$ which
locally agrees with $\mathrm{Fix} (\psi)$. This implies in particular
that $M$ is totally real in $\mathrm{H}(W^s)$ with $\dim_\mathbb{R} M
= \dim_\mathbb{C} \mathrm{H}(W^s)$. We will also show that the
exponential map identifies $M$ with a precisely defined open cone in
$\mathrm{i} \mathfrak{sp}_\mathbb{R}\,$. We begin with the following
statement.
\begin{lemma}\label{lem:M-closed}
The image $M$ of $\pi$ is closed as a subset of $\mathrm{H}(W^s)$.
\end{lemma}
\begin{proof}
Let $h \in \mathrm{cl}(M) \subset \mathrm{H} (W^s)$. By the
definition of $M$, we still have $h \sigma(h)^{-1} =: m \in M\,$. If
$h_n$ is any sequence from $M$ with $h_n \to h\,$, then $h_n
\sigma(h_n)^{-1} \to m\,$. But $m$ has a unique square root $\sqrt{m}
\in M$ and $h_n = \sigma(h_n)^{-1} \to \sqrt{m} \,$. Hence $h =
\sqrt{m} \in M\,$.
\end{proof}
\begin{remark}
$M$ of course fails to be closed as a subset of $\mathrm{Sp}\,$. For
example, $g = \mathrm{Id}$ is in the closure of $M \subset
\mathrm{Sp}$ but is not in $M\,$.
\end{remark}
\begin{lemma}\label{lem:max-rank}
The exponential map $\exp : \, \mathfrak{sp} \to \mathrm{Sp}$ has
maximal rank along $\mathfrak{t}_+\,$.
\end{lemma}
\begin {proof}
We are going to use the fact that the squaring map $S : \,
\mathrm{Sp} \to \mathrm{Sp}\,$, $g \mapsto g^2$, is a local
diffeomorphism of $\mathrm{Sp}$ at any point $t \in T_+\,$. To show
this, we compute the differential of $S$ at $t$ and obtain
\begin{displaymath}
    D_t S = dL_{\,t^2} \circ (\mathrm{Id}_\mathfrak{sp} +
    \mathrm{Ad}(t^{-1})) \circ dL_{\,t^{-1}} \;,
\end{displaymath}
where $dL_g$ denotes the differential of the left translation $L_g :
\,\mathrm{Sp}\to \mathrm{Sp}\,$, $g_1 \mapsto g g_1\,$. The middle
map $\mathrm{Id}_\mathfrak{sp} + \mathrm{Ad}(t^{-1}) : \, \mathfrak
{sp} \to \mathfrak{sp}$ is regular because all of the eigenvalues of
$t \in T_+$ are positive real numbers. Since $dL_{\, t^{-1}} : \, T_t
\mathrm{Sp} \to \mathfrak{sp}$ and $dL_{\,t^2} : \, \mathfrak{sp} \to
T_{t^2} \mathrm{Sp}$ are iso\-morphisms, it follows that $D_t S :\,
T_t \mathrm{Sp} \to T_{t^2} \mathrm{Sp}$ is an isomorphism.

Turning to the proof of the lemma, given $\xi \in \mathfrak{t}_+$ we
now choose $n \in \mathbb{N}$ so that $2^{-n} \xi$ is in a
neighborhood of $0 \in \mathfrak{sp}$ where $\exp$ is a
diffeomorphism. It follows that for $U$ a sufficiently small
neighborhood of $\xi\,$, the exponential map expressed as
\begin{displaymath}
    U \ni \eta \mapsto 2^{-n} \eta \mapsto \exp(2^{-n}\eta)
    \stackrel{S^n}{\mapsto}(\exp(2^{-n}\eta))^{2^n} = \exp(\eta)
\end{displaymath}
is a diffeomorphism of $U$ onto its image.
\end{proof}
Now recall $M = \mathrm{Int}(\mathrm{Sp}_\mathbb{R}) T_+$ and consider
the cone
\begin{displaymath}
    \mathcal{C} := \mathrm{Ad}(\mathrm{Sp}_\mathbb{R}) \mathfrak{t}_+
    \subset \mathrm{i}\mathfrak{sp}_\mathbb{R} \;.
\end{displaymath}
It follows from the equivariance of $\exp$, i.e., $\exp(\mathrm{Ad}
(g) \xi) = \mathrm{Int}(g) \exp(\xi)$, that $\exp : \, \mathcal{C}
\to \mathrm{Int}(\mathrm{Sp}_\mathbb{R}) T_+ = M$ is surjective.
Furthermore, $\exp \vert_{\mathfrak{t}_+} : \, \mathfrak{t}_+ \to
T_+$ is injective and for every $\xi \in \mathfrak{t}_+$ the isotropy
groups of the $\mathrm{Sp}_\mathbb{R}$-actions at $\xi$ and $\exp
(\xi)$ are the same. Therefore $\exp : \, \mathcal{C} \to M$ is also
injective.

In fact, much stronger regularity holds. For the statement of this
result we recall the anti-holomorphic involution $\psi:\, \mathrm{H}
(W^s) \to \mathrm{H}(W^s)$ defined by $h \mapsto \sigma(h^{-1})$ and
let $\mathrm{Fix}(\psi)^0$ denote the connected component of
$\mathrm{Fix} (\psi)$ that contains $M$.
\begin{proposition}
The image $M \subset \mathrm{H}(W^s)$ of $\pi : \, h \mapsto h
\sigma(h^{-1})$ is the closed, connected, totally real submanifold
$\mathrm{Fix}(\psi)^0$, which is half-dimensional in the sense that
$\dim_\mathbb{R} M = \dim_\mathbb{C} \mathrm{H} (W^s)$. The set
$\mathcal{C} = \mathrm{Ad}(\mathrm{Sp}_\mathbb{R}) \mathfrak{t}_+\,$,
which is an open positive cone in $\mathrm{i} \mathfrak{sp}_\mathbb
{R}\,$, is in bijection with $M$ by the real-analytic diffeomorphism
$\exp : \, \mathcal{C} \to M$.
\end{proposition}
\begin{proof}
Lemma \ref{lem:max-rank} implies that $\mathfrak{t}_+$ possesses an
open neighborhood $U$ in $\mathrm{i} \mathfrak{sp}_\mathbb{R}$ so
that $\exp \vert_U$ is everywhere of maximal rank. Because $T_+$ lies
in $\mathrm{H}( W^s)$ and $\mathrm{H}(W^s)$ is open in $\mathrm{Sp}
\,$, by choosing $U$ small enough we may assume that $\exp(\frac{1}
{2} U) \subset \mathrm {H} (W^s)$ and therefore that $\exp(U) \subset
M = \exp (\mathcal{C})$. Since $\exp$ is a local diffeomorphism on
$U$, we may also assume that $U \subset \mathcal{C} = \mathrm{Ad}
(\mathrm{Sp}_\mathbb{R}) \mathfrak{t}_+\,$, and it then follows that
$\mathcal{C} = \mathrm{Ad}(\mathrm{Sp}_\mathbb{R}) U$. In particular,
this shows that $\mathcal{C}$ is open in $\mathrm{i}\mathfrak
{sp}_\mathbb{R}\,$. In summary,
\begin{displaymath}
  \mathcal{C} = \mathrm{Ad}(\mathrm{Sp}_\mathbb{R}) U
  \stackrel{\exp}{\longrightarrow}
  \mathrm{Int}(\mathrm{Sp}_\mathbb{R}) \exp(U) = M \subset
  \mathrm{Fix}(\psi)^0 \;.
\end{displaymath}
By the equivariance of $\exp\,$, we also know that it is everywhere
of maximal rank on $\mathcal{C}\,$.

Now $\psi$ is an anti-holomorphic involution. Therefore, $\mathrm
{Fix}(\psi)^0$ is a totally real, half-dimensional closed submanifold
of $\mathrm{H}(W^s)$, and since $\mathcal{C}$ is open in $\mathrm{i}
\mathfrak{sp}_\mathbb{R}\,$, we also know that $\dim_\mathbb{C}
\mathcal{C} = \dim_\mathbb{R} \mathrm{Fix}(\psi)^0$. The maximal rank
property of $\exp$ then implies that $M = \mathrm{im}(\exp :
\mathcal{C} \to \mathrm{Fix} (\psi)^0)$ is open in $\mathrm{Fix}
(\psi)^0$. In Lemma \ref{lem:M-closed} it was shown that $M$ is
closed in $\mathrm{H} (W^s)$. Thus it is open and closed in the
connected manifold $\mathrm{Fix}(\psi)^0$, and consequently $\exp :
\, \mathcal{C} \to M = \mathrm{Fix}(\psi)^0$ is a local
diffeomorphism of manifolds. Since we already know that $\exp : \,
\mathcal{C} \to M$ is bijective, the desired result follows.
\end{proof}
\begin{corollary}\label{semigroup product structure}
The two identifications $\mathrm{Sp}_\mathbb{R} \times M =
\mathrm{H}(W^s)$ defined by $(g,m) \mapsto gm$ and $(g,m) \mapsto m
g^{-1}$ are real-analytic diffeomorphisms. The fundamental group of
$\mathrm{H}(W^s)$ is isomorphic to $\pi_1(\mathrm{Sp}_\mathbb{R})
\simeq \mathbb{Z}\,$.
\end{corollary}
\begin{proof}
The first statement is proved by explicitly constructing a smooth
inverse to each of the two maps. For this let $g^\prime m^\prime = m
g^{-1} = h \in \mathrm{H}(W^s)$ and note that $m = \sqrt{\pi(h)}$ and
$m^\prime = \sqrt{\pi^\prime (h)}$. Since the square root is a smooth
map on $M\,$, a smooth inverse in the two cases is defined by
\begin{displaymath}
    h \mapsto ( h \sqrt{\pi^\prime(h)}^{\ -1} , \sqrt{\pi^\prime(h)})
    \;, \quad \text{resp.} \quad
    h \mapsto ( h^{-1} \sqrt{\pi(h)} , \sqrt{\pi(h)} ) \;.
\end{displaymath}
The second statement follows from $\mathcal{C} \simeq M$ by $\exp$
and the fact that $\mathrm{Sp}_\mathbb{R}$ is a product of a cell and
a maximal compact subgroup $K\,$. We choose $K$ to be the unitary
group $\mathrm{U} = \mathrm{U}(V, \langle \, ,\,\rangle)$ acting
diagonally on $W = V \oplus V^*$ and recall that $\pi_1(\mathrm{U})
\simeq \mathbb{Z}\,$.
\end{proof}

\subsection{Oscillator semigroup and metaplectic group}
\label{oscillator semigroup}

Recall that we are concerned with the Lie algebra representation of
$\mathfrak{sp}_\mathbb{R} \subset \mathfrak{sp}$ which is defined by
the identification of $\mathfrak{sp}$ with the set of symmetrized
elements of degree two in the Weyl algebra $\mathfrak{w}(W)$ and the
representation of $\mathfrak{w}(W)$ on $\mathfrak{a}(V)$. In
$\S$\ref{metaplectic rep} we construct the oscillator representation
of the metaplectic group $\mathrm{Mp}\,$, a 2:1 cover of
$\mathrm{Sp}_\mathbb{R}\,$, which integrates this Lie algebra
representation. Observe that since $\pi_1(\mathrm{Sp}_\mathbb{R})
\simeq \mathbb{Z}$ and $\mathbb{Z}$ has only one subgroup of index
two, there is a unique such covering $\tau : \, \mathrm{Mp} \to
\mathrm{Sp}_\mathbb{R}\,$. The method of construction \cite{H2} we
use first yields a re\-presentation of the 2:1 covering space
$\widetilde{\mathrm{H}} (W^s)$ of $\mathrm{H}(W^s)$ and then realizes
the oscillator representation of $\mathrm{Mp}$ by taking limits that
correspond to going to $\mathrm{Sp}_\mathbb{R}$ in the boundary of
$\mathrm{H}(W^s)$. This representation of the \emph{oscillator
semigroup} $\widetilde{\mathrm{H}}(W^s)$ is for our purposes at least
as important as the representation of $\mathrm{Mp}\,$.

The goal of the present section is to lift all essential structures
on $\mathrm{H}(W^s)$ to $\widetilde{\mathrm{H}}(W^s)$.

\subsubsection{Lifting the semigroup}

We begin by recalling a few basic facts about covering spaces. If $G$
is a connected Lie group, its universal covering space $U$ carries a
canonical group structure: an element $u \in U$ in the fiber over $g
\in G$ is a homotopy class $u \equiv [\alpha_g]$ of paths $\alpha_g :
\, [0,1] \to G$ connecting $g$ with the neutral element $e \in G\,$;
and an associative product $U \times U \to U$, $(u_1,u_2) \mapsto u_1
u_2\,$, is defined by taking $u_1 u_2$ to be the unique homotopy
class which is given by pointwise multiplication of any two paths
representing the homotopy classes $u_1, u_2\,$. The fundamental group
$\pi_1(G) \equiv \pi_1(G,e)$ acts on $U$ by monodromy, i.e., if
$[\alpha_g] = u \in U$ and $[c] = \gamma \in \pi_1(G)$, then one sets
$\gamma(u) := [ \alpha_g \ast c] \in U$ where $\alpha_g \ast c$ is
the path from $g$ to $e$ which is obtained by composing the path
$\alpha_g$ with the loop $c$ based at $e\,$. This $\pi_1(G)$-action
satisfies the compatibility condition $\gamma_1 (u_1) \gamma_2(u_2) =
(\gamma_1 \gamma_2)(u_1 u_2)$ and in that sense is central.

The situation for our semigroup $\mathrm{H} (W^s)$ is analogous
except for the minor complication that the distinguished point $e =
\mathrm{Id}$ does not lie in $\mathrm{H}(W^s)$ but lies in the
closure of $\mathrm{H}(W^s)$. Hence, by the same principles, the
universal cover $U$ of $\mathrm{H}(W^s)$ comes with a product
operation and there is a central action of $\pi_1(\mathrm{H} (W^s))$
on $U$. Moreover, the product $U \times U \to U$ still is
associative. To see this, first notice that the subsemigroup $T_+
\subset \mathrm{H}(W^s)$ is simply connected and as such is
canonically embedded in $U$. Then for $u_1, u_2, u_3 \in U$ observe
that $u_1 (u_2 u_3) = \gamma ((u_1 u_2) u_3)$ where $\gamma \in \pi_1
(\mathrm{H}(W^s))$ could theoretically depend on the $u_j\,$.
However, any such dependence has to be continuous and the fundamental
group is discrete, so in fact $\gamma$ is independent of the $u_j$
and, since $\gamma$ is the identity when the $u_j$ are in $T_+$
(lifted to $U$), the associativity follows.

Let now $\Gamma \simeq 2\mathbb{Z}$ denote the subgroup of index two in
$\pi_1 (\mathrm{H}(W^s)) \simeq \mathbb{Z}$ and consider $\widetilde{
\mathrm{H}} (W^s) := U/\Gamma$, which is our object of interest.
Since the $\Gamma$-action on $U$ is central, i.e., $\gamma_1 (u_1)
\gamma_2 (u_2) = (\gamma_1 \gamma_2) (u_1 u_2)$ for all $\gamma_1 ,
\gamma_2 \in \Gamma$ and $u_1 , u_2 \in U$, the product $U \times U
\to U$ descends to a product $U/\Gamma \times U/\Gamma \to U/\Gamma
\,$. Thus $U/ \Gamma = \widetilde{\mathrm{H}} (W^s)$ is a semigroup,
and the situation at hand is summarized by the following statement.
\begin{proposition}
The 2:1 covering $\tau_H : \, U / \Gamma = \widetilde{\mathrm{H}}
(W^s) \to \mathrm{H}(W^s)$, $[\alpha_h]\, \Gamma \mapsto h\,$, is a
homomorphism of semigroups.
\end{proposition}

\subsubsection{Actions of the metaplectic group}

Recall that we have two 2:1 coverings: a homomorphism of groups
$\tau:\, \mathrm{Mp} \to \mathrm{Sp}_\mathbb{R}\,$, along with a
homomorphism of semigroups $\tau_H : \, \widetilde{\mathrm{H}}(W^s)
\to \mathrm{H}(W^s)$. Now, a pair of elements $(g^\prime,g) \in
\mathrm{Sp}_\mathbb{R} \times \mathrm{Sp}_\mathbb{R}$ determines a
transformation $h \mapsto g^\prime h g^{-1}$ of $\mathrm{H} (W^s)$,
and by the homotopy lifting property of covering maps a corresponding
action of $\mathrm{Mp} \times \mathrm{Mp}$ on $\widetilde{
\mathrm{H}}(W^s)$ is obtained as follows.

Consider the canonical mapping $\mathrm{Mp} \times M \to
\mathrm{Sp}_\mathbb{R} \times M$ given by $\tau\,$. By the
real-ana\-lytic diffeomorphism $\mathrm{Sp}_\mathbb{R} \times M \to
\mathrm{H} (W^s)$, $(g,m) \mapsto g m\,$, this map can be regarded as
a 2:1 covering of $\mathrm{H}(W^s)$, and since any two 2:1 coverings
are isomorphic we get an identification of the covering space
$\widetilde{\mathrm{H}}(W^s)$ with $\mathrm{Mp} \times M\,$.
Moreover, the action of the group $\mathrm{Mp}$ on itself by left
translations induces on $\mathrm{Mp} \times M \simeq
\widetilde{\mathrm{ H}} (W^s)$ an $\mathrm{Mp}$-action which, by
construction, satisfies the relation
\begin{displaymath}
    \tau_H(g \cdot h) = \tau(g) \tau_H(h) \quad (g \in \mathrm{Mp}
    \, , \,\, h \in \widetilde{\mathrm{H}}(W^s) ) \;.
\end{displaymath}
This can be viewed as a statement of $\mathrm{Mp}$-equivariance of the
covering map $\tau_H\,$.

Now, we have another real-analytic diffeomorphism
$\mathrm{Sp}_\mathbb{R} \times M \to \mathrm{H}(W^s)$ by $(g,m)
\mapsto m g^{-1}$, which transfers left translation in
$\mathrm{Sp}_\mathbb{R}$ to right multiplication on
$\mathrm{H}(W^s)$, and by using it we can repeat the above
construction. The result is another identification $\widetilde{
\mathrm{H}}(W^s) \simeq \mathrm{Mp} \times M$ and another
$\mathrm{Mp}$-action on $\widetilde{\mathrm{H}}(W^s)$. Altogether we
then have two actions of $\mathrm{Mp}$ on $\widetilde{\mathrm{H}}
(W^s)$. The essence of the next statement is that they commute.
\begin{proposition}\label{lifting the product action}
There is a real-analytic action $(g_1,g_2) \mapsto (g_1,g_2) \cdot h$
of $\, \mathrm{Mp} \times \mathrm{Mp}$ on $\widetilde{\mathrm{H}}
(W^s)$ such that the covering $\tau_H : \, \widetilde{\mathrm{H}}
(W^s) \to \mathrm{H} (W^s)$ is $(\mathrm{Mp} \times
\mathrm{Mp})$-equivariant:
\begin{displaymath}
    \tau_H( (g_1,g_2) \cdot h) = \tau(g_1) \tau_H(h) \tau(g_2)^{-1}\;.
\end{displaymath}
\end{proposition}
\begin{proof}
By construction, the stated equivariance property of $\tau_H$ holds
for each of the two actions of $\mathrm{Mp}$ separately. It then
follows that it holds for all $(g_1,g_2) \in \mathrm{Mp} \times
\mathrm{Mp}$ if the two actions commute. But by $\tau_H( (g_1,e)
\cdot h) = \tau(g_1) \tau_H(h)$ and $\tau_H( (e,g_2) \cdot h) =
\tau_H(h) \tau(g_2)^{-1}$ the commutator
\begin{displaymath}
    g := (g_1,e) (e,g_2) (g_1,e)^{-1} (e,g_2)^{-1}
\end{displaymath}
acts trivially on $\mathrm{H}(W^s)$ by $\tau_H\,$, i.e., $\tau_H(g
\cdot h) = \tau_H(h)$. Therefore $g$ can be regarded as being in the
covering group $\tau^{-1}(\mathrm{Id}) = \mathbb{Z}_2$ of the
covering $\tau : \, \mathrm{Mp} \to \mathrm{Sp}_\mathbb{R}\,$. Since
we can connect both $g_1$ and $g_2$ to the identity $e \in
\mathrm{Mp}$ by a continuous curve, it follows from the discreteness
of $\mathbb{Z}_2$ that $g \in \mathrm{Mp} \times \mathrm{Mp}$ acts
trivially on $\widetilde{\mathrm{H}}(W^s)$.
\end{proof}
Notice that since the submanifold $M \subset \mathrm{H} (W^s)$ is
simply connected, there exists a can\-onical lifting of $M$ (which we
still denote by $M$) to the cover $\widetilde{\mathrm{H}}(W^s)$; this
is the unique lifting by which $T_+ \subset M$ is embedded as a
subsemigroup in $\widetilde{\mathrm{H}}(W^s)$. Proposition
\ref{lifting the product action} then allows us to write $\widetilde{
\mathrm{H}}(W^s) = \mathrm{Mp}.M.\mathrm{Mp}$.

\subsubsection{Lifting involutions}

Let us now turn to the issue of lifting the various involutions at
hand. As a first remark, we observe that any Lie group automorphism
$\varphi: \, \mathrm{Sp}_\mathbb{R} \to \mathrm{Sp}_\mathbb{R}$
uniquely lifts to a Lie group automorphism $\widetilde{\varphi}$ of
the universal covering group $\widetilde{\mathrm{Sp}}_\mathbb{R}\,$,
and the latter induces an auto\-morphism of the fundamental group
$\pi_1(\mathrm{Sp}_\mathbb{R}) \simeq \mathbb{Z}$ viewed as a
subgroup of the center of $\widetilde{\mathrm{Sp}}_\mathbb{R}\,$. Now
$\mathrm{Aut}( \pi_1( \mathrm{Sp}_\mathbb{R} )) \simeq \mathrm{Aut}
(\mathbb{Z}) \simeq \mathbb{Z}_2$ and both elements of this
automorphism group stabilize the subgroup $\Gamma \simeq 2\mathbb{Z}$
in $\pi_1(\mathrm{Sp}_\mathbb{R})$. Therefore $\widetilde{ \varphi}$
induces an automorphism of $\mathrm{Mp} = \widetilde{\mathrm{Sp}
}_\mathbb{R} / \Gamma\,$.

Since the operation $h\mapsto h^{-1}$ canonically lifts from
$\mathrm{Sp}_\mathbb{R}$ to $\mathrm{Mp}$ and $h \mapsto (h^{-1}
)^\dagger$ is a Lie group automorphism of $\mathrm{Sp}_\mathbb{R}\,$,
it follows that Hermitian conjugation $h\mapsto h^\dagger$ has a
natural lift to $\mathrm{Mp}\,$. The same goes for the Lie group
automorphism $h \mapsto s h s$ of $\mathrm{Sp}_\mathbb{R}\,$.
\begin{proposition}
Hermitian conjugation $h \mapsto h^\dagger$ and the involution $h
\mapsto s h s$ lift to unique maps with the property that they
stabilize the lifted manifold $M$. In particular, the basic
anti-holomorphic map $\psi:\, \mathrm{H}(W^s) \to \mathrm{H}(W^s)$,
$h \mapsto \sigma (h^{-1})= s h^\dagger s$ lifts to an
anti-holomorphic map $\widetilde{\psi}: \, \widetilde{\mathrm{H}}
(W^s)\to \widetilde{\mathrm {H}}(W^s)$ which is the identity on $M$
and $\mathrm{Mp} \times \mathrm{Mp}$-equivariant in that
$\widetilde{\psi} (g_1 x\, g_2^{-1}) = g_2 \widetilde{\psi}(x)
g_1^{-1}$ for all $g_1, g_2\in \mathrm{Mp}$ and $x \in
\widetilde{\mathrm{H}}(W^s)$.
\end{proposition}
\begin{proof}
Recall that the simply connected space $M \subset \mathrm{H} (W^s)$
has a canonical lifting (still denoted by $M$) to $\widetilde{
\mathrm{H}}(W^s)$. Since all of our involutions stabilize $M$ as a
submanifold of $\mathrm{H}(W^s)$, they are canonically defined on the
lifted manifold $M\,$. In particular, the involution $\psi$ on $M$ is
the identity map, and therefore so is the lifted involution
$\widetilde{\psi}$.

Note furthermore that the involution defined by $h \mapsto s h s$ is
holomorphic on $\mathrm{H}(W^s)$ and that the other two are
anti-holomorphic. Now $\widetilde{\mathrm{H}}(W^s)$ is connected and
the lifted version of $M$ is a totally real submanifold of
$\widetilde{\mathrm{H}}(W^s)$ with $\dim_\mathbb{R} M =
\dim_\mathbb{C} \widetilde{ \mathrm{H}}(W^s)$. In such a situation
the identity principle of complex analysis implies that there exists
at most one extension (holomorphic or anti-holomorphic) of an
involution from $M$ to $\widetilde{ \mathrm{H}}(W^s)$. Therefore, it
is enough to prove the existence of extensions.

Since $h \in \widetilde{\mathrm{H}}(W^s)$ is uniquely representable
as $h = g m$ with $g \in \mathrm{Mp}$ and $m \in M$, the involution
$h \mapsto h^\dagger$ is extended by $gm \mapsto (gm)^\dagger =
m^\dagger g^\dagger$. Similarly, $h \mapsto s h s$ extends by $gm
\mapsto (sgs)(sms)$, and $h \mapsto s h^\dagger s$ does so by the
composition of the other two.

The equivariance property of $\widetilde{\psi}$ follows from the fact
that $g \mapsto s g^\dagger s$ on $\mathrm{Mp}$ coincides with the
operation of taking the inverse, $g \mapsto g^{-1}$.
\end{proof}

\subsection{Oscillator semigroup representation}
\label{sec:semigrouprep}

Here we construct the fundamental representation of the semigroup
$\widetilde{\mathrm{H}}(W^s)$ on the Hilbert space $\mathcal{A}_V$,
which in the present context we call Fock space. Our approach is
parallel to that of Howe \cite{H2}: the Fock space we use is related
to the $L^2$-space of Howe's work by the Bargmann transform
\cite{Folland}. (Using the language of physics one would say that
Howe works with the position space wave function while our treatment
relies on the phase space wave function.) In particular, following
Howe we take advantage of a realization of $\mathrm{H} (W^s)$ as the
complement of a certain determinantal variety in the Siegel upper
half plane.

\subsubsection{Cayley transformation}

Let us begin with some background information on the Cayley
transformation, which is defined to be the meromorphic mapping
\begin{displaymath}
    C : \,\, \mathrm{End}(W) \to \mathrm{End}(W), \quad
    g \mapsto \frac{\mathrm{Id}_W + g}{\mathrm{Id}_W - g} \,.
\end{displaymath}
If $g \in \mathrm{Sp}\,$, then from $A(gw , gw^\prime) = A(w ,
w^\prime)$ we have
\begin{displaymath}
    A((\mathrm{Id}_W + g)w \, , (\mathrm{Id}_W - g)w') +
    A((\mathrm{Id}_W - g) w \, , (\mathrm{Id}_W + g)w') = 0
\end{displaymath}
for all $w , w' \in W$. By assuming that $(\mathrm{Id}_W - g)$ is
invertible and then replacing $w$ and $w'$ by $(\mathrm{Id}_W -
g)^{-1}w$ resp.\ $(\mathrm{Id}_W - g)^{-1} w'$, we see that $C$ maps
the complement of the determinantal variety $\{ g \in \mathrm{Sp}
\mid \mathrm{Det}(\mathrm{Id}_W - g) = 0\}$ into $\mathfrak{sp}\,$.

The inverse of the Cayley transformation is given by
\begin{displaymath}
    g = C^{-1}(X) = \frac{X - \mathrm{Id}_W}{X + \mathrm{Id}_W} \,.
\end{displaymath}
Reversing the above argument, one shows that if $(X + \mathrm{Id}_W)$
is invertible and $X \in \mathfrak{sp}\,$, then $C^{-1}(X) \in
\mathrm{Sp}\,$. Moreover, by the relation $X + \mathrm{Id}_W = 2
(\mathrm{Id}_W - g)^{-1}$ for $C(g) = X\,$, if $\mathrm{Id}_W - g$ is
regular, then so is $X + \mathrm{Id}_W\,$, and vice versa. Thus if we
introduce the sets
\begin{displaymath}
    D_{\mathrm{Sp}} := \{g \in \mathrm{Sp} \mid \mathrm{Det}(
    \mathrm{Id}_W - g) = 0 \}\;, \quad D_{\mathfrak{sp}} :=
    \{X \in \mathfrak{sp} \mid \mathrm{Det}(X + \mathrm{Id}_W) = 0 \}\;,
\end{displaymath}
the following is immediate.
\begin{proposition}\label{Cayley biholo}
The Cayley transformation defines a bi-holomorphic map
\begin{displaymath}
    C:\,\, \mathrm{Sp} \setminus D_{\mathrm{Sp}}\to \mathfrak{sp}
    \setminus D_{\mathfrak{sp}}\;.
\end{displaymath}
\end{proposition}
Now we consider the restriction of $C$ to the semigroup $\mathrm{H}
(W^s)$. Letting $\dagger$ be the Hermitian conjugation of the
previous section, denote by $\mathfrak{Re}(X) = \frac{1}{2}(X+
X^\dagger)$ the real part of an operator $X \in \mathrm{End}(W)$ and
define the associated Siegel upper half space $\mathfrak{S}$ to be
the subset of elements $X \in \mathrm{End}(W)$ which are symmetric
with respect to the canonical symmetric bilinear form $S$ on $W = V
\oplus V^*$ with $\mathfrak{Re} (X) > 0\,$. Notice the relations
$S(w,w^\prime) = A(w,sw^\prime)$ and $A(sw,sw^\prime) = -
A(w,w^\prime)$, from which it is seen that $X$ is symmetric if and
only if $sX \in \mathfrak{sp}\,$. Define $D_\mathfrak{S} := \{X\in
\mathfrak{S}\mid \mathrm{Det}(sX + \mathrm{Id}_W) = 0 \}$, let
\begin{displaymath}
    \zeta^+_s := \mathfrak{S} \setminus D_\mathfrak{S} \,,
\end{displaymath}
and define a slightly modified Cayley transformation by
\begin{displaymath}
    a(g) := s \, \frac{\mathrm{Id}_W + g}{\mathrm{Id}_W - g}\,.
\end{displaymath}
Translating Proposition \ref{Cayley biholo}, it follows that $a$
defines a bi-holomorphic map from $\mathrm{Sp} \setminus
D_{\mathrm{Sp}}$ onto the set of $S$-symmetric operators with
$D_\mathfrak{S}$ removed.
\begin{proposition}
The modified Cayley transformation $a : \, \mathrm{Sp} \setminus
D_{\mathrm{Sp}} \to \mathrm{End}(W)$ given by $g \mapsto s \,
(\mathrm{Id}_W + g) (\mathrm{Id}_W - g)^{-1}$ restricts to a
bi-holomorphic map
\begin{displaymath}
    a : \,\, \mathrm{H}(W^s) \to \zeta^+_s\,.
\end{displaymath}
\end{proposition}
This result is an immediate consequence of the following identity.
\begin{lemma}
For $g \in \mathrm{Sp} \setminus D_{\mathrm{Sp}}$ let $a(g) = X$ and
define $Y := (sX + \mathrm{Id}_W)^{-1}$. Then
\begin{displaymath}
    {\textstyle{\frac{1}{2}}} ( s - g^\dagger s g ) =
    Y^\dagger (X + X^\dagger) Y\,.
\end{displaymath}
In particular, one has the following equivalence:
\begin{displaymath}
    \big( \mathfrak{Re}(X) > 0 \;\; \text{and} \;\;
    \mathrm{Det}(sX + \mathrm{Id}_W) \not = 0 \big) ~
    \Leftrightarrow ~ s - g^\dagger s g > 0 \;.
\end{displaymath}
\end{lemma}
\begin{proof}
It is convenient to rewrite $s - g^\dagger s g$ as
\begin{displaymath}
    s - g^\dagger s g = {\textstyle{\frac{1}{2}}}
    (\mathrm{Id}_W - g^\dagger) s (\mathrm{Id}_W + g)+
    {\textstyle{\frac{1}{2}}}(\mathrm{Id}_W + g^\dagger) s
    (\mathrm{Id}_W - g)\,.
\end{displaymath}
Using $a(g) = X$ one directly computes that
\begin{displaymath}
    {\textstyle{\frac{1}{2}}}(\mathrm{Id}_W - g) = (sX +
    \mathrm{Id}_W)^{-1} \quad \text{and} \quad
    {\textstyle{\frac{1}{2}}}(\mathrm{Id}_W + g) = (sX +
    \mathrm{Id}_W)^{-1} sX \,.
\end{displaymath}
The desired identity follows by inserting these relations in the
previous equation.
\end{proof}
\begin{remark}\label{rem:twine-inv}
The modified Cayley transformation intertwines the anti-holomorphic
involution $\psi : \, h \mapsto s h^\dagger s$ with the operation of
taking the Hermitian conjugate $X \mapsto X^\dagger$:
\begin{displaymath}
    a(\psi(h)) = a(h)^\dagger \;.
\end{displaymath}
Since $\zeta_{\,s}^+$ is obviously stable under Hermitian
conjugation, this is another proof of the stability of
$\mathrm{H}(W^s)$ under the involution $\psi\,$; cf.\ Corollary
\ref{cor:3.2}.
\end{remark}

\subsubsection{Construction of the semigroup representation}

Let us now turn to the main goal of this section. Recall that we have
a Lie algebra representation of $\mathfrak{sp}$ on $\mathfrak{a}(V) =
\mathrm{S}(V^*)$ which is defined by its canonical embedding in
$\mathfrak{w}_2(W)$. We shall now construct the corresponding
representation of the semigroup $\widetilde{\mathrm{H}}(W^s)$ on the
Fock space $\mathcal{A}_V$.

It will be seen later that the character of this representation on
the lifted toral semigroup $T_+$ is $\mathrm{Det}^{-\frac{1}{2}}(s -
sh)$. This extends to $M = \mathrm{Int}(\mathrm{Mp}) T_+$ by the
invariance of the character with respect to the conjugation action of
$\mathrm{Mp}\,$. Since $\widetilde{\mathrm{H}}(W^s)$ is connected and
$M$ is totally real of maximal dimension in $\widetilde{\mathrm{H}}
(W^s)$, the identity principle then implies that if a semigroup
representation of $\widetilde{\mathrm{H}}(W^s)$ can be constructed
with a holomorphic character, this character must be given by the
square root function $h \mapsto \mathrm{Det}^{-\frac{1}{2}}(s - sh)$.

There is no difficulty discussing the square root on the simply
connected submanifold $M$. However, in order to make sense of the
square root function on the full semigroup, we must lift all
considerations to $\widetilde{\mathrm{H}}(W^s)$. For convenience of
notation, given $h \in \mathrm{H}(W^s)$ we let $a_h := a(h)$, and for
$x \in \widetilde{\mathrm{H}}(W^s)$ we simply write $a_x \equiv
a(\tau_H(x))$ where $\tau_H: \, \widetilde{\mathrm{H}}(W^s) \to
\mathrm{H}(W^s)$ is the canonical covering map. Then we put
\begin{displaymath}
    f(h) := \mathrm{Det}(a_h + s) = \mathrm{Det}(2s
    (\mathrm{Id}_W - h)^{-1}) \;,
\end{displaymath}
and wish to define $\phi :\, \widetilde{\mathrm{H}}(W^s)\to
\mathbb{C}$ to be the square root of $f$ which agrees with the
positive square root on $T_+\,$. (Here we regard $T_+$ as being in
$\widetilde{\mathrm{H}} (W^s)$ by its canonical lifting as a
subsemigroup of $\mathrm{H} (W^s)$ as in the previous section). This
is possible because $\phi$ is naturally defined on $\{(\xi , \eta)
\in \mathrm{H}(W^s)\times \mathbb{C} \mid f(\xi) = \eta^2 \}$ which
is itself a 2:1 cover of $\mathrm{H} (W^s)$. Since up to equivariant
equivalence there is only one such covering, namely $\widetilde
{\mathrm{H}}(W^s) \to \mathrm{H}(W^s)$, it follows that we may define
$\phi $ on $\widetilde{\mathrm{H}} (W^s)$ as desired. For the
construction of the oscillator representation it is useful to observe
that $\phi$ can be extended to a slightly larger space. This
extension is constructed as follows.

Regard the complex symplectic group $\mathrm{Sp}$ as the total space
of an $\mathrm{Sp}_\mathbb{R}$-principal bundle $\pi : \, \mathrm{Sp}
\to \pi(\mathrm{Sp})$, $g \mapsto g \sigma(g^{-1})$. Recall that the
restricted map $\pi : \, M \to M$ is a diffeomorphism, and that $M$
contains the neutral element $\mathrm{Id} \in \mathrm{Sp}$ in its
boundary. We choose a small ball $B$ centered at $\mathrm{Id}$ in the
base $\pi(\mathrm{Sp})$, and using the fact that $M$ can be
identified with a cone in $\mathrm{i}\mathfrak{sp}_\mathbb{R}$ we
observe that $A := B \, \cup \, M \subset \pi(\mathrm{Sp})$ is
contractible. Now $U := \pi^{-1}(A)$ is diffeomorphic to a product
$\mathrm{Sp}_\mathbb{R} \times A$ and thus comes with a 2:1 covering
$\widetilde{U} \to U$ defined by $\tau: \, \mathrm{Mp} \to
\mathrm{Sp}_\mathbb{R}\,$. The covering space $\widetilde{U}$
contains $\widetilde{ \mathrm{H}}(W^s)$, and is invariant under the
$\mathrm{Mp}$-action by right multiplication. By construction it also
contains the metaplectic group $\mathrm{Mp}\,$, which covers the
group $\mathrm{Sp}_\mathbb{R}$ in $\mathrm{Sp}\,$.

Recall the definition of the determinant variety $D_{\mathrm{Sp}} =
\{g\in \mathrm{Sp} \mid \mathrm{Det}(\mathrm{Id}_W - g) = 0 \}$. Let
$\widetilde{D}_{\mathrm{Sp}}$ denote the set of points in
$\widetilde{U}$ which lie over $D_{\mathrm{Sp}} \cap U$ by the
covering $\widetilde{U} \to U$.
\begin{proposition}\label{Mp in the boundary} There is a unique
continuous extension of $\phi $ from $\widetilde{\mathrm{H}}(W^s)$ to
its closure in $\widetilde{U}$ so that $\phi^2$ agrees with the lift
of $f$ from $U$. The intersection of $\widetilde{D}_{\mathrm{Sp}}$
with any $\mathrm{Mp}$-orbit in $\widetilde{U}$ is nowhere dense in
that orbit and the restriction of the extended function $\phi$ to the
complement of that intersection is real-analytic.
\end{proposition}
\begin{remark}
Before beginning the proof, it should be clarified that at the points
of the lifted determinant variety, i.e., where the lifted square root
$\phi$ of the function $f(g) = \mathrm{Det}(2s\,(\mathrm{Id}_W -
g)^{-1})$ has a pole, \emph{continuity of the extension} means that
the reciprocal of $\phi$ extends to a continuous function which
vanishes on that set.
\end{remark}
\begin{proof}
The intersection of $D_{\mathrm{Sp}}$ with any $\mathrm{Sp}_\mathbb
{R}$-orbit in $U$ is nowhere dense in that orbit; therefore the same
holds for the intersection of $\widetilde{D}_{\mathrm{Sp}}$ with any
$\mathrm{Mp}$-orbit in $\widetilde{U}$.

Let $x \in \widetilde{U} \setminus \widetilde{D}_{\mathrm{Sp}}$ be a
point of the boundary of $\widetilde{\mathrm{H}}(W^s)$. Choose a
local contractible section $\Sigma \subset \widetilde{U}$ of
$\widetilde{U} \to A$ with $x \in \Sigma$ and a neighborhood $\Delta$
of the identity in $\mathrm{Mp}$ so that the map $\Delta \times
\Sigma \to \widetilde{U}\,$, $(g,s) \mapsto s g^{-1}$, realizes
$\Delta \times \Sigma$ as a neighborhood $\widetilde{V}$ of $x$ which
has empty intersection with $\widetilde{D}_{\mathrm{Sp}}\,$. By
construction $\widetilde{V} \cap \widetilde{\mathrm{H}}(W^s)$ is
connected and is itself simply connected. Thus the desired unique
extension of $\phi$ exists on $\widetilde{V}$. At $x$ this extension
is simply defined by taking limits of $\phi$ along arbitrary
sequences $\{ x_n \}$ from $\widetilde{\mathrm{H}}(W^s)$. Thus the
extended function (still called $\phi$) is well-defined on the
closure of $\widetilde{\mathrm {H}}(W^s)$ and is real-analytic on the
complement of $\widetilde{D} _{\mathrm{Sp}}$ in every $\mathrm
{Mp}$-orbit in that closure. It extends as a continuous function on
the full closure of $\widetilde{ \mathrm{H}} (W^s)$ by defining it to
be identically $\infty$ on $\widetilde{D}_{ \mathrm{Sp}}\,$, i.e.,
its reciprocal vanishes identically at these points.
\end{proof}
Now let us proceed with our main objective of defining the semigroup
representation on $\widetilde{\mathrm{H}}(W^s)$. Recall the
involution $\psi : \, \mathrm{H}(W^s) \to \mathrm{H}(W^s)$, $h
\mapsto \sigma(h)^{-1}$, and its lift $\widetilde{\psi}$ to
$\widetilde {\mathrm{H}}(W^s)$. The following will be of use at
several points in the sequel.
\begin{proposition}\label{psi equivariance}
$\phi \circ \widetilde{\psi} = \overline{\phi}\,$.
\end{proposition}
\begin{proof}
By direct calculation, $f \circ \psi = \overline{f}\,$. Thus, since
$f = \phi^2$, we have either $\phi \circ \widetilde{\psi} =
\overline{\phi}$ or $\phi \circ \widetilde{\psi} = - \overline{
\phi}\,$. The latter is not the case, as $\phi$ is not purely
imaginary on the non-empty set $\mathrm{Fix} (\widetilde{\psi})$.
\end{proof}
The semigroup representation $R :\, \widetilde{\mathrm{H}}(W^s) \to
\mathrm{End}(\mathcal{A}_V)$ will be given by a certain averaging
process involving the standard representation of the Heisenberg
group. The latter representation is defined as follows. For elements
$w = v + cv$ of the real vector space $W_\mathbb{R}\,$, the operator
$\delta (v) + \mu(cv)$ is self-adjoint and its exponential
\begin{displaymath}
    T_{v+cv} := \mathrm{e}^{\mathrm{i}\delta(v) + \mathrm{i}\mu(cv)}
\end{displaymath}
converges and is unitary (see, e.g., \cite{RS}). These operators
satisfy the relation
\begin{equation}\label{Heisenberg multiplication}
    T_{w^{\vphantom{\prime}}} T_{w'} = T_{w + w^\prime} \,
    \mathrm{e}^{\frac{\mathrm{i}}{2} \omega (w,w')} \quad
    (w, w^\prime \in W_\mathbb{R}) \;,
\end{equation}
where $\omega := \mathrm{i} A \vert_{W_\mathbb{R}}$ is the induced real
symplectic structure. If $T \mapsto T^\dagger$ denotes the adjoint
operation in $\mathrm{End}(\mathcal{A}_V)$, it follows from
$\delta(v)^\dagger = \mu(cv)$ that
\begin{displaymath}
    T_w^\dagger = T_{-w}^{\vphantom{\dagger}} = T_w^{-1} \quad
    (w \in W_\mathbb{R}) \;.
\end{displaymath}
Thus if $H := W_\mathbb{R} \times \mathrm{U}_1$ is equipped with the
Heisenberg group multiplication law,
\begin{equation*}
    (w,z)\,(w',z') := (w + w', z \, z' \mathrm{e}^{
    \frac{\mathrm{i}}{2} \omega(w,w')}) \;,
\end{equation*}
then $(w,z)\mapsto z\, T_w$ is an irreducible unitary representation
of $H$ on $\mathcal A_V$. It is well known that up to equivariant
isomorphisms there is only one such representation.

The oscillator representation $x \mapsto R(x)$ of $\widetilde{
\mathrm{H}}(W^s)$ is now defined by
\begin{displaymath}
    R(x) = \int_{W_\mathbb{R}} \gamma_x(w) T_w \, \mathrm{dvol}(w) \;,
    \quad \gamma_x(w) := \phi(x) \, \mathrm{e}^{-\frac{1}{4} \langle w
    ,\, a_x w \rangle} \,.
\end{displaymath}
Here $\mathrm{dvol}$ is the Euclidean volume element on
$W_\mathbb{R}$ which we normalize so that
\begin{displaymath}
    \int_{W_\mathbb{R}} \mathrm{e}^{-\frac{1}{4} \langle w , w \rangle}
    \, \mathrm{dvol}(w) = 1 \;.
\end{displaymath}
It should be stressed that we often parameterize $W_\mathbb{R} \simeq
V$ by the map $v \mapsto v + cv = w\,$.

Notice that by the positivity of $\mathfrak{Re}(a_x)$ the Gaussian
function $w \mapsto \gamma_x(w)$ decreases rapidly, so that all
integrals involved in the discussion above and below are easily seen
to converge. In particular, since the unitary operator $T_w$ (for $w
\in W_\mathbb{R}$) has $L^2$-norm $\Vert T_w \Vert = 1$, it follows
for any $x \in \widetilde{\mathrm{H}} (W^s)$ that
\begin{displaymath}
    \Vert R(x) \Vert \le \vert \phi(x) \vert \int_{W_\mathbb{R}}
    \mathrm{e}^{-\frac{1}{4}\langle w , \, \mathfrak{Re}(a_{x})
    w\rangle}\, \mathrm{dvol}(w) =: C(x) \;,
\end{displaymath}
where the bound $C(x)$ by direct computation of the integral is a
finite number:
\begin{equation}\label{eq:Schranke}
    C(x) = \vert \phi(x) \vert \, \mathrm{Det}^{-1/2} \left(
    \mathfrak{Re}(a_x) \right) = 2^{\, \dim_\mathbb{C} V} \left|
    \frac{\mathrm{Det} (\mathrm{Id}_W - h)} {\mathrm{Det}
    (s - h^\dagger s h)} \right|^{1/2}\;, \quad h = \tau_H(x) \;.
\end{equation}
Thus $R(x)$ is a bounded linear operator on $\mathcal{A}_V$. In
Proposition \ref{operator bound} we will establish the uniform bound
$\Vert R(x) \Vert \le C(x) < 1$ for all $x \in M$. It is also clear
that the operator $R(x)$ depends continuously on $x \in \widetilde{
\mathrm{H}}(W^s)$.

The main point now is to prove the semigroup multiplication rule
$R(xy) = R(x)R(y)$. For this we apply the Heisenberg multiplication
formula (\ref{Heisenberg multiplication}) to the inner integral of
\begin{displaymath}
    R(x)R(y) = \int_{W_\mathbb{R}} \left( \int_{W_\mathbb{R}}
    \gamma_x(w-w') \gamma_y(w') T_{w-w'} T_{w'}\, \mathrm{dvol}(w')
    \right)\, \mathrm{dvol}(w)
\end{displaymath}
to see that $R(xy) = R(x)R(y)$ is equivalent to the multiplication
rule $\gamma_{xy} = \gamma_x \sharp \gamma_y$ where the right-hand
side means the twisted convolution
\begin{equation}\label{twisted convolution}
    \gamma_x \sharp \gamma_y (w) := \int _{W_\mathbb{R}} \gamma_x(w-w')
    \gamma_y(w') \, \mathrm{e}^{\frac{\mathrm{i}}{2}\omega (w,w')}
    \, \mathrm{dvol}(w') \;.
\end{equation}
For the proof of the formula $\gamma_x \sharp \gamma_y = \gamma_{xy}
\,$, we will need to know that $\phi$ transforms as
\begin{equation}\label{transformation rule}
    \phi(xy) = \phi(x)\phi(y)\mathrm{Det}^{-\frac{1}{2}}(a_x + a_y)\;.
\end{equation}
This transformation behavior, in turn, follows directly from the
expression for the semigroup multiplication rule transferred to the
upper half-space $\zeta ^+_{\,s}$; we record this expression in the
following statement and refer to \cite{H2}, p.\ 78, for the
calculation.
\begin{proposition}\label{prop:semigroup-law}
Identifying $\mathrm{H}(W^s)$ with $\zeta_{\,s}^+$ by the modified
Cayley transformation and denoting by $(X,Y) \mapsto X \circ Y$ the
semigroup multiplication on $\zeta^+_{\,s}$, one has
\begin{equation}\label{transferred multiplication}
    X \circ Y + s = (Y+s)(X+Y)^{-1}(X+s)
    = X + s - (X-s)(X+Y)^{-1}(X+s) \,.
\end{equation}
\end{proposition}
Given Proposition \ref{prop:semigroup-law}, to prove the
transformation rule (\ref{transformation rule}) just set $X = a_h$
and $Y = a_{h'}$ and note that, since the semigroup multiplication
law for $\widetilde{\mathrm{H}}(W^s)$ by definition yields $a_h \circ
a_{h'} = a_{hh'}\,$, the first expression in (\ref {transferred
multiplication}) implies
\begin{equation}\label{squared transformation law}
    f(hh') = f(h) f(h') \mathrm{Det}^{-1} (a_h + a_{h'}) \;,
\end{equation}
where $f(h) = \mathrm{Det}(a_h + s)$ as above. The transformation
rule for $\phi$ follows by taking the square root of (\ref{squared
transformation law}). As usual, the sign of the square root is
determined by taking the positive sign at points of the lift of $T_+$
in $\widetilde{\mathrm{H}}(W^s)$.

Now we come to the main point.
\begin{proposition}
The twisted convolution for $x , y \in \widetilde{\mathrm{H}}(W^s)$
satisfies $\gamma_x \sharp \gamma_y = \gamma_{xy}$.
\end{proposition}
\begin{proof}
Observe that the phase factor for $w \, , w^\prime \in W_\mathbb{R}$
can be reorganized as
\begin{displaymath}
    \mathrm{e}^{\frac{\mathrm{i}}{2} \omega(w,w^\prime)} =
    \mathrm{e}^{-\frac{1}{2} A(w,w^\prime)} =
    \mathrm{e}^{\frac{1}{4} \langle w^\prime , \, s w \rangle
    - \frac{1}{4} \langle w , \, s w^\prime \rangle}\;.
\end{displaymath}
Inserting the definitions of $\gamma_x$ and $\gamma_y$ in $\gamma_x
\sharp \gamma_y$ we then have
\begin{displaymath}
    \gamma_x \sharp \gamma_y (w) = \phi(x)\phi(y) \, \mathrm{e}^{
    - \frac{1}{4}\langle w, \, a_x w\rangle}\int \mathrm{e}^{-
    \frac{1}{4}\langle w^\prime , (a_x + a_y) w^\prime \rangle +
    \frac{1}{4}\langle w^\prime , (a_x + s) w \rangle + \frac{1}{4}
    \langle w,(a_x - s)w^\prime \rangle}\,\mathrm{dvol}(w^\prime)\;.
\end{displaymath}
Since $\mathfrak{Re} (a_x + a_y) > 0\,$, the integrand is a rapidly
decreasing function of $w^\prime \in W_\mathbb{R}$ and convergence of
the integral over the domain $W_\mathbb{R}$ is guaranteed.

We now wish to explicitly compute the integral by completing the
square and shifting variables. For this it is a useful preparation to
write
\begin{displaymath}
    \langle w^\prime , a_x w \rangle = A(w^\prime , s a_x w)
    \quad (w^\prime \in W_\mathbb{R}) \;,
\end{displaymath}
and similarly for the other terms. We then holomorphically extend the
right-hand side to $w^\prime$ in $W$, and by making a shift of
integration variables
\begin{displaymath}
    w^\prime \to w^\prime + (a_x + a_y)^{-1} (a_x + s) w \;,
\end{displaymath}
we bring the convolution integral into the form
\begin{displaymath}
    \gamma_x \sharp \gamma_y (w) = \mathrm{e}^{-\frac{1}{4}
    A(w,\, s (a_x \circ a_y) w)} \phi(x) \phi(y) \int_{W_\mathbb{R}}
    \mathrm{e}^{-\frac{1}{4} A(w^\prime , \, (s a_x + s a_y) w^\prime)}
    \, \mathrm{dvol}(w^\prime) \;,
\end{displaymath}
where $a_x \circ a_y = -s + (a_y + s)(a_x + a_y)^{-1} (a_x + s) = a_x
- (a_x - s)(a_x + a_y)^{-1}(a_x+s)$ is the semigroup multiplication
on $\zeta_{\,s}^+$. Using $A(w,s w^\prime) = \langle w , w^\prime
\rangle$ for $w \in W_\mathbb{R}$ and the defining relation $a_x
\circ a_y = a_{xy}\,$, we see that the first factor on the right-hand
side is $\mathrm{e}^{- \frac{1}{4} \langle w , \, a_{xy} w \rangle}$.
By the transformation rule (\ref{transformation rule}) the integral
is evaluated as
\begin{displaymath}
    \int_{W_\mathbb{R}} \mathrm{e}^{-\frac{1}{4} \langle w^\prime ,\,
    (a_x + a_y) w^\prime \rangle}\, \mathrm{dvol}(w^\prime) =
    \mathrm{Det}^{-1/2}(a_x+a_y) = \frac{\phi(xy)}{\phi(x)\phi(y)}\;,
\end{displaymath}
and multiplying factors it follows that
\begin{displaymath}
    \gamma_x \sharp \gamma_y(w) = \phi(xy) \, \mathrm{e}^{-
    \frac{1}{4} \langle w ,\, a_{xy} w \rangle} = \gamma_{xy}(w)\;,
\end{displaymath}
which is the desired semigroup property.
\end{proof}
\begin{corollary}
The mapping $R : \, \widetilde{\mathrm{H}}(W^s) \to \mathrm{End}
(\mathcal{A}_V)$ defined by
\begin{displaymath}
    R(x) = \int_{W_\mathbb{R}} \gamma_x(w) T_w \,\mathrm{dvol}(w)
\end{displaymath}
is a representation of the semigroup $\widetilde{\mathrm{H}}(W^s)$.
\end{corollary}
We conclude this section by deriving a formula for the adjoint.
\begin{proposition}
The adjoint of $R(x)$ is computed as $R(x)^\dagger = R(\widetilde{
\psi}(x))$. In particular, $R(x)R(x)^\dagger = R(x\, \widetilde{
\psi}(x))$.
\end{proposition}
\begin{proof}
We recall the relations $\overline{\phi} = \phi \circ
\widetilde{\psi}$ from Proposition \ref{psi equivariance} and
$(a_h)^\dagger = a_{\psi(h)}$ from Remark \ref{rem:twine-inv}. Since
$\overline{\langle w , a_h w \rangle} = \langle w , (a_h)^\dagger w
\rangle\,$, it follows that
\begin{displaymath}
    \overline{\gamma_x} = \gamma_{\widetilde{\psi}(x)} \,.
\end{displaymath}
The desired formula, $R(x)^\dagger = R(\widetilde{\psi}(x))$, now
results from this equation and the identities $T_w^\dagger = T_{-w}$
and $\gamma_x(-w) = \gamma_x(w)$. With this in hand, the second
statement $R(x) R(x)^\dagger = R(x\, \widetilde{\psi} (x))$ is a
consequence of the semigroup property.
\end{proof}

\subsubsection{Basic conjugation formula}

Here we compute the effect of conjugating (in the semigroup sense)
operators of the form $q(w)$, $w \in W$, with operators $R(x)$ coming
from the semigroup. This is an immediate consequence of an analogous
result for the operators $T_w\,$.  For this we first allow $T_w$ to
be defined for $w = v + \varphi \in W$ by
\begin{displaymath}
    T_w := \mathrm{e}^{\mathrm{i} q(w)} = \mathrm{e}^{\mathrm{i}
    \delta(v) + \mathrm{i} \mu(\varphi)} \;.
\end{displaymath}
These operators are no longer defined on Fock space, but are defined
on $\mathcal O(V)$. They satisfy
\begin{equation}\label{Heisenberg transformation rule}
    T_w T_{\tilde w} = T_{w + \tilde w} \, \mathrm{e}^{-
    \frac{1}{2} {A(w,\tilde w)}}\,.
\end{equation}
Note that for $x \in \widetilde{\mathrm{H}}(W^s)$ and $w \in W$ the
operators $R(x) T_w$ and $T_{\tau_H(x) w} R(x)$ are bounded on
$\mathcal{A}_V$.
%
% cmrz:
% The truth of this unsubstantiated claim will not be obvious
% to every reader.
%
Thus we interpret the following as a statement about operators on
that space.
\begin{proposition}\label{Heisenberg conjugation}
For $w \in W$ and $x \in \widetilde{\mathrm{H}}(W^s)$ one has the
relation
\begin{displaymath}
    R(x) T_w = T_{\tau_H(x) w} R(x)\,.
\end{displaymath}
\end{proposition}
\begin{proof}
For convenience of notation we write
\begin{displaymath}
    R(x) = \phi(x) \int_{W_\mathbb{R}} \mathrm{e}^{-\frac{1}{4}
    A(\tilde{w} , \, s a_x \tilde{w})} T_{\tilde{w}}\,
    \mathrm{dvol} (\tilde{w})\,.
\end{displaymath}
Thus
\begin{displaymath}
    R(x) T_w = \phi(x) \int_{W_\mathbb{R}}\mathrm{e}^{-\frac{1}{4}
    A(\tilde{w} ,\,s a_x \tilde{w}) - \frac{1}{2} A(\tilde{w},w)}
    T_{\tilde{w} + w}\, \mathrm{dvol}(\tilde{w})\;.
\end{displaymath}
Now let $h := \tau_H(x)$ and change variables by the translation
$\tilde{w} \mapsto \tilde{w} - w + h w\,$.  Using the definition $s
a_x = (\mathrm{Id}_W + h)(\mathrm{Id}_W - h)^{-1}$ and the relation
\begin{displaymath}
    A((\mathrm{Id}_W - h) w_1 ,(\mathrm{Id}_W + h)w_2) =
    - A((\mathrm{Id}_W + h)w_1, (\mathrm{Id}_W - h)w_2)
\end{displaymath}
for all $w_1 , w_2 \in W$, one simplifies the exponent to obtain
\begin{displaymath}
    R(x) T_{w} = \phi(x) \int_{W_\mathbb{R}}\mathrm{e}^{-\frac{1}{4}
    A(\tilde{w} ,\,s a_x \tilde{w})  -\frac{1}{2} A(h w , \tilde{w})}
    T_{h w + \tilde{w}}\, \mathrm{dvol}(\tilde{w})\;.
\end{displaymath}
Reading (\ref{Heisenberg transformation rule}) backwards one sees
that this expression equals $T_{hw} R(x)$.
\end{proof}
The basic conjugation rule now follows immediately.
\begin{proposition}\label{basic conjugation formula}
For every $x \in \widetilde{\mathrm{H}}(W^s)$ and $w \in W$ it
follows that
\begin{displaymath}
    R(x) q(w) = q(\tau_H(x)w) R(x)\,.
\end{displaymath}
\end{proposition}
\begin{proof}
Apply Proposition \ref{Heisenberg conjugation} for $w$ replaced by
$tw$ and differentiate both sides of the resulting formula at $t=0$.
\end{proof}

\subsubsection{Spectral decomposition and operator bounds}

Numerous properties of $R$ are derived from a precise description of
the spectral decomposition of $R(x)$ for $x \in M\,$. Since every
orbit of $\mathrm{Sp}_\mathbb{R}$ acting by conjugation on $M$ has
non-empty intersection with $T_+\,$, it is important to understand
this decomposition when $x \in T_+\,$. For this we begin with the
case where $V$ is one-dimensional.
\begin{proposition}\label{1-d eigenvalues}
Suppose that $V$ is one-dimensional and that the $T_+$-action on $W =
V \oplus V^*$ is given by $x \cdot (v+\varphi ) = \lambda v + \lambda
^{-1}\varphi$ where $\lambda > 1\,$. If $f$ is a basis vector of
$V^\ast$ then the monomials $\{f^m \}_{m \in \mathbb{N} \cup \{0\}}$
form a basis of $\mathcal A_V$ and one has
\begin{displaymath}
    R(x)f^m = \lambda^{-m-1/2} f^m \;.
\end{displaymath}
\end{proposition}
\begin{proof}
First note that if $w = v + cv\,$, then
\begin{equation}\label{eigenvalue}
    a_x w = s \frac{1 + x}{1 - x} \cdot (v+cv) = -
    \frac{1+\lambda}{1-\lambda}\,v +\frac{1+\lambda^{-1}}{1 -
    \lambda^{-1}}\,cv = \frac{\lambda+1} {\lambda -1}(v+cv) \;.
\end{equation}
Thus the Gaussian function $\gamma_x(w)$ in the present case is
\begin{displaymath}
    \gamma_x (v + cv) = \phi(x)\, \mathrm{e}^{-\frac{1}{2}
    \frac{\lambda +1} {\lambda-1} |v|^2} \;.
\end{displaymath}
To apply the operator $T_{v+cv}$ to $f^m$ we use the description
\begin{displaymath}
    T_{v + cv} = \mathrm{e}^{\mathrm{i}\delta(v) + \mathrm{i}\mu(cv)}
    = \mathrm{e}^{\mathrm{i}\mu(cv)} \mathrm{e}^{-\frac{1}{2} |v|^2}
    \mathrm{e}^{\mathrm{i}\delta (v)}\,.
\end{displaymath}
Decomposing $T_{v+cv}$ in this way is not allowed on the Fock space,
but is allowed if we regard $T_{v+cv}$ as an operator on the full
space $\mathcal{O}(V)$ of holomorphic functions. The calculations are
now carried out on this larger space.

Recall that $\delta(v) f^m = m f(v) f^{m-1}$. From this we obtain the
explicit expression
\begin{displaymath}
    T_{v+cv} f^{m} = \mathrm{e}^{-\frac{1}{2}\vert v\vert ^2}
    \mathrm{e}^{\mathrm{i} \mu(cv)} \sum_{l=0}^m \frac{\mathrm{i}^l}
    {l!}m(m-1)\cdots (m-l+1) f(v)^l f^{m - l}\,.
\end{displaymath}
Our goal is to compute
\begin{displaymath}
    I := \int_V \mathrm{e}^{-\frac{1}{2}\frac{\lambda+1}{\lambda-1}
    |v|^2} T_{v+cv} f^m \, \mathrm{dvol}(v) \;,
\end{displaymath}
where $\mathrm{dvol}(v)$ corresponds to $\mathrm{dvol}(w)$ by the
isomorphism $V \simeq W_\mathbb{R}\,$. Expanding the exponential
$\mathrm{e}^{ \mathrm{i} \mu(cv)}$ and using $\mu(cv) f^m = |v|^2
f(v)^{-1} f^{m+1}$, the integral $I$ is a sum of Gaussian expected
values of terms of the form $|v|^{2k} f(v)^{-k} f(v)^{l}$. The only
terms which survive are those with $k = l\,$. Thus
\begin{displaymath}
\begin{aligned}
    I&= f^m \sum_{k =0}^m (-1)^k \binom{m}{k} \int_V
    \frac{|v|^{2k}}{k!} \, \mathrm{e}^{-\frac{\lambda}{\lambda-1}
    |v|^2}\, \mathrm{dvol}(v) \\ &= 2^{-1} \sum_{k=0}^m (-1)^k
    \binom{m}{k} \bigg(\frac{\lambda -1}{\lambda}\bigg)^{k+1}
    f^m = 2^{-1} (1 - \lambda^{-1})\lambda^{-m} f^m \;.
\end{aligned}
\end{displaymath}
Now $\phi(x)^2 = \mathrm{Det}(a_x+s) = \mathrm{Det}(2s
(\mathrm{Id}_W-\tau_H(x))^{-1})= (-2/ (1-\lambda)) (2/(1 -
\lambda^{-1}))$, and $\phi(x) = 2 \lambda^{-1/2} (1 -
\lambda^{-1})^{-1}$, since we are to take the positive square root at
points $x \in T_+\,$. Hence, $R(x) f^m = \phi(x) 2^{-1} (1-
\lambda^{-1}) \lambda^{-m} f^m = \lambda^{-m-1/2} f^m$ as claimed.
\end{proof}
\begin{remark}
Note that as $x \in T_+$ goes to the unit element (or, equivalently,
$\lambda \to 1$), the expression $R(x) f^m$ converges to $f^m$ in the
strong sense for all $m \in \mathbb{N} \cup \{0\}\,$.
\end{remark}
Now let $V$ be of arbitrary dimension and assume that $x \in T_+$ is
diagonalized on $W = V \oplus V^*$ in a basis $\{ e_1 , \ldots, e_d ,
c e_1,\ldots ,c e_d \}$ with eigenvalues $\lambda_1, \ldots \lambda_d
, \lambda_1^{-1}, \ldots ,\lambda_d^{-1}$ respectively. Since $x \in
T_+$, we have $\lambda_i > 1$ for all $i$. For $f_i := c e_i$ and $m
:= (m_1,\ldots , m_d)$ we employ the standard multi-index notation
$f^m := f_1^{m_1} \cdots f_d^{m_d}$ and $\lambda^m := \lambda_1^{m_1}
\cdots \lambda_d^{m_d}$. In this case the multi-dimensional integrals
split up into products of one-dimensional integrals. Thus, the
following is an immediate consequence of the above.
\begin{corollary}\label{eigenvalues in T_+}
Let $x \in T_+$ be diagonal in a basis $\{ e_i \}$ of $V$ with
eigenvalues $\lambda_i$ ($i = 1, \ldots, d$). If $f^m$ is a monomial
$f^m \equiv \prod_i (c e_i)^{m_i}$, then $R(x)f^m = \lambda^{-m-1/2}
f^m$.
\end{corollary}
One would expect the same result for the spectrum to hold for every
conjugate $g T_+ g^{-1}$, and this expectation is indeed borne out.
However, in the approach we are going to take here, we first need the
existence and basic properties of the oscillator representation of
the metaplectic group. The following is a first step in this
direction.
\begin{proposition}\label{operator bound}
The operator norm function $\mathrm{Mp} \times T_+ \to \mathbb{R}^{>0}$,
$(g,t) \mapsto \Vert R(g\,t g^{-1}) \Vert$ is bounded by a continuous
$\mathrm{Mp}$-independent function $C(t) < 1\,$.
\end{proposition}
\begin{proof}
For any $x \in \widetilde{\mathrm{H}}(W^s)$ we already have the bound
$\Vert R(x) \Vert \le C(x)$ where $C(x)$ was computed in
(\ref{eq:Schranke}). That function $C(x)$ clearly is invariant under
conjugation $x \mapsto g x g^{-1}$ by $g \in \mathrm{Mp}\,$.
Evaluating it for the case of an element $x \equiv t \in T_+$ with
eigenvalues $\lambda_i$ one obtains
\begin{displaymath}
    C(t) = 2^{\, \dim_\mathbb{C} V} \prod\nolimits_i \left(
    \lambda_i^{1/2}+\lambda_i^{-1/2} \right)^{-1} \;.
\end{displaymath}
The inequality $C(t) < 1$ now follows from the fact that $\lambda_i
> 1$ for all $i\,$.
\end{proof}
Since $R(x)^\dagger = R(\widetilde{\psi}(x))$ and $\Vert R(t g)
\Vert^2 = \Vert R(t g)^\dagger R(t g) \Vert = \Vert R(g^{-1} t^2 g)
\Vert$, we infer the following estimates.
\begin{corollary}\label{general operator bound}
For all $t \in T_+$ and $g \in \mathrm{Mp}$ one has $\Vert R(t g) \Vert < 1$
and $\Vert R(gt) \Vert < 1\,$.
\end{corollary}

\subsection{Representation of the metaplectic group}
\label{metaplectic rep}

Recall that we have realized the metaplectic group $\mathrm{Mp}$ in
the boundary of the oscillator semigroup $\widetilde{\mathrm{H}}
(W^s)$ and that $\widetilde{\mathrm{H}}(W^s)$ contains the lifted
manifold $T_+$ in such a way that the neutral element $\mathrm{Id}
\in \mathrm{Mp}$ is in its boundary. Here we show that for $x \in
T_+$ and $g \in \mathrm{Mp}$ the limit $\lim_{x \to \mathrm{Id}}
R(gx)$ is a well-defined unitary operator $R'(g)$ on Fock space and
$R' :\, \mathrm{Mp} \to \mathrm{U} (\mathcal{A}_V)$ is a unitary
representation.  The basic properties of this \emph{oscillator
representation} are then used to derive important facts about the
semigroup representation $R\,$.

Convergence will eventually be discussed in the so-called
\emph{bounded strong* topology} (see \cite{H2}, p.\ 71). For the
moment, however, we shall work with the slightly weaker notion of
bounded strong topology where one only requires uniform boundedness
and pointwise convergence of the operators themselves (with no
mention made of their adjoints). Note that since $\Vert R(gx) \Vert <
1$ by Corollary \ref{general operator bound}, we need only prove the
convergence of $R(gx) f$ on a dense set of functions $f \in
\mathcal{A}_V$. Let us begin with $g = \mathrm{Id}\,$.
\begin{lemma}
If a sequence $x_n \in T_+$ converges to $\mathrm{Id} \in
\mathrm{Mp}\,$, then the sequence $R(x_n)$ converges in the bounded
strong topology to the identity operator on Fock space.
\end{lemma}
\begin{proof}
If $f$ is any $T_+$-eigenfunction, the sequence $R(x_n) f$ converges
to $f$ by the explicit description of the spectrum given in Corollary
\ref{eigenvalues in T_+}. The statement then follows because the
subspace generated by these functions is dense.
\end{proof}
Using this lemma along with the semigroup property, we now show that
the limiting operators exist and are well-defined.
\begin{proposition}\label{well--defined}
If $x_n \in T_+$ converges to $\mathrm{Id} \in \mathrm{Mp}\,$, then
for every $g \in \mathrm{Mp}$ the sequence of operators $R(g x_n)$
converges pointwise, i.e., $R(g x_n) f \to R^\prime(g) f$ for all $f$
in $\mathcal{A}_V$. The limiting operator $R'(g)$ is independent of
the sequence $\{ x_n \}\,$.
\end{proposition}
\begin{proof}
For any $m , n \in \mathbb{N}$ there exists some $t = t(m,n) \in T_+$
sufficiently near the identity so that $\tilde{x}_m = t^{-1} x_m$ and
$\tilde{x}_n = t^{-1} x_n$ are still in $T_+\,$. By the semigroup
property $R(g x_n) = R(g t) R(\tilde{x}_n)$ we then have
\begin{displaymath}
    \Vert R(g x_m)f - R(g x_n)f \Vert \le \Vert R(g t) \Vert
    \, \Vert R(\tilde{x}_m)f - R(\tilde{x}_n)f \Vert\,.
\end{displaymath}
Letting $t = t(m,n)\to\mathrm{Id}$ it follows from Corollary
\ref{general operator bound} that
\begin{displaymath}
    \Vert R(g x_m)f - R(g x_n)f \Vert \le \Vert R(x_m)f - R(x_n)f \Vert\,.
\end{displaymath}
Thus the Cauchy property of $R(x_n)f$ is passed on to $R(g x_n)f$ and
therefore the sequence $R(g x_n)f$ converges in the Hilbert space
$\mathcal{A}_V$. Let $\lim_{n \to \infty} R(g x_n)f =: R^\prime(g)
f$.

To show that the limit is well-defined, pick from $T_+$ another
sequence $y_n \to \mathrm{Id}\,$, let $\lim_{n \to \infty} R(g y_n)f
=: R^{\prime\prime}(g) f\,$, and notice that $\Vert R^\prime(g) f -
R^{\prime\prime}(g)f \Vert$ is no bigger than
\begin{displaymath}
    \Vert R^\prime(g) f - R(g x_n)f \Vert + \Vert R(g x_n) f - R(g
    y_n)f \Vert + \Vert R(g y_n)f - R^{\prime\prime}(g) f \Vert \;.
\end{displaymath}
Using the same reasoning as above, the middle term is estimated as
\begin{displaymath}
    \Vert R(g x_n) f - R(g y_n) f \Vert \le \Vert R(x_n) f - R(y_n) f
    \Vert \le \Vert R(x_n) f - f \Vert + \Vert R(y_n) f - f \Vert \;.
\end{displaymath}
In the limit $n \to \infty$ this yields the desired result $R^\prime
(g) = R^{\prime\prime}(g)$.
\end{proof}
\begin{remark}
Since $\Vert R(g x_n) \Vert < 1$ the sequence $R(g x_n)$ converges to
$R^\prime(g)$ in the bounded strong topology. Such convergence
preserves the product of operators, which is to say that if $A_n \to
A$ and $B_n \to B$, then $A_n B_n \to AB\,$. Indeed,
\begin{displaymath}
    \Vert (A_n B_n - AB) f \Vert \le \Vert A_n (B_n - B)f \Vert +
    \Vert (A_n - A) B f \Vert \;,
\end{displaymath}
and convergence follows from $\Vert A_n \Vert < 1$ and $A_n \to A$,
$B_n \to B$. Note in particular that if $R(g x_n) \to R^\prime(g)$
and $R(g^{-1} x_n) \to R^\prime(g^{-1})$ then $R(g x_n) R(g^{-1} x_n)
\to R^\prime(g) R^\prime(g^{-1})$.
\end{remark}
The bounded strong$^*$ topology also requires pointwise convergence
of the sequence of adjoint operators. Therefore we must also consider
sequences of the form $R(g x_n)^\dagger$. For this (see the proof of
Theorem \ref{thm:Weil-repn} below) we will use the following fact.
\begin{lemma}
Let $\{ A_n \}$ and $\{ B_n \}$ be sequences of bounded operators and
let $C_n := A_n B_n\,$.  If $C_n$ and $B_n$ converge pointwise with
$B_n \to B$ and the sequence $\{ A_n \}$ is uniformly bounded, then
$A_n$ converges pointwise on the image of $B\,$.
\end{lemma}
\begin{proof}
If $f \in \mathrm{im}\, B\,$ then $f = \lim f_n$ where $f_n = B_n h$
for some Hilbert vector $h$. Write
\begin{displaymath}
    (A_m - A_n)f = A_m (f - f_m) + (A_m B_m - A_n B_n)h +
    A_n (f_n - f)
\end{displaymath}
and use the uniform boundedness of $A_n$ to show that $A_n f$
converges.
\end{proof}
Applying this with $A_n = R(g x_n g^{-1})$, $B_n = R(x_n)$ and $C_n =
A_n B_n = R(g x_n)R(g^{-1} x_n)$, we have the following statement
about convergence along the conjugate $g T_+ g^{-1}$.
\begin{proposition}
For $g \in \mathrm{Mp}$ and $\{ x_n \}$ any sequence in $T_+$ with
$x_n \to \mathrm{Id} \in \mathrm{Mp}\,$, it follows that $R(g x_n
g^{-1})$ converges pointwise to $R'(g) R'(g^{-1})$.
\end{proposition}
Next, if we take three sequences in $T_+$ and write
\begin{equation}\label {other side convergence}
    R(x_n y_n z_n) = R(x_n g^{-1}) R(g y_n g^{-1}) R(g z_n) \to
    \mathrm{Id}_{\mathcal{A}_V} \;,
\end{equation}
then it follows that the sequence $R(x_n g^{-1})$ converges to an
operator $B(g^{-1})$ on the image of $R'(g) R'(g^{-1}) R'(g)$ with
$B(g^{-1}) R'(g) R'(g^{-1}) R'(g) = \mathrm{Id}_{\mathcal{A}_V}\,$.
In particular, the operator $R'(g)$ is injective for all $g \in
\mathrm{Mp}\,$. Finally, we define $y_n$ by $y_n^2 = x_n$ and write
$R(g x_n) = R(g y_n g^{-1}) R(g y_n)$. Taking the limit of both sides
of this equation entails that
\begin{equation}\label{last formal step}
    R'(g) = R'(g) R'(g^{-1}) R'(g) \;,
\end{equation}
and since $R'(g)$ is injective, this now allows us to reach the main
goal of this section.
\begin{theorem}\label{thm:Weil-repn}
For every $g \in \mathrm{Mp}$ and every sequence $\{ x_n \} \subset
T_+$ with $x_n \to \mathrm{Id}$ the sequence $\{R( g x_n) \}$
converges in the bounded strong$^*$ topology. The limit $R'(g)$ is
independent of the sequence and defines a unitary representation
$R':\,\mathrm{Mp} \to \mathrm{U} (\mathcal A_V)$.
\end{theorem}
\begin{proof}
% Please do not remove the parentheses in {}From
%
{}From (\ref{last formal step}) we have $R'(g)(\mathrm{Id}_{
\mathcal{A}_V} - R'(g^{-1}) R'(g)) = 0$ and, since $R'(g)$ is
in\-jective, $R'(g^{-1}) R'(g) = \mathrm{Id}_{\mathcal{A}_V}\,$.
Hence $R'(g^{-1})$ is surjective, and thus $R'(g) \in
\mathrm{GL}(\mathcal{A}_V)$ by exchanging $g \leftrightarrow g^{-1}$.
For the homomorphism property we write $R(g_1 x_n) R(g_2 y_n) = R(g_1
x_n g_2 y_n) = R(g_1 x_n g_1^{-1}) R(g_1 g_2 y_n)$ and take limits to
obtain $R'(g_1) R'(g_2) = R'(g_1 g_2)$.

Convergence in the bounded strong$^*$ topology also requires
convergence of the adjoint. This property follows from $R(g
x_n)^\dagger = R(\widetilde{\psi}(g x_n)) = R(x_n g^{-1})$ and the
discussion after (\ref{other side convergence}), since $R'(g)$ is now
known to be an isomorphism. Unitarity of the representation is then
immediate from $R(g x_n)^\dagger \to R'(g)^\dagger$ and $R(x_n
g^{-1}) \to B(g^{-1}) = R'(g)^{-1}$.

Finally, we must show that $R' : \, \mathrm{Mp} \to \mathrm{U}
(\mathcal{A}_V)$ is continuous. This amounts to showing that if $\{
g_k \}$ is a sequence in $\mathrm{Mp}$ which converges to $g\,$, then
$R'(g_k) f \to R'(g) f$ for any $f \in \mathcal{A}_V$. Hence, we let
$\{ x_n \}$ be a sequence in $T_+$ with $x_n \to \mathrm{Id}$ and
choose $t = t(m,n)$ as in the proof of Proposition
\ref{well--defined} so that
\begin{displaymath}
    \Vert R(g_k x_m) - R(g_k x_n) \Vert \le \Vert R(g_k t)
    \Vert\, \Vert R(\tilde x_m) - R(\tilde x_n) \Vert \;,
\end{displaymath}
and then let $t \to \mathrm{Id}\,$. Using the uniform boundedness of
$R(g_k t)$ as $t \to \mathrm{Id}\,$, this shows that the convergence
$R(g_k x_n) \to R'(g_k)$ is uniform in $g_k\,$. Since we have $g_k
x_n \to g x_n$ for every fixed $n\,$, the continuity of $x \mapsto
R(x) f$ then implies that $R'(g_k)f \to R'(g)f$.
\end{proof}
Let us underline two important consequences.
\begin{proposition}
For $g_1 , g_2 \in \mathrm{Mp}$ and $x \in \widetilde{\mathrm{H}}
(W^s)$ it follows that
\begin{displaymath}
    R(g_1 x\, g_2) = R'(g_1) R(x) R'(g_2) \;.
\end{displaymath}
\end{proposition}
\begin{proof}
If $y_m$ and $z_n$ are sequences in $T_+$ which converge to
$\mathrm{Id}\,$, then, since $x \mapsto R(x) f$ is continuous for all
$f$ in Fock space, $R(g_1 y_m x\, g_2 z_n)$ converges pointwise to
$R(g_1 x \, g_2)$. On the other hand, we have $R(g_1 y_m x\, g_2 z_n)
= R(g_1 y_m) R(x) R(g_2 z_n)$ by the semigroup property, and the
right-hand side converges pointwise to $R'(g_1) R(x) R'(g_2)$.
\end{proof}
We refer to  $R' :\, \mathrm{Mp} \to \mathrm{U}(\mathcal{A}_V)$ as
the \emph{oscillator} representation of the metaplectic group.  It
has the following fundamental conjugation property.
\begin{proposition}
Let $\mathrm{Mp} \times W \to W$, $(g,w) \mapsto \tau(g) w$ denote
the representation of $\mathrm{Mp}$ on $W$ defined by first applying
the covering map $\mathrm{Mp} \to \mathrm{Sp}_\mathbb{R}$ and then
the standard representation of $\mathrm{Sp}\,$. If we let $W$ act on
$\mathfrak{a}(V)$ by the Weyl representation $q$ then
\begin{displaymath}
    R'(g) q(w) R'(g)^{-1} = q(\tau(g) w)\,.
\end{displaymath}
\end{proposition}
\begin{proof}
Since the inverse operator $R'(g)^{-1}$ is now available, this
follows from the conjugation property at the semigroup level (see
Proposition \ref{basic conjugation formula}).
\end{proof}
Note that analogously we have the classical conjugation formula for
the representation of the Heisenberg group on the Fock space
$\mathcal{A}_V$, i.e.,
\begin{equation*}
    R'(g) T_w R'(g)^{-1} = T_{\tau(g) w}
\end{equation*}
for all $g \in \mathrm{Mp}$ and $w \in W_\mathbb{R}$ (see Proposition
\ref{Heisenberg conjugation}).

\subsubsection{The trace-class property}

Recall that a linear operator $L$ on a Hilbert space $\mathcal{V}$ is
of trace class if and only if the non-negative self-adjoint operator
$\sqrt{LL^\dagger}$ has finite trace. If $L$ is of trace class, then
for every unitary basis $\{ u_i \}$ of $\mathcal{V}$ the trace
\begin{displaymath}
    \mathrm{Tr}\,L := \sum \langle u_i\, , L\, u _i\rangle
\end{displaymath}
converges absolutely. It is independent of the choice of unitary
basis and defines a linear functional on the space $C_1$ of operators
of trace class.
\begin{proposition}
For every $x \in \widetilde{\mathrm{H}}(W^s)$ the operator $R(x)$ is
of trace class.
\end{proposition}
\begin{proof}
Recall that $R(x) R(x)^\dagger = R(y)$, where $y = x \widetilde
{\psi}(x)\in M$. Since $y = g t^2 g^{-1}$ for some $t\in T_+$ and
$\sqrt{R(g t^2 g^{-1})} = R'(g)R(t)R'(g)^{-1}$, the desired result
follows from the explicit formula in Corollary \ref{eigenvalues in
T_+} for the eigenvalues of $t$.
\end{proof}
\begin{proposition}\label{prop:intrep}
An integral representation of the trace functional $\mathrm{Tr} : \,
C_1 \to \mathbb{C}$ on the space of trace-class operators on Fock
space is given by
\begin{displaymath}
    \mathrm{Tr}\, L = \sqrt{2}^{\,\mathrm{dim} W_\mathbb{R}} \Phi(L)
    \;,\quad \Phi(L) = \int_{W_\mathbb{R}} \langle T_w 1, L\, T_w
    1 \rangle_{\mathcal{A}_V}\, \mathrm{dvol}(w) \;.
\end{displaymath}
\end{proposition}
\begin{proof}
By the relation $w = v+cv$ we have $T_w 1 = T_{v+cv} 1 = \mathrm{e}
^{-\frac{1}{2}|v|^2 + \mathrm{i} \mu(cv)}$. As before, if $\{ e_i \}$
is an orthonormal basis of $V$ let $f_i := c e_i$ denote the dual
basis of $V^\ast$. Expanding $\mathrm{e}^{\mathrm{i} \mu(cv)}$ with
the help of multi-index notation $m = (m_1,\ldots,m_d)$ we get
\begin{displaymath}
    T_{v+cv} 1 = \mathrm{e}^{-\frac{1}{2}|v|^2}
    \sum \frac{(\mathrm{i}\overline{v})^m}{m!} f^m \;,
\end{displaymath}
where $m! = m_1! \cdots m_d!$ and $f^m = f_1^{m_1} \cdots f_d^{m_d}$.
By the isomorphism $V \simeq W_\mathbb{R}$ the integral $\Phi(L)$
pulls back to $\Phi(L) = \int_V \langle T_{v+cv} 1, L\, T_{v+cv} 1
\rangle\, \mathrm{dvol}(v)$ and inserting the series expansion of
$T_{v+cv} 1$ we obtain a double sum
\begin{displaymath}
    \Phi(L) = \sum \frac{\langle f^m,L f^{m^\prime}\rangle}
    {m! \, m^\prime!} \int_V (-\mathrm{i}v)^{m} (\mathrm{i}
    \overline{v})^{m^\prime} \mathrm{e}^{-|v|^2} \mathrm{dvol}(v)\;.
\end{displaymath}
If $m \not= m^\prime$ the integral vanishes, while direct computation
for $m = m^\prime$ yields
\begin{displaymath}
    \int_V v^{\,m} \overline{v}^{\,m} \mathrm{e}^{-|v|^2}
    \mathrm{dvol}(v) = 2^{- \mathrm{dim}_\mathbb{C} V} m! \;.
\end{displaymath}
The normalization here is determined by our convention $\int_V
\mathrm{e}^{-\frac{1}{2} |v|^2} \mathrm{dvol}(v) = 1$. Thus
$2^{\mathrm{dim}_\mathbb{C} V} \Phi(L) = \sum m!^{-1} \langle f^m , L
f^m \rangle$. The formula for $\mathrm{Tr}\, L$ now follows from
$\mathrm{dim}\, W_\mathbb{R} = 2\, \mathrm{dim}_\mathbb{C} V$ and the
fact that the set of normalized monomials $\{f^m / \sqrt{m!}\}$
constitute an orthonormal basis of Fock space.
\end{proof}
\begin{proposition}\label{holomorphicity} For every $P_1,P_2$ in the
Weyl algebra and every $x \in \widetilde{\mathrm{H}}(W^s)$ the
operator $q(P_1) R(x) q(P_2)$ is of trace class on the Fock space
$\mathcal{A}_V$. Furthermore, the function $\widetilde{\mathrm{H}}
(W^s) \to \mathbb{C}\,$, $x \mapsto \mathrm{Tr}\; q(P_1) R(x)
q(P_2)$, is holomorphic.
\end{proposition}
\begin{proof}
Note that every element in $\widetilde{\mathrm{H}}(W^s)$ can be
written as a product $xyz$ of elements in $\widetilde{\mathrm{H}}
(W^s)$. Accordingly we are going to show that operators of the form
\begin{displaymath}
    q(P_1) R(xyz) q(P_2) = (q(P_1)R(x)) R(y) (R(z)q(P_2))
\end{displaymath}
are of trace class. Recall that the space of trace-class operators is
a left and right ideal in the space of compact operators. It
therefore suffices to prove that operators of the form $B = q(P)
R(x)$ are compact. By linearity it is enough to handle the special
case where $q(P) = \mu(cv)^k \delta(v')^l$ for $k,l \in \mathbb{N}
\cup \{0\}$. With respect to a basis of Fock space consisting of
naturally defined monomials we will show that the matrix of
$BB^\dagger$ has finite trace. Thus the Hilbert-Schmidt norm $\Vert B
\Vert_\mathrm{HS}^2 := \mathrm{Tr}(BB^\dagger)$ is finite. Since
Hilbert-Schmidt operators are compact, the desired result follows.

For the computations we begin by observing that
\begin{displaymath}
    A := B B^\dagger = q(P) R(x) R(x)^\dagger q(P)^\dagger
    = q(P) R(y) q(P)^\dagger \;,
\end{displaymath}
where $y = x \widetilde{\psi}(x)\in M$. Let $\{e_i\}$ (resp.\ $\{f_i
\}$), $1 \le i \le d$, be the basis of $V$ (resp.\ dual basis of
$V^\ast$) so that $R(y)$ is diagonalized, i.e., $R(y) f^m =
\lambda^{-m-1/2} f^m$. As before, we are using multi-index notation
and the numbers $\lambda_i$ are the eigenvalues of $y \in M$ on the
basis vectors $e_i$ so that $\lambda_i > 1$ for all $i\,$. Let $v =
\sum a_i \, e_i$ and $v^\prime = \sum b_i \, e_i\,$. Thus
\begin{displaymath}
    \delta (v')^l = \sum b_\beta\frac{\partial^\beta}
    {\partial f^\beta} \quad \text{and} \quad \mu(cv)^k =
    \sum \overline{a}_\alpha f^\alpha \,,
\end{displaymath}
where the sums run over all multi-indices $\alpha$ and $\beta$ with
$|\alpha| = k$ and $|\beta| = l\,$.

After expanding $q(P) R(y) q(P)^\dagger$ we have individual terms
\begin{displaymath}
    C(\alpha,\beta,\tilde{\alpha},\tilde{\beta}) := f^\alpha
    \frac{\partial^\beta}{\partial f^\beta} R(y) f^{\tilde\beta}
    \frac{\partial^{\tilde\alpha}}{\partial f^{\tilde\alpha}}\;.
\end{displaymath}
Since we want to compute the trace of the matrix of $A$ with respect
to the basis of monomials $\{ f^m \}$, we need only consider those
operators $C \equiv C(\alpha,\beta,\tilde{\alpha},\tilde{\beta})$ for
which
\begin{equation}\label{diagonal}
    \alpha - \beta = \tilde{\alpha} - \tilde{\beta}\;.
\end{equation}

Now we write $A f^m$ in the monomial basis and must estimate the
coefficient of $f^m\,$. We will do this by estimating the eigenvalue
of $C$ on $f^m$ for those $C$ which satisfy (\ref{diagonal}), and
will do so with an estimate $K_m$ \emph{independent} of $\alpha,
\beta, \tilde\alpha ,\tilde\beta\,$. Hence we may eventually estimate
this coefficient by $N \, K_m$ where $N$ is the number of operators
$C(\alpha, \beta, \tilde\alpha, \tilde\beta)$ with multi-indices
satisfying (\ref{diagonal}). Since $N$ is bounded by a constant
independent of $m\,$, it comes out as a factor in our estimate of
$\mathrm{Tr}\, A$ and is therefore of no relevance to the argument.
Due to the fact that the vectors $v = \sum a_i \, e_i$ and $v^\prime
= \sum b_i \, e_i$ are also fixed, it is enough to compute $\sum_m
K_m\,$.

For those $C$ which satisfy (\ref{diagonal}) it follows that $C f^m =
r\, \lambda^{-(m - \alpha + \beta) - 1/2} f^m$ where
\begin{displaymath}
    r = \frac{m!}{(m-\tilde\alpha)!}\; \frac{(m - \tilde\alpha
    + \tilde\beta)!}{(m - \tilde\alpha + \tilde\beta - \beta)!}\,.
\end{displaymath}
Now, using the inequality $\frac{I!}{(I-J)!} \le \vert I \vert^{\vert
J \vert}$ and replacing each $\lambda_i$ by the smallest eigenvalue
$\lambda_{\rm min}\,$, up to a constant independent of $|m| = m_1 +
\ldots + m_d$ we have
\begin{displaymath}
    r\, \lambda^{- (m - \alpha + \beta) - 1/2} \le \mathrm{const}
    \times |m|^p \, \lambda_{\rm min}^{-|m|} =: K_m \;,
\end{displaymath}
where $p$ is a positive integer which does not depend on $m\,$.

To complete the computation, notice that the number of $m$ with $|m|
= n$ is bounded by $n^{d-1}$ times a constant. Therefore we get a
finite sum
\begin{displaymath}
    \sum K_m \le \mathrm{const} \times \sum_{n=0}^\infty
    n^{p+d-1} \lambda_\mathrm{min}^{-n} < \infty \;.
\end{displaymath}

The holomorphicity follows by inserting $L(x) = q(P_1) R(x) q(P_2)$
in the formula of Proposition \ref{prop:intrep} and interchanging the
$\bar\partial_x$-operator with the integral.
\end{proof}

\subsection{Compatibility with Lie algebra representation}

We now show that the semi\-group representation $R : \, \widetilde{
\mathrm{H}}(W^s) \to \mathrm{End}(\mathcal{A}_V)$ is compatible with the
$\mathfrak{sp}$-representation
\begin{displaymath}
    \mathfrak{sp} \stackrel{\tau^{-1}} {\longrightarrow} \mathfrak{w}
    (W) \stackrel{q}{\longrightarrow} \mathfrak{gl} (\mathfrak{a}(V))
    \;, \quad \mathfrak{a}(V) = \mathrm{S}(V^\ast) \;.
\end{displaymath}

Let $h \in \mathrm{H}(W^s)$ and $Y \in \mathfrak{sp}\,$. Then, since
the semigroup $\mathrm{H}(W^s)$ is open in $\mathrm{Sp}\,$, there
exists some $\varepsilon > 0$ so that the curve $[-\varepsilon ,
\varepsilon] \ni t \mapsto \mathrm{e}^{tY} h$ lies in $\mathrm{H}
(W^s)$. Fix some point $x \in \tau_H^{-1}(h)$ and let $t \mapsto
\mathrm{e}^{tY} \cdot x$ denote the lifted curve in $\widetilde
{\mathrm{H}}(W^s)$.
\begin{lemma}
For all $x \in \widetilde{\mathrm{H}}(W^s)$ and all $Y \in
\mathfrak{sp}$ it follows that
\begin{displaymath}
    \frac{d}{dt}\bigg\vert_{t = 0} R(\mathrm{e}^{tY} \cdot x) =
    q (\tau^{-1}(Y)) R(x) \;.
\end{displaymath}
\end{lemma}
\begin{proof}
Recall that the operator $R(x)$ is the result of integrating the
Heisenberg translations $T_w$ against the Gaussian density
$\gamma_x(w)\, \mathrm{dvol}(w)$. Thus
\begin{equation}\label{eq:3.12}
    \frac{d}{dt}\bigg\vert_{t = 0} R(\mathrm{e}^{tY} \cdot x)
    = \int_{W_\mathbb{R}} \frac{d}{dt}\bigg\vert_{t = 0} \gamma_{\,
    \mathrm{e}^{tY} \cdot x}(w) T_w \, \mathrm{dvol}(w) \;.
\end{equation}

For $w_1 , w_2 \in W$ the linear transformation $Y: \, w \mapsto w_1
A(w_2 , w) + w_2 A(w_1 , w)$ is in $\mathfrak{sp}\,$, and $\mathfrak
{sp}$ is spanned by such transformations. It is therefore sufficient
to prove the statement of the lemma for $Y$ of this form. Hence let
$Y := w_1 A(w_2 ,\cdot) + w_2 A(w_1 ,\cdot)$ and observe that the
corresponding element in the Weyl algebra is
\begin{displaymath}
    \tau^{-1}(Y) = {\textstyle{\frac{1}{2}}} (w_1 w_2 + w_2 w_1) \;.
\end{displaymath}
Now, defining $T_w$ for $w \in W$ by $T_w = \mathrm{e}^{\mathrm{i}
q(w)}$ as before, we have
\begin{equation}\label{eq:3.13}
    q(\tau^{-1}(Y)) R(x) = - \frac{d^2}{dt_1 dt_2}\bigg\vert_{
    t_1 = t_2 = 0} T_{t_1 w_1 + t_2 w_2} R(x) \;.
\end{equation}
Therefore, for $\tilde{w} := t_1 w_1 + t_2 w_2$ consider the
expression
\begin{displaymath}
    T_{\tilde{w}} R(x) = \int_{W_\mathbb{R}} \gamma_x(w)
    T_{\tilde{w}} T_w \, \mathrm{dvol}(w) \;.
\end{displaymath}
Using $T_{\tilde{w}} T_w = \mathrm{e}^{-\frac{1}{2} A(\tilde{w},w)}
T_{\tilde{w} + w}$ and shifting integration variables $w \to w -
\tilde{w}$ we obtain
\begin{equation}\label{eq:3.14}
    T_{\tilde{w}} R(x) = \int_{W_\mathbb{R}} \gamma_x(w - \tilde{w})\,
    \mathrm{e}^{-\frac{1}{2} A(\tilde{w},w)} T_w \, \mathrm{dvol}(w)\;.
\end{equation}
Comparing Eqs.\ (\ref{eq:3.13},\ref{eq:3.14}) with (\ref{eq:3.12}) we
see that the formula of the lemma is true if
\begin{displaymath}
    \frac{d}{dt} \bigg\vert_{t = 0} \gamma_{\,\mathrm{e}^{tY} \cdot
    x}(w) = - \frac{d^2}{dt_1 dt_2} \bigg\vert_{t_1 = t_2 = 0}
    \gamma_x(w - t_1 w_1 - t_2 w_2) \, \mathrm{e}^{-\frac{1}{2}
    A(t_1 w_1 + t_2 w_2 , w)} \;.
\end{displaymath}
But checking this equation is just a simple matter of computing
derivatives. Recall that $\gamma_x(w) = \phi(x)\, \mathrm{e}^{-
\frac{1}{4} A(w,\, sa_x w)}$ and $\phi(x) = \mathrm{Det}^{1/2}(a_x +
s)$. Writing $h := \tau_H (x)$ and using $\mathrm{Tr}\, Y=0$ one
computes the left-hand side to be
\begin{displaymath}
    \frac{d}{dt} \bigg\vert_{t = 0} \gamma_{\,\mathrm{e}^{tY} \cdot
    x}(w) = \gamma_x(w) \left( {\textstyle{\frac{1}{4}}}
    \mathrm{Tr}\, (Y s a_x) + {\textstyle{\frac{1}{2}}} A(h(1 -
    h)^{-1} w , Y h (1 - h)^{-1} w) \right) \;.
\end{displaymath}
On substituting $Y = w_1 A(w_2 , \cdot) + w_2 A(w_1 , \cdot)$, this
expression immediately agrees with the result of taking the two
derivatives on the right-hand side.
\end{proof}

\section{Spinor-oscillator character}
\setcounter{equation}{0}

The purpose of this chapter is to introduce the character of the spinor-oscillator representation of a certain super-semigroup $(\widetilde{H},\mathcal{F})$ in the orthosymplectic Lie supergroup of $W = V \oplus V^\ast = W_0 \oplus W_1\,$. A summary of this short chapter is as follows.

Referring to \cite{B} and \cite{HK} for details, we begin by briefly recalling the basic notions of Lie supergroups (in this case semigroups) and their representations. Next, we recall from $\S$\ref{sect:osp-repn} the infinite-dimensional representation of the complex Lie superalgebra $\mathfrak{osp}(W)$ on the spinor-oscillator module $\mathcal{A}_V$ (a.k.a.\ Fock space). The complex Lie group $G = \mathrm{SO}(W_1) \times \mathrm{Sp}(W_0)$ associated to the even part of $\mathfrak{osp}(W)$ is the base manifold of the associated Lie supergroup $\mathrm{OSp}\,$. As a $2:1$ covering space of the domain $\mathrm{SO} (W_1) \times \mathrm{H} (W_0^s)$ in $G$, the semigroup $\widetilde{H} := \mathrm{Spin}(W_1) \times_{\mathbb{Z}_2} \widetilde{\mathrm{H}}(W_0^s)$ inherits complex supermanifold structure. The $\mathfrak{osp}$-representation is integrated to $\widetilde{H}$ as a super-semigroup representation on $\mathcal{A}_V$. Using the character of a supergroup representation as a model, we introduce a superfunction $\chi$ on $\widetilde{H}$ which we regard as the character of this semigroup representation. We refer to it as the spinor-oscillator character for short.

In the last subsection of the chapter we let $V = U \otimes \mathbb{C}^N$ and recall the setting of a Howe dual pair $(\mathfrak{g},\mathfrak{k}) = (\mathfrak{osp}(U \oplus U^\ast),\mathfrak{o}_N)$ or $(\mathfrak{osp} (\widetilde{U} \oplus \widetilde{U}^\ast),\mathfrak{sp}_N)$. In this setting, we evaluate the spinor-oscillator character $\chi$ in two respects: (i) we Haar-average it over $K$, which amounts to projecting from $\mathcal{A}_V$ to the submodule $\mathcal{A}_V^K$ of $K$-invariants, and (ii) we restrict it to a toral set $T^+ \subset \widetilde{H}$ in a super-semigroup over $\mathfrak{g}\,$. We then show that the restricted character $\chi_{T^+} (t)$ coincides with the integral function $I(t)$ of Eq.\ (\ref{eq:1.1}).

\subsection{Background on Lie supergroups and their representations}
\label{sect:background}

Given a (finite-dimensional) complex Lie superalgebra $\mathfrak{g}\,$,
an associated complex Lie supergroup is a ringed space $(G,\mathcal{F})$
where $G$ is a complex Lie group associated to $\mathfrak{g}_0$ which in addition integrates the representation of $\mathfrak{g}_0$ on ${\mathfrak g}_1$. The group operations on $G$ lift to sheaf morphisms that satisfy the natural conditions imposed by associativity, inverse and fixing the identity. Uniqueness theorems allow us to choose $\mathcal{F}$ as the sheaf of germs of holomorphic functions with values in the Grassmann algebra $\Lambda := \wedge \mathfrak{g}_1^*$. Here we follow Berezin's
construction of the group structure, in particular his construction
of the derivations associated to $\mathfrak{g}$ which are defined by left and right `multiplication'.

The first step of this construction is to consider the complex Lie algebra \begin{displaymath}
    \tilde{\mathfrak{g}} = \Lambda_0 \otimes \mathfrak{g}_0 + \Lambda_1 \otimes \mathfrak{g}_1 \,,
\end{displaymath}
with Lie bracket
\begin{displaymath}
    [\alpha \otimes X , \beta \otimes Y] := \beta\alpha \otimes [X,Y] \,,
\end{displaymath}
where $\Lambda_0$ is the subspace of even elements in $\Lambda$ including the degree-zero elements $\wedge^0 \mathfrak{g}_1^* = \mathbb{C}$, and $\Lambda_1$ is the subspace of odd elements. Letting $\Lambda_0'$ denote the subspace of $\Lambda_0$ without the complex line $\wedge^0 \mathfrak{g}_1^*$, we observe that $\mathfrak{n} := \Lambda_0^\prime \otimes \mathfrak{g}_0 + \Lambda_1 \otimes \mathfrak{g}_1$ is an ideal in $\tilde{g}$ consisting of nilpotent elements. It can therefore be integrated to a simply connected nilpotent complex Lie group. This leads to the semidirect-product complex Lie group $\tilde{G} = G N$ which is associated to $\tilde{g}$.

Grassmann analytic continuation (GAC) is a process that extends functions in the structure sheaf $\mathcal{F}$ of $G$ to holomorphic functions with values in $\Lambda$ on the complex Lie group $\tilde{G}$ (\cite{B}, p.\ 250-257; see also \cite{HK}, $\S1$). The Lie supergroup structure morphisms of $\mathcal{F}$ are defined by the standard complex Lie group structure of $\tilde{G}$. Indeed, the left and right representations of $\mathfrak{g}$ as derivations on $\mathcal{F}$ are defined via the standard invariant vector fields defined by $\tilde{\mathfrak{g}}$ on $\tilde{G}\,$; and representations of the Lie supergroup $(G,\mathcal{F})$ are defined by representations (with coefficients in $\Lambda$) of the complex Lie group $\tilde{G}$. We will sketch some aspects of this below, %
% where?
%
referring to \cite{B} and \cite{HK} for details.

Our goal here is to introduce the spinor-oscillator representation
and define its character. While this is a super-semigroup representation
on an infinite-dimensional space, we begin by recalling the
basics of Lie supergroup representations on finite-dimensional spaces.
In abstract terms, a representation of a Lie supergroup $(G, \mathcal{F}_G)$ is a morphism $(\rho, \rho^*)$ of Lie supergroups
to $(\mathrm{GL}(V), \mathcal{F}_{\mathrm{GL}(V)})$, where $V = V_0 \oplus V_1$ is some graded vector space. Here we are interested in holomorphic representations, so that $\mathrm{GL}(V)$ is the complex Lie group $\mathrm{GL}(V_0)\times \mathrm{GL}(V_1)$ and $\mathcal{F}_{\mathrm{GL} (V)}$ is its standard matrix structure sheaf with values in the Grassmann algebra $\Lambda\,$. The map $\rho : \; G \to \mathrm{GL}(V_0) \times \mathrm{GL}(V_1)$ is a holomorphic homomorphism of complex Lie groups.

Let us assume that we are given a representation $\rho_*:\;\mathfrak{g}\to \mathfrak{gl}(V)$, which we extend to $\rho_\ast : \; \tilde{\mathfrak{g}} \to \Lambda_0 \otimes \mathfrak{gl}(V)_0 + \Lambda_1 \otimes \mathfrak{gl} (V)_1$ by $\rho_\ast(\alpha \otimes X) = \alpha \otimes \rho_\ast(X)$. With $\rho$ as above we then explicitly construct the morphism $\rho^*$, and hence the character $\rho^*(\mathrm{STr})$, as follows. First, writing elements $\tilde{g} \in \tilde G$ as $\tilde{g} = g \exp(\Xi)$ with $g\in G$ and $\Xi\in\mathfrak{n}\,$, we consider the Lie
%
% We purposely do not use GL(V) instead of End(V) as GL(V) has
% already been given a different meaning above.
%
group representation $\tilde\rho:\; \tilde{G}\to \Lambda\otimes \mathrm{End}(V)$ given by $g\,\mathrm{e}^{\Xi} \mapsto \rho(g)\, \mathrm{e}^{\,\rho_\ast(\Xi)}$. By the $\mathbb{Z}_2$-grading $V = V_0 \oplus V_1$ the image matrices are of the form $\begin{pmatrix} A & B\\ C & D \end{pmatrix}$, where the coefficients in $A$ and $D$ are elements of $\Lambda_0$ and those in $B$ and $C$ are elements of $\Lambda_1\,$. %
%
% I don't get what the following sentence was supposed to be good for:
%
% These should be regarded as functions on $G$ which have been Grassmann
% analytically continued to $\tilde{G}$.
%
Now if $f$ is a superfunction in $\mathcal{F}_{\mathrm{GL}(V)}$ with Grassmann analytic continuation $\tilde{f}$, then $\tilde{f} \circ \tilde\rho $ is a $\Lambda $-valued holomorphic function on $\tilde{G}$ which is the GAC of a function $\rho^\ast(f)$ in $\mathcal{F}_G\,$. To determine the latter, one restricts $\tilde{f} \circ \tilde\rho$ to a \emph{characteristic subset} $\Gamma(\xi_1, \ldots ,\xi_m)$ which is defined by a basis $\{\xi_1, \ldots ,\xi_m\}$ of $\mathfrak {g}_1^*$. If $\{F_1, \ldots ,F_m\}$ denotes the dual basis, $\Gamma(\xi_1,\ldots ,\xi_m)$ is the image of the map $G\to \tilde G$ given by $g \mapsto g\,\mathrm{e}^{\,\sum \xi_j \, F_j}$, so that $\rho^*(f) := \tilde{f}(\rho(g)\,\mathrm{e}^{\,\sum \xi_j\, \rho_* (F_j)})$. The character of the representation $(\rho, \rho^*)$ then is the superfunction in $\mathcal{F}_G$ which is defined in the expected way: $\chi(g) := \mathrm{STr}\,(\rho(g)\, \mathrm{e}^{\,\sum \xi_j\, \rho_*(F_j)})$.

\subsection{Character of the spinor-oscillator representation}
\label{sect:4.2}

In $\S$\ref{integrating} 
we constructed a semigroup representation $R \equiv R_0 :\; \widetilde{\mathrm{H}}(W_0^s)\to \mathrm{End}(\mathcal{A}_{V_0})$
and also a group representation $R^\prime :\; \mathrm{Mp}\to \mathrm{U} (\mathcal{A}_{V_0})$, which exponentiate the oscillator representation $\mathfrak{sp}(W_0) \to \mathrm{End}(\mathcal{A}_{V_0})$. These are compatible in that $\mathrm{Mp}$ acts on $\widetilde{\mathrm{H}} (W_0^s)$ by translation on the left and right, and $R_0(g_1 x \, g_2) = R'(g_1) R_0(x) R'(g_2)$ for all $g_1 , g_2 \in \mathrm{Mp}$ and $x \in \widetilde{\mathrm{H}} (W_0^s)$.

In the same vein, the complex spinor representation $\mathfrak{o}(W_1) \to \mathfrak{gl}(\wedge V_1^\ast)$ exponentiates to a holomorphic Lie group representation $R_1:\, \mathrm{Spin}(W_1) \to \mathrm{GL}(\wedge V_1^\ast)$. Consider now the tensor product $\mathcal{A}_V := \wedge V_1^\ast \otimes \mathcal{A}_{V_0}\,$. By trivial extension, our representations $R_0, R_1$ give rise to representations $R_0 :\; \widetilde{\mathrm{H}} (W_0^s)\to \mathrm{End} (\mathcal{A}_V)$ and $R_1 :\, \mathrm{Spin} (W_1) \to \mathrm{GL} (\mathcal{A}_V)$. Note that $R_1$ and $R_0$ commute, as they act on different factors of the tensor product $\mathcal{A}_V$. Note also that the group $\mathbb {Z}_2$ acts on $\widetilde{\mathrm{H}}(W_0^s)$ and $\mathrm{Spin} (W_1)$ by deck transformations of the $2:1$ coverings $\widetilde{\mathrm{H}} (W_0^s) \to \mathrm{H}(W_0^s)$ and $\mathrm{Spin}(W_1) \to \mathrm{SO} (W_1)$. The non-trivial element of $\mathbb{Z}_2$ is represented by a sign change, $R_0 \to - R_0$ and $R_1 \to - R_1\,$.

In the sequel, let $\widetilde{H}$ denote the semigroup
\begin{displaymath}
    \widetilde{H} := \mathrm{Spin}(W_1) \times_{\mathbb{Z}_2}    \widetilde{\mathrm{H}} (W_0^s) .
\end{displaymath}
Given the representations $R_1$ and $R_0\,$, we form the semigroup
representation
\begin{displaymath}
    R : \; \widetilde{H} \to \mathrm{End}(\mathcal{A}_V)\;, \quad [g_1 ; g_0] \equiv [-g_1 ; -g_0] \mapsto R_1(g_1) R_0(g_0) \;.
\end{displaymath}
% Please do not remove the parentheses in {}From
%
{}From $\S$\ref{sect:osp-repn} and our labors in $\S$\ref{integrating} we know that the semigroup representation $R$ does the job of partially integrating the spinor-oscillator representation $q :\; \mathfrak{g} \to \mathfrak{gl} (\mathcal{A}_V)$ of $\mathfrak{g} = \mathfrak{osp}(W) \subset \mathfrak{q}(W)$. Thus, in summary, what we have is a representation $(R,R_\ast \equiv q)$ on $\mathcal{A}_V$ of the complex Lie super-semigroup $(\widetilde{H},\mathcal{F})$ with complex Lie superalgebra $\mathfrak{g}\,$.

We now proceed according to the blueprint of the finite-dimensional setting in $\S$\ref{sect:background}. Since the semigroup structure of $\widetilde{H}$ comes from the semigroup structure of the complex domain $\mathrm{H}(W_0^s)$ in $\mathrm{Sp}(W_0)$, the mapping $R:\; \widetilde{H} \to \mathrm{End}(\mathcal{A}_V)$ is a semigroup morphism which is holomorphic in the appropriate infinite-dimensional sense. In order to regard this as a morphism of super-semigroups, we extend it as in the finite-dimensional setting to a semigroup representation $\tilde{R}$ with values in $\Lambda \otimes \mathrm{End}(\mathcal{A}_V)$. Just as in the finite-dimensional case, the morphism $R^*$ is defined by pulling back superfunctions defined on the image space of endomorphisms. By an immediate extension of Proposition \ref{holomorphicity}, our semigroup representation is such that for $x \in \widetilde{H}$ the operators $R(x)$ multiplied by any polynomial in the $\mathfrak{osp}$-generators are trace class. Thus $R^*(\mathrm{STr})$ is well-defined and, restricting to the characteristic set as above, we define the \emph{character of the spinor-oscillator super-semigroup representation} by
\begin{displaymath}
    \chi(x) := \mathrm{STr}_{\mathcal{A}_V}  R(x)\,      \mathrm{e}^{\,\sum\xi_j\,q(F_j)} ,
\end{displaymath}
which is a holomorphic function on $\widetilde{H}$ with values in $\wedge \mathfrak{g}_1^\ast = \wedge \mathfrak{osp}_1^\ast$.
%
% The old notation (11/2007) was $\gamma$.

\subsection{Identification of the restricted character with $I(t)$}
\label{character=integral}

Here we show that restriction of $\chi$ to a certain toral set in $\widetilde{H}$ yields the integrand $Z(t,k)$ of the autocorrelation function $I(t)$ described in $\S$\ref{introduction}. In other words, we show that $Z(t,k)$ can be expressed as the supertrace of an operator on the spinor-oscillator module $\mathcal{A}_V$. Then, by taking the Haar average over the compact group $K$, we identify $I(t)$ with the supertrace of an operator on the submodule $\mathcal{A}_V^K$ of $K$-invariants.

We begin with a summary of the relevant facts. From Proposition
\ref{basic conjugation formula} we recall the basic conjugation rule
for the oscillator representation:
\begin{displaymath}
    R_0(x)\, q(w) = q(\tau_H(x)\,w) R_0(x) ,
\end{displaymath}
where $x \in \widetilde{ \mathrm{H}}(W_0^s)$, $w \in W_0\,$, and
$\tau_H(x) \in \mathrm{H}(W_0^s) \subset \mathrm{Sp}(W_0)$. On the
side of the spinor representation, the corresponding conjugation
formula is
\begin{displaymath}
    R_1(y)\, q(w) = q(\tau_S(y)\, w) R_1(y) , \quad w \in W_1 \,,
    \quad  y \in \mathrm{Spin}(W_1) .
\end{displaymath}
This defines the $2:1$ covering homomorphism $\mathrm{Spin}(W_1) \to
\mathrm{SO}(W_1)$, $y \mapsto \tau_S(y)$, which exponentiates the
isomorphism of Lie algebras $\tau :\, \mathfrak{s} \cap \mathfrak{c}_2 (W_1) \to \mathfrak{o}(W_1)$.

In $\S3$ we have discussed the oscillator character $\phi $ in great detail. In particular, we know that
\begin{displaymath}
    \mathrm{Tr}\, R_0(x) = \phi(x) \;, \quad \phi(x)^2 =
    \mathrm{Det}^{-1}(\mathrm{Id}_{W_0} - \tau_H(x)) .
\end{displaymath}
%
% Notation \phi interferes with notation for the parameters
%
An analogous result is known for the case of the spinor representation; see, e.g., the textbook \cite{BGV}. Defining the spinor character as the supertrace with respect to the canonical $\mathbb{Z}_2$-grading of the spinor module, one has
\begin{displaymath}
    \mathrm{STr}\, R_1(y) = \psi(y) , \quad \psi(y)^2 =
    \mathrm{Det}(\mathrm{Id}_{W_1} - \tau_S(y)) .
\end{displaymath}
Thus the spinor character, just like the oscillator character, is a square root. By taking the supertrace over the total Fock representation space, we obtain the formula
\begin{equation}\label{eq:spin-osc-char}
    \mathrm{STr}_{\mathcal{A}_V} R_1(y) R_0(x) = \phi(x) \psi(y)
    =: \sqrt{ \frac{\mathrm{Det}(\mathrm{Id}_{W_1} - \tau_S(y))}
    {\mathrm{Det}(\mathrm{Id}_{W_0} - \tau_H(x))}} .
\end{equation}
For $W_0 = W_1\,$, the case of our interest, $\mathrm{Spin} (W_1)$
intersects with $\widetilde{\mathrm{H}} (W_0^s)$ and the square root
$\psi(y)$ is defined in such a way that $\psi(y) = \phi(x)^{-1}$ for
$x = y \in \mathrm{Spin}(W_1) \cap \widetilde{\mathrm{H}} (W_0^s)$.

Now let $U = U_0 \oplus U_1$ be a $\mathbb{Z}_2$-graded vector space,
and let $V = U \otimes \mathbb{C}^N$. The Lie group $\mathrm{GL}(V_1)
\times \mathrm{GL}(V_0)$ acts on $W = (V_1^{\vphantom{\ast}} \oplus
V_1^\ast) \oplus (V_0^{\vphantom{\ast}} \oplus V_0^\ast)$ by
\begin{displaymath}
    (g_1,g_0) . (v_1 \oplus \varphi_1 \oplus v_0 \oplus \varphi_0)
    = (g_1 v_1) \oplus (\varphi_1 \circ g_1^{-1}) \oplus (g_0 v_0)
    \oplus (\varphi_0 \circ g_0^{-1}) .
\end{displaymath}
This action serves to realize the group $G := (\mathrm{GL}(U_1) \times \mathrm{GL}(U_0)) \times_{\mathbb{C}^\times} \mathrm{GL}(\mathbb{C}^N)$ as a subgroup of $\mathrm{GL}(V_1) \times \mathrm{GL}(V_0) \subset \mathrm{SO}(W_1) \times \mathrm{Sp}(W_0)$. For the purpose of letting $G$ act on $\mathcal{A}_V$, let this representation be lifted to that of a double covering $\widetilde{G}$ of $G$ by
\begin{displaymath}
    \iota : \; \widetilde{G} \hookrightarrow \widetilde{H} = \mathrm{Spin} (W_1) \times_{\mathbb{Z}_2} \widetilde{\mathrm{H}} (W_0^s) .
\end{displaymath}
Our next statement gives the value of the spinor-oscillator character on $(t_1,t_0;g) \in \widetilde{G}$ where $g \in \mathrm{GL}(\mathbb{C}^N)$ and $t_s = \mathrm{diag}(t_{s,1}, \ldots, t_{s,\,n})$ are diagonal matrices in $\mathrm{GL}(U_s)$ (for $s = 0, 1$) or rather, in the pertinent double covering. Note that since the action of $\widetilde{G}$ on the spinor-oscillator module $\mathfrak{a}(V)$ is degree-preserving, there is no longer any need to work with the completion $\mathcal{A}_V$ to a Hilbert space.
\begin{lemma}\label{lem:STrAV}
If $\dim U_0 = \dim U_1 = n$ and $| t_{0,j} | > 1$ for all $j = 1,
\ldots, n\,$, then
\begin{displaymath}
    \mathrm{STr}_{\mathfrak{a}(V)}(R\circ\iota)(t_1,t_0;g) = \prod_{j=1}^n \sqrt{\frac{t_{1,j}}{t_{0,j}}}^{\,N}\frac{\mathrm{Det}(\mathrm{Id}_N - (t_{1,j}\,g)^{-1})}{\mathrm{Det}(\mathrm{Id}_N-(t_{0,j}\,g)^{-1})}\;.
\end{displaymath}
\end{lemma}
\begin{proof}
Since $t_1$ and $t_0$ are assumed to be of diagonal form, the
statement holds true for a general value of $n$ if it does so for
the special case of $n = 1$. Hence let $n = 1$.

In that case $t_1$ and $t_0$ are single numbers and $t_s g$ acts on
$W_s = V_s^{\vphantom{\ast}} \oplus V_s^\ast \simeq \mathbb{C}^N
\oplus (\mathbb{C}^N)^\ast$ as $(t_s g).(v \oplus \varphi) = (t_s g
v) \oplus \varphi \circ (t_s g)^{-1}$ for $s = 0,1$. From equation
(\ref{eq:spin-osc-char}) we then have
\begin{displaymath}
    \mathrm{STr}_{\mathfrak{a}(V)}(R\circ\iota)(t_1,t_0;g) = \sqrt{\frac{
    \mathrm{Det} (\mathrm{Id}_N - t_1 g)\,\mathrm{Det}(\mathrm{Id}_N
    - (t_1 g)^{-1})} {\mathrm{Det}(\mathrm{Id}_N-t_0 g)\,\mathrm{Det}
    (\mathrm{Id}_N-(t_0 g)^{-1})}} \;,
\end{displaymath}
which turns into the stated formula on pulling out a factor of $\mathrm {Det}(-t_1 g) / \mathrm{Det}(- t_0 g)$ from under the square root. (Of course, the double covering of $\mathrm{GL}(U_1) \times \mathrm{GL}(U_0)$ is to be used in order to define this square root globally.)
\end{proof}
In the formula of Lemma \ref{lem:STrAV} we now set $t_{1,j} = \mathrm
{e}^{\mathrm{i}\psi_j}$ and $t_{0,j} = \mathrm{e}^{\phi_j}$. We then
put $g^{-1} \equiv k \in K$ and integrate against Haar measure $dk$
of unit mass on $K$. This integral and the summation defining the
supertrace can be interchanged, as $\mathrm{STr}_{\mathfrak{a}(V)}
(R\circ\iota)(t_1,t_0;k^{-1})$ is a finite sum of power series and the
conditions $\mathfrak{Re} \, \phi_j > 0$ ensure uniform and absolute
convergence. Averaging over $K$ with respect to Haar measure has the effect of projecting from $\mathfrak{a}(V)$ to the $K$-trivial isotypic component $\mathfrak{a}(V)^K$, thus we arrive at
\begin{equation}\label{eq:4.3mrz}
    \mathrm{STr}_{\mathfrak{a}(V)^K} (R\circ\iota)(t_1,t_0;\mathrm{Id})
    =\mathrm{e}^{(N/2) \sum_j (\mathrm{i}\psi_j - \phi_j)}\,
    \int\limits_K \prod_{j=1}^n \frac{\mathrm{Det}(\mathrm{Id}_N
    - \mathrm{e}^{- \mathrm{i}\psi_j}\, k)}{\mathrm{Det}(
    \mathrm{Id}_N - \mathrm{e}^{- \phi_j}\, k)} \, dk \;.
\end{equation}
In the case of an even dimension $N$, the domain of definition of this formula is a complex torus $T^+ := T_1 \times T_0^+$ where $T_1 = (\mathbb{C}^\times)^n$ and $T_0^+ \subset (\mathbb{C}^\times)^n$ is
the open subset determined by the conditions $|t_{0,j} | = \mathrm{e}^{ \mathfrak{Re}\, \phi_j} > 1$ for all $j\,$. For odd $N$ we must continue to work with a double cover (also denoted by $T^+$) to take the square root $\mathrm{e}^{(N/2) \sum_j (\mathrm{i}\psi_j - \phi_j)}$.

Let now $\mathfrak{g}$ be the Howe dual partner of $\mathrm{Lie}(K)$ in
$\mathfrak{osp}(W)$. We know from Proposition \ref{prop:dualpairs}
that $\mathfrak{g} = \mathfrak{osp}(U \oplus U^\ast)$ for $K =
\mathrm{O}_N$ and $\mathfrak{g} = \mathfrak{osp}(\widetilde{U} \oplus
\widetilde{U}^\ast)$ for $K = \mathrm{USp}_N\,$. Recall also from
$\S$\ref{subsubsec:WoHdP} that the $\mathfrak{g}$-representation on
$\mathfrak{a}(V)^K$ is irreducible and of highest weight $\lambda_N =
(N/2) \sum_j (\mathrm{i}\psi_j - \phi_j)$. Denote by $\Gamma_\lambda$
the set of weights of this representation. Let $B_\gamma = (-1)^{
|\gamma|} \dim {\mathfrak{a}(V)^K}_\gamma$ be the dimension of the
weight space ${\mathfrak{a}(V)^K}_\gamma$ multiplied with the correct
sign to form the supertrace.
\begin{corollary}\label{cor:weightexpansion}
On $T^+$ we have
\begin{displaymath}
    \sum_{\gamma \in \Gamma_\lambda} B_\gamma \, \mathrm{e}^\gamma
    = \mathrm{e}^{\lambda_N} \int\limits_K \prod_{j=1}^n \frac{
    \mathrm{Det}(\mathrm{Id}_N - \mathrm{e}^{-\mathrm{i}\psi_j} \, k)}
    {\mathrm{Det}(\mathrm{Id}_N - \mathrm{e}^{-\phi_j}\, k)}\, dk \;.
\end{displaymath}
\end{corollary}
\begin{remark}
On the right-hand side we recognize the correlation function (see
$\S$\ref{introduction}) which is the object of our study and, as we have explained, is related to the character of the irreducible $\mathfrak {g}$-representation on $\mathfrak{a}(V)^K$. The left-hand side gives this character (restricted to the toral set $T^+$) in the form of a weight expansion, some information about which has already been provided by Corollary \ref{cor:weights} of $\S$\ref{subsubsec:WoHdP}.
\end{remark}

\section{Proof of the character formula}

Here we complete our task of deriving the formula (\ref{eq:WeylCharacter}),
%
% There was an earlier attempt to disentangle notation:
% Use $\mathrm{W}$ for the Weyl group, but $W$ for $V \oplus V^\ast$.
%
\begin{equation}\label{eq:Weylagain}
    I(t)=\sum_{[w] \in W \!\!/ W_\lambda} \mathrm{e}^{ w(\lambda_N)} \frac{\prod_{\beta \in \Delta_{\lambda,1}^+}(1-\mathrm{e}^{ -w(\beta)})} {\prod_{\alpha \in \Delta_{\lambda,0}^+}
    (1 - \mathrm{e}^{-w(\alpha)})} (\ln t) \;.
\end{equation}
Let us sketch this derivation, thereby giving an outline of this chapter.

Throughout we will be concerned with the irreducible $\mathfrak{g} $-representation on the subspace $\mathcal{A}_V^K$ of $K$-invariants in Fock space, where $\mathfrak{g} \subset \mathfrak{osp}(V \oplus V^\ast)$ denotes the Howe partner defined by the $K$-action on $V \oplus V^\ast$, $V = U \otimes \mathbb{C}^N$. Here our dealings with the `big' Lie superalgebra $\mathfrak{osp}(V \oplus V^\ast)$ in $\S$\ref{sect:4.2} are repeated at the level of the `small' Lie superalgebra $\mathfrak{g}\,$. In particular, we associate with $\mathfrak{g}$ a complex super-semigroup $(\widetilde{H}^\prime, \mathcal{F})$ which serves to partially integrate the $\mathfrak{g}$-representation on $\mathcal{A}_V^K$.

We then focus our attention on the spinor-oscillator character $\chi$ pulled back to $\widetilde{H}^\prime \times K$ and Haar-averaged over $K$. Let $\chi^\prime$ be the resulting superfunction on $\widetilde{H}^\prime$. Its restriction $\chi_{T^+}^\prime$ to a toral set $T^+ \subset \widetilde{H}^\prime$ is the numerical function that we were led to consider in $\S$\ref{character=integral}; by Eq.\ (\ref{eq:4.3mrz}) it is the function $I(t)$ which is to be computed.

To prove the formula (\ref{eq:Weylagain}) for the character $I(t)$, we study the \emph{full} superfunction $\chi^\prime : \; \widetilde{H}^\prime \to \wedge \mathfrak{g}_1^\ast$. General methods show that $\chi^\prime$ has two distinctive properties: (i) it is radial with respect to the vector fields given by $\mathfrak{g}\,$, and (ii) it is an eigenfunction for every Laplace-Casimir operator, i.e.\ every differential operator $D(I)$ associated to a Casimir invariant $I$, on $(\widetilde{H} ^\prime, \mathcal{F})$. Hence we look closely at the differential equations $D(I) \chi^\prime = \lambda \chi^\prime$. For $I_\ell = \sum \big( \phi_j^{2\ell} - (-1)^\ell \psi_j^{2\ell} \big)$, $\ell \in \mathbb{N}$, regarded as an element of the center of the universal enveloping algebra of $\mathfrak{g}$, we show that $D(I_\ell) \chi^\prime = 0$. It follows that the radial part of $D(I_\ell)$, which is the differential operator corresponding to $D(I_\ell)$ on $T^+$, annihilates the restricted character $I(t) = \chi_{T^+}^\prime(t)$.

Hence we come to the final steps of our proof of the formula for $I(t)$, the first of which is to derive explicit formulas for the radial parts of the Laplace-Casimir operators $D(I_\ell)$. For this we implement a good portion of Berezin's theory of radial operators, which has been adapted to the context of the present paper in \cite{HK}. Using this theory, we show that the character $\chi_{T^+}^\prime$ is annihilated by an infinite set of differential operators $D_\ell \circ J$ ($\ell \in \mathbb{N}$), where
\begin{equation}\label{eq:diff-ops}
    D_\ell= \sum_{j=1}^n \frac{\partial^{2\ell}}{\partial\phi^{2\ell}_j}- (-1)^\ell\sum_{j=1}^n \frac{\partial^{2\ell}}{\partial\psi^{2\ell}_j}
\end{equation}
and $J$ is the square root of a certain Jacobian. These differential equations alone are not enough to pin down a Weyl-group invariant holomorphic function on $T^+$.  However, we are able to derive additional information about the region of the nonzero weights of the Fourier development of $\chi_{T^+}^\prime$, and we then show that the function on the right-hand side of (\ref{eq:Weylagain}) is up to normalization the unique Weyl-group invariant holomorphic function on $T^+$ which satisfies these additional conditions and is annihilated by the $D_\ell \circ J$.

\subsection{Properties of the character $\chi^\prime$}\label{sect:5.1}

Here we will be working with the following data:
\begin{itemize}
\item a complex Lie superalgebra $\mathfrak{g} = \mathfrak{osp}(U \oplus U^\ast)$ or $\mathfrak{osp}(\widetilde{U} \oplus \widetilde{U}^\ast)$;
\item a complex Lie super-semigroup $(\widetilde{H}^\prime, \mathcal{F})$ over $\mathfrak{g}\,$;
\item the character $\chi^\prime$ of a representation $(\rho,\rho_\ast)$ of $(\widetilde{H}^\prime, \mathcal{F}, \mathfrak{g})$ on $\mathcal{A}_V^K$.
\end{itemize}
To begin, we fill in some details omitted from the introductory part of this chapter.

First of all, to construct $\widetilde{H}^\prime$ we take $G \subset \mathrm{SO}(W_1) \times \mathrm{Sp}(W_0)$ to be the complex Lie group associated to the even part $\mathfrak{g}_0 \subset \mathfrak{g}$ and let $H^\prime \subset G$ be the semigroup which is defined by intersecting $G$ with $\mathrm{SO}(W_1) \times \mathrm{H} (W_0^s)$. We then define $\widetilde{H}^\prime$ to be the pre-image of $H^\prime$ in the $2:1$ covering space $\widetilde{H} = \mathrm{Spin}(W_1) \times_{\mathbb{Z}_2} \widetilde{\mathrm{H}}(W_0^s)$ of $\mathrm{SO}(W_1) \times \mathrm{H} (W_0^s)$.

Second, the semigroup representation $\rho$ emerges in a natural way, as follows. As a Lie semigroup with Lie algebra $\mathfrak{g}_0 \oplus \mathfrak{k}$ contained in $\mathfrak{osp}(W)$ (by the isomorphism $\Psi$ of $\S$\ref{sect:howe-pairs}), the direct product $\widetilde{H}^\prime \times K$ is naturally embedded into $\widetilde{H}$:
\begin{displaymath}
    \iota : \; \widetilde{H}^\prime \to \widetilde{H} \times K .
\end{displaymath}
The semigroup representation $\rho : \; \widetilde{H}^\prime \to \mathrm{End}(\mathcal{A}_V^K)$ is then defined by pulling back $R$ to $\widetilde{H}^\prime \times \{\mathrm{Id}\}$ and projecting on $\mathcal{A}_V^K$:
\begin{displaymath}
    \rho(x)= (R\circ\iota)(x\,,\mathrm{Id})\big\vert_{\mathcal{A}_V^K}\;.
\end{displaymath}
Third, by the general principles explained earlier, the character $\chi^\prime$ of the super-semigroup representation $(\rho,\rho_\ast)$
of $(\widetilde{H}^\prime, \mathcal{F} , \mathfrak{g})$ on $\mathcal{A}_V^K$ is determined by
\begin{displaymath}
    \chi^\prime(x) = \mathrm{STr}_{\mathcal{A}_V^K}\, \rho(x) \, \mathrm{e}^{\sum \xi_j\, \rho_\ast(F_j)} .
\end{displaymath}
Here $\{ F_j \}$ is a basis of $\mathfrak{g}_1$ as usual, $\{ \xi_j \}$ is the dual basis of $\mathfrak{g}_1^\ast$, and the Lie superalgebra representation $\rho_\ast : \; \mathfrak{g} \to \mathfrak{gl}(\mathfrak{a} (V)^K)$ is obtained by pulling back the spinor-oscillator representation by the canonical embedding ($\S$\ref{sect:howe-pairs}) of the Howe dual pair $(\mathfrak{g},\mathfrak{k})$ into $\mathfrak{osp}(W)$ and projection to $\mathfrak{a}(V)^K$.

In the current subsection we show that the character $\chi^\prime$ is, as would be expected, a radial superfunction. We also show that it is an eigenfunction of every Laplace-Casimir operator $D(I)$ and if $\mathrm{dim}\, U_0 = \mathrm{dim}\, U_1 = \mathbb{C}^n$, i.e., if we are dealing with $\mathfrak{g} = \mathfrak{osp}_{2n|2n}\,$, then the Laplace-Casimir operators annihilate $\chi^\prime$.
%
% introduce the notation \mathfrak{osp}_{2n|2n}

To simplify our notation, we now drop the primes and write $\widetilde{H}, \chi$ instead of $\widetilde{H}^\prime, \chi^\prime$.

\subsubsection{Radiality of $\chi$}

A holomorphic superfunction $f$ on $\widetilde{H}$ is \emph{radial} if and only if for every $X \in \mathfrak{g}$ the sum $L_X + R_X$ of the derivations defined by the left and right representations of $X$ annihilate it. For a homogeneous element $X$ of $\mathfrak{g}$ the action of these derivations on $f$ is defined as follows (see \cite{B}, p.\ 258, and \cite{HK}, $\S1$). First, one considers the Grassmann analytic continuation (GAC) $\tilde{f}$ of $f$. If $X$ is even, then one differentiates $\tilde{f}$ with respect to the local action of the 1-parameter group $\mathrm{e}^{tX}$. If $X$ is odd, then one chooses an arbitrary element $\alpha \in \Lambda_1$ and differentiates $\tilde{f}$ with respect to the local action of $\mathrm{e}^{tY}$ where $Y = \alpha X$. One shows in this latter case that the result is of the form $\alpha L_X(\tilde f)$ where $L_X$ is an odd derivation which does not depend on $\alpha$. Of course $\alpha L_X$ could be identically zero; so it might be necessary to extend the Grassmann algebra in order to prevent this from happening unless $L_X$ vanishes identically. Thus in both the odd and even cases we have an operator $L_X$ on the sheaf of $\Lambda $-valued holomorphic functions on the complex Lie group $\tilde{G}$.  One checks that these operators stabilize the subspace of functions on $\tilde{G}$ which arise through GAC from $(G,\mathcal{F}_G)$ and that the resulting map $\mathfrak{g} \to \mathrm{Der}(\mathcal{F}_G)$, $X\mapsto L_X$ is a Lie superalgebra morphism. Carrying this out in the analogous way by multiplying the 1-parameter groups $\mathrm{e}^{-tX}$ and $\mathrm{e}^{ -tY}$ on the right, one obtains the morphism defined by $X\mapsto R_X\,$.

The key for showing that the operators $L_X+R_X$ annihilate the character $\chi$ is the fact that the GAC $\tilde\chi$ of $\chi$ is $\mathrm{STr}\; \tilde{\rho}$, where $\tilde{\rho}$ is the associated complex Lie semigroup representation of $\widetilde{H} N$. One defines this by $\tilde{\rho}(x\, \mathrm{e}^{\Xi}) := \rho(x)\, \mathrm{e}^{\rho_\ast (\Xi)}$, as before.
%
% where the representation $\rho_\ast$ has been extended to
% $\tilde{\mathfrak{g}} = \Lambda_0 \otimes \mathfrak{g}_0 + \Lambda_1
% \otimes \mathfrak{g}_1$ by $\rho(\omega \otimes X) = \omega \otimes
% \rho_\ast(X)$ for $\omega \in \Lambda$ and $X \in \mathfrak{g}\,$.
%
% Conjugation rules such as $\rho(x) q(w) = q(\tau_H(x) w) \rho(x)$ on the  % symplectic side, which derive from Proposition \ref{basic conjugation
% formula} by specialization to our Howe pair setting, are exactly what is % needed to prove that $\tilde{\rho}$ is a (semigroup) representation with % values in $\Lambda \otimes \mathrm{End}(\mathcal{A}_V^K)$.
%
% [crmz:] Here we are leaving somewhat of a gap.
%
\begin{proposition}
The character $\chi$ is a radial holomorphic superfunction on $\widetilde{H}$.
\end{proposition}
\begin{proof}
Let $X \in \mathfrak{g}_1$ and $Y = \alpha X$ be as above. By using the multiplicative semigroup property, the fact that $\mathrm{STr}\; [ \tilde\rho(g) , \rho_\ast(Y) ] = 0$, and using the $\Lambda$-linearity of $\mathrm{STr}$ to factor out $\alpha$, we observe that $\alpha (L_X + R_X)$ annihilates $\mathrm{STr}\, \tilde{\rho}$. As we mentioned above, in order to conclude that $L_X + R_X$ annihilates this, it may be necessary to extend the Grassmann coefficients. For $X \in \mathfrak{g}_0$ the argument is even simpler, as it isn't necessary to multiply by $\alpha$.
\end{proof}

\subsubsection{The character $\chi$ is a Laplace-Casimir eigenfunction}

Let us emphasize that with or without $\alpha$, after returning from GAC functions on $\widetilde{H} N$ to functions on $\widetilde{H}$, the left derivation of $\chi$ by $X \in \mathfrak{g}$ is given with respect to a basis by $(L_X \chi)(x) = \mathrm{STr}\,(\rho_\ast(X)\,\rho(x) \, \mathrm{e}^{ \sum \xi_j \, \rho_\ast(F_j)})$. If $I$ is any element of the universal enveloping algebra $\mathsf{U}(\mathfrak{g})$ and $D(I)$ is the differential operator associated to $I$ by $X \mapsto L_X$,
%
% of course, one might also use the derivation by right translations
%
then
\begin{displaymath}
    D(I)\chi (x) = \mathrm{STr}\, (\rho_\ast(I)\, \rho(x)\,
    \mathrm{e}^{\,\sum \xi_j \, \rho_\ast(F_j)})\,.
\end{displaymath}
If $I$ is in the center of $\mathsf{U}(\mathfrak{g})$, we refer to $D(I)$ as a Laplace-Casimir operator.
\begin{proposition}
$\chi$ is an eigenfunction of every Laplace-Casimir operator $D(I)$.
\end{proposition}
\begin{proof}
Since $I$ lies in the center of $\mathsf{U}(\mathfrak{g})$, the operator $\rho_\ast(I)$ commutes with all operators defined by $\mathsf{U} (\mathfrak{g})$ on $\mathfrak{a}(V)^K$. Now according to Proposition \ref{prop:Howeduality} the subalgebra $\mathfrak{g}^{(-2)} \oplus \mathfrak{g}^{(0)} \subset \mathfrak{g}$ of degree-non-increasing operators stabilizes the vacuum space $\langle 1 \rangle_\mathbb{C} \subset \mathfrak{a}(V)^K$. By the irreducibility of the $\mathfrak{g}$-representation on $\mathfrak{a}(V)^K$ this subalgebra
stabilizes no other proper subspace of $\mathfrak{a}(V)^K$. Therefore, the linear operator $\rho_\ast(I)$ stabilizes $\langle 1 \rangle_\mathbb{C}$ with some eigenvalue $\lambda(I)$. Furthermore, $1 \in \mathbb{C} \subset \mathfrak{a}(V)^K$ is a cyclic vector for the action of $\mathsf {U}(\mathfrak{g})$ on $\mathfrak{a}(V)^K$. Thus $\rho_\ast(I) \equiv \lambda(I)\, \mathrm{Id}_{\mathfrak{a}(V)^K}$ and the desired result follows.
\end{proof}

\subsubsection{Vanishing of the $D(I_\ell)$-eigenvalues}
\label{sect:CasEigenvalues}

Recall now from $\S$\ref{sect:osp-cas} that for every $\ell \in
\mathbb{N}$ we have a Casimir element $I_\ell \in \mathsf{U}
(\mathfrak{osp})$ of degree $2\ell$. Recall also that under the
assumption of equal dimensions $V_0 \simeq V_1$ we introduced $\partial ,
\widetilde{\partial} \in \mathfrak{osp}_1\, $, $C = [\partial, \widetilde{\partial} ] \in \mathfrak{osp}_0 \,$, and $F_\ell \in
\mathsf{U} (\mathfrak{osp})$ such that $I_\ell= [\partial,F_\ell]$ and $[\partial,C] = 0\,$. For the proof of Proposition \ref{lem:D(Il)chi=0} below, we will make use of these objects at the level of $U_0 \simeq U_1\,$.

Consider now any irreducible $\mathfrak{osp}$-representation on a
$\mathbb{Z}_2$-graded vector space $\mathcal{V}$ with the property
that the $\mathcal{V}$-supertrace of $\mathrm{e}^{-tC}$ ($t>0$) exists. Let $\lambda(I_\ell)$ be the scalar value of the Casimir invariant $I_\ell$ in the representation $\mathcal{V}$. Then a short computation using $I_\ell = [\partial , F_\ell]$ and $[\partial,C] = 0$ shows that $\lambda(I_\ell)$ multiplied by $\mathrm{STr}_\mathcal{V}\, \mathrm{e}^{-t C}$ vanishes:
\begin{displaymath}
    \lambda(I_\ell)\, \mathrm{STr}_\mathcal{V}\, \mathrm{e}^{-t C} = \mathrm{STr}_\mathcal{V}\, \mathrm{e}^{-t C}
    I_\ell = \mathrm{STr}_\mathcal{V}\, \mathrm{e}^{-t C}
    [\partial , F_\ell] = \mathrm{STr}_\mathcal{V}\, [\partial ,
    \mathrm{e}^{-t C} F_\ell] = 0 \;,
\end{displaymath}
since the supertrace of any bracket is zero. Thus we are facing a
dichotomy: either we have $\mathrm{STr}_\mathcal{V}\, \mathrm {e}^{
-t C} = 0\,$, or else $\lambda(I_\ell) = 0$ for all $\ell \in
\mathbb{N}\,$. Now it turns out that our representation $\mathfrak{a}
(V)^K$ realizes the latter alternative, which leads to the following
consequence.
\begin{proposition}\label{lem:D(Il)chi=0}
Let $U = U_0 \oplus U_1$ be a $\mathbb{Z}_2$-graded vector space with $U_0 \simeq U_1\,$ and $\chi$ be the character of the super-semigroup representation of $\widetilde H$ which is the integrated form of the irreducible $\mathfrak{g}$-representation on $\mathfrak{a}(V)^K$ for $V = U \otimes \mathbb{C}^N$. Then $D(I_\ell)\chi = 0$ for all $\ell \in \mathbb{N}\,$.
\end{proposition}
\begin{remark}
The condition $U_0 \simeq U_1$ is needed in order for the formula
$I_\ell = [ \partial , F_\ell]$ of Lemma \ref{lem:Cas-exact} to be
available.
\end{remark}
\begin{proof}
For any real parameter $t>0$ the supertrace of the operator $\rho( \mathrm{e}^{-t C})$ on $\mathcal{A}_V^K$ certainly exists and is non-zero. In fact, using formula (\ref{eq:4.3mrz}) one computes the value as
\begin{displaymath}
    \mathrm{STr}_{\mathcal{A}_V^K}\,\rho(\mathrm{e}^{-t C}) = \int_K \frac{\mathrm{Det}^n (\mathrm{Id}_N-\mathrm{e}^{-t}k)} {\mathrm{Det}^n (\mathrm{Id}_N - \mathrm{e}^{-t}k)}\, dk = 1 \not= 0 \;.
\end{displaymath}
The dichotomy of $\lambda(I_\ell)\, \mathrm{STr}_{\mathcal{A}_V^K} \, \rho(\mathrm{e}^{-t C}) = 0$ therefore gives $D(I_\ell) \chi = \lambda(I_\ell) \chi = 0\,$.
\end{proof}

\subsection{Derivation of the differential equations}

Here we outline a foundational result which leads to a proof that the differential operators $D_\ell \circ J$, where $J$ is the square root of a certain (super-)Jacobian and $\ell \in \mathbb N$, annihilate $I(t)$. Due primarily to Berezin \cite{B}, this result has been adapted to our context in \cite{HK}.

\subsubsection{Radial operators}

At this stage, another object enters: a space $T^+$ which plays the role of maximal complex torus in $\widetilde{H}$. To introduce it, we recall from the beginning of $\S$\ref{sect:5.1} that we are given a complex semigroup $H^\prime$ inside the complex Lie group $G$ with Lie algebra $\mathfrak{g}_0\,$. In terms of this structure, the space $T^+$ is defined as the pre-image in $\widetilde{H}$ of the intersection of the standard Cartan torus $T \subset G$ with $H^\prime$.

Let now $B$ be a neighborhood (open in $\widetilde{H}$) of a regular point $x \in T^+$. Then if $f :\; B \to \wedge g_1^\ast$ is any radial holomorphic (super-)function, we denote by
\begin{displaymath}
    \mathcal{R} f : \; T^+ \cap B \to \wedge^0 \mathfrak{g}_1^\ast = \mathbb{C}
\end{displaymath}
its restriction to a function (with numerical values) on the toral subset. 
The restriction map $\mathcal{R}$ so defined is known  
to be injective, irrespective of the choice of $B$ \cite{HK}.
As before, for $I$ an element of the universal enveloping algebra $\mathsf{U}(\mathfrak{g})$, we let $D(I)$ denote the associated differential operator. If $I$ lies in the center of $\mathsf{U}( \mathfrak{g})$ then $D(I)$ takes radial holomorphic functions to radial holomorphic functions. In this case, we may use the injectivity of the restriction map to define the radial part $\dot{D}(I)$ by
\begin{displaymath}
    \dot{D}(I) \circ \mathcal{R} = \mathcal{R} \circ D(I) .
\end{displaymath}

We now require an understanding of the correspondence $I \mapsto \dot{D}(I)$ between Casimir invariants and the radial parts of invariant differential operators. While this correspondence can be described in fairly explicit terms and has been the subject of recent research by one of the authors \cite{HK}, 
here we only summarize the final outcome needed for the present paper. For this one defines the meromorphic function
\begin{displaymath}
    J(t) = \frac{\prod_{\alpha\in\Delta_0^+} 2\sinh\frac{\alpha(\ln t)} {2}} {\prod_{\beta\in\Delta_1^+} 2 \sinh\frac{\beta(\ln t)}{2}}\;,
\end{displaymath}
where $\Delta^+ = \Delta_0^+ \cup \Delta_1^+$ is a system of even and odd positive roots (see $\S$\ref{sect:simple-roots}).

% Please do not remove the parentheses in {}From
%
{}From $\S$\ref{sect:osp-cas} we again recall that for every $\ell \in \mathbb{N}$ we have a Casimir element $I_\ell \in \mathsf{U} (\mathfrak{osp})$ of degree $2\ell$. We also recall the expression (\ref{eq:diff-ops}) for the differential operators $D_\ell$ in terms of the local coordinates $\phi_1, \ldots, \phi_n, \psi_1, \ldots, \psi_n$ we have been using all along.
\begin{theorem}\label{radialparts}
Let the neighborhood $B \subset \widetilde{H}$ of a regular point $x \in T^+$ be such that all points of the intersection $B \cap T^+$ are regular. Then, defining the radial part $\dot{D} (I_\ell)$ of the Laplace-Casimir operator $D(I_\ell)$ as above, one has
\begin{displaymath}
    \dot{D} (I_\ell) = J^{-1} (D_\ell + Q_{\ell-1}) \circ J\,,
\end{displaymath}
where $Q_{\ell-1}$ is a polynomial combination with constant coefficients
of the operators $D_1, \ldots, D_{\ell-1}$ which is of total degree
at most $2\ell - 2$.
\end{theorem}
\begin{remark}
While some choice of domain $B$ is necessary to ensure that both $J(t)$ and $J(t)^{-1}$ exist for all $t \in B \cap T^+$, the expression for $\dot{D}(I_\ell)$ does not depend on $B$.
\end{remark}
\begin{remark} The statement of Theorem \ref{radialparts} is the local version of a result due to Berezin \cite{B}. The proof is in \cite{HK}.
\end{remark}

\subsubsection{The differential equations}

In view of the formula for $\dot{D}(I_\ell)$, and knowing that $D(I_\ell) \chi = 0$ for all $\ell$, the following is a key technical step.
\begin{lemma}\label{lem:tech-step}
For all $\ell \in \mathbb{N}$ we have $D_\ell \, J = 0$.
\end{lemma}
\begin{proof}
One has the relation $\sinh \frac{x+y}{2} \sinh \frac{x-y}{2} =
\frac{1}{2}(\cosh x -\cosh y)$. Hence,
\begin{displaymath}
    J = \frac{\prod_{j<k} 4\,(\cosh \phi_j-\cosh \phi_k)(\cos\psi_j
    -\cos\psi_k)} {\prod_{j,\,k} 2\,(\cosh\phi_j -\cos\psi_k)} \; R
\end{displaymath}
with $R = \prod_{j=1}^n 2\sinh\phi_j$ or $R = \prod_{j=1}^n 2\mathrm{i} \sin\psi_j$ depending on whether $\mathfrak{g} = \mathfrak{osp} (U\oplus U^*)$ or $\mathfrak{g} = \mathfrak{osp} (\widetilde{U} \oplus \widetilde{U}^*)$. By the Cauchy determinant formula
\begin{displaymath}
    \mathrm{Det} \left(\frac{1}{x_j - y_k}\right)_{j,\,k=1,\ldots,\,n} =
    \frac{\prod_{j<k} (x_j-x_k)(y_k-y_j)}{\prod_{j,\,k=1}^n (x_j - y_k)},
\end{displaymath}
we obtain, up to a constant factor,
\begin{displaymath}
    J \propto \mathrm{Det} \left( \frac{\sinh\phi_j}{\cosh\phi_j -
    \cos\psi_k} \right)\;, \quad \text{or} \quad J \propto \mathrm{Det}
    \left(\frac{\mathrm{i}\sin\psi_j}{\cosh\phi_j - \cos\psi_k} \right).
\end{displaymath}
Consider the case of $\mathfrak{g} = \mathfrak{osp}(U \oplus U^\ast)$. Expanding the determinant $J$ as a sum over permutations and applying the differential operator $D_\ell$ to the summands, we then obtain $D_\ell \, J \propto$
\begin{align*}
    &D_\ell \sum_{\sigma\in\mathrm{S}_n} (-1)^{|\sigma|} \prod_{j=1}^n
    \frac{\sinh \phi_j}{\cosh \phi_j -\cos\psi_{\sigma(j)}} =
    \sum_\sigma (-1)^{|\sigma|} D_\ell \prod_{j=1}^n \frac{\sinh \phi_j}
    {\cosh\phi_j - \cos\psi_{\sigma(j)}} = \\
    &\sum_\sigma (-1)^{|\sigma|}  \sum_{k=1}^n \prod_{j\ne k}
    \frac{\sinh\phi_j}{\cosh\phi_j -\cos\psi_{\sigma(j)}}
    \left(\frac{\partial^{2\ell}}{\partial \phi_k^{2\ell}} -(-1)^\ell
    \frac{\partial^{2\ell}}{\partial \psi_{\sigma(k)}^{2\ell}} \right)
    \frac{\sinh \phi_k}{\cosh\phi_k - \cos\psi_{\sigma(k)}} \;.
\end{align*}
Now
\begin{displaymath}
    d_\ell \equiv \frac{\partial^{2\ell}}{\partial \phi^{2\ell}}
    - (-1)^\ell \frac{\partial^{2\ell}}{\partial \psi^{2\ell}} =
    d_1 \sum_{j=0}^{\ell-1} (-1)^{j}\frac{\partial^{2\ell-2}}{\partial
    \phi^{2\ell-2-2j}\partial \psi^{2j}} \;,
\end{displaymath}
and the statement for $\mathfrak{g} = \mathfrak{osp}(U \oplus U^\ast)$ is a consequence of the equation
\begin{displaymath}
    \left( \frac{\partial^2}{\partial \phi^2}+\frac{\partial^2}
    {\partial \psi^2}\right) \frac{\sinh\phi}{\cosh\phi - \cos\psi} = 0 ,
\end{displaymath}
which results from $(\cosh\phi - \cos\psi)^{-1} 2 \sinh\phi = \coth( \frac{\phi + \mathrm{i}\psi}{2}) + \coth(\frac{\phi - \mathrm{i}\psi}{2})$ and the fact that every (anti-)holomorphic function on a domain in $\mathbb{C}$ is harmonic.

The reasoning for the case of $\mathfrak{g} = \mathfrak{osp}(\widetilde{U} \oplus \widetilde{U}^\ast)$ is no different.
\end{proof}
\begin{corollary}\label{cor:DlJchi}
The restriction $\chi_{T^+}$ of the character $\chi$ from $\widetilde{H}$ to $T^+$ satisfies the system of differential equations $D_\ell \, J \, \chi_{T^+} = 0$ for all $\ell \in \mathbb{N}\,$.
\end{corollary}
\begin{proof}
Since $D(I_\ell) \chi = 0$ and hence by restriction $\dot{D}(I_\ell) \chi_{T^+} = 0$, it follows from Theorem \ref{radialparts} that $J^{-1} (D_\ell + Q_{\ell-1}) J \chi_{T^+} = 0$ for every $\ell\in \mathbb{N}\,$.
Now by applying the operator $\dot{D}(I_\ell)$ to the constant function $1$ and using Lemma \ref{lem:tech-step}, we obtain $0 = J \dot{D}(I_\ell) 1 = D_\ell \, J + Q_{\ell-1}\, J = c_{\ell-1}\,J$, where $c_{\ell-1}$ is the constant term of the differential operator $Q_{\ell-1}\,$, and from this we conclude that $c_{\ell - 1} = 0$ for all $\ell\in \mathbb{N}\,$. It then follows by induction on $\ell$ that $D_\ell \, J \, \chi  = 0$ for all $\ell \in \mathbb{N}\,$.
\end{proof}

The methods of this section can also be used to derive differential
equations for the characters of a certain class of irreducible
representations of $\mathfrak{gl}(U) \simeq \mathfrak{g}^{(0)}$.
Define
\begin{displaymath}
    J_0 = \frac{\prod_{j < k} 4 \sinh\frac{\mathrm{i}(\psi_j -
    \psi_k)}{2} \, \sinh\frac{\phi_j - \phi_k}{2}}{\prod_{j,\,k}
    2 \sinh\frac{\phi_j - \mathrm{i} \psi_k}{2}} \;.
\end{displaymath}
Here, $\{ \mathrm{i}(\psi_j - \psi_k), \phi_j - \phi_k \mid j < k \}$
and $\{ \phi_j - \mathrm{i} \psi_k \}$ are the sets of even and odd
positive roots of $\mathfrak{g}^{(0)}$. The following statement is
Corollary 4.12 of \cite{CFZ} adapted to the present context and
notation. The idea of the proof is the same as that of Proposition
\ref{lem:D(Il)chi=0} in conjunction with Corollary \ref{cor:DlJchi}.
\begin{corollary}\label{cor:4.4}
Let $\gamma$ be the (restricted) character of an irreducible representation of the Lie supergroup $(\mathfrak{gl}(U),\mathrm{GL}(U_0) \times \mathrm{GL}(U_1))$ on a finite-dimensional $\mathbb{Z}_2$-graded vector space $V = V_0 \oplus V_1\,$. If $U_0 \simeq U_1$ but $\mathrm{dim}(V_0) \not= \mathrm{dim}(V_1)$, then $D_\ell J_0\,\gamma = 0$ for all $\ell \in \mathbb{N}\,$.
\end{corollary}

\subsection{Global $G_\mathbb R$-invariance and the Weyl group}
\label{sect:WeylGroup}

Recall that $\chi$ is only invariant by the local action of the supergroup $(G,\mathfrak{g})$ on $\widetilde{H}$. However, there exists a real form $G_\mathbb{R}$ which acts globally on $\widetilde{H}$ by conjugation and therefore $\chi$ is invariant by this action.

In order to identify these real symmetry groups $G_\mathbb{R}$ in the
two cases at hand, we first observe that the good real group acting in
the spinor-oscillator representation $\mathcal{A}_V$ is
\begin{displaymath}
    \mathrm{Spin}(W_{1,\mathbb{R}}) \times_{\mathbb{Z}_2}
    \mathrm{Mp}(W_{0,\mathbb{R}}) =: G^\prime \;,
\end{displaymath}
which contains $K = \mathrm{O}_N$ and $K = \mathrm{USp}_N$ as
subgroups. Since we are studying the character $\chi$ of the
$\mathfrak{g}$-representation on the subspace $\mathcal{A}_V^K$ of
$K$-invariants, we now seek the subgroup $G_\mathbb{R} \subset
G^\prime$ which centralizes $K$; this means that we are asking the
exponentiated version of a question which was answered at the Lie
algebra level in $\S$\ref{sect:2.7}. Here, restricting the group
$G^\prime$ to the centralizer of $K$ we find
\begin{displaymath}
    G_\mathbb{R} = \left\{ \begin{array}{ll} \mathrm{Spin}( (U_1^{
    \vphantom{\ast}}\oplus U_1^\ast)_\mathbb{R})\times_{\mathbb{Z}_2}
    \mathrm{Mp}( (U_0^{\vphantom{\ast}} \oplus U_0^\ast)_\mathbb{R})
    &\quad K = \mathrm{O}_N \;, \\ \mathrm{USp}(U_1^{\vphantom{\ast}}
    \oplus U_1^\ast) \times \mathrm{SO}^\ast(U_0^{\vphantom{\ast}}
    \oplus U_0^\ast) &\quad K = \mathrm{USp}_N\;.\end{array}\right.
\end{displaymath}
We observe that $G_\mathbb{R}$ for the case of $K = \mathrm{O}_N$ is
just the lower-dimensional copy of $G^\prime$ which corresponds to
$U_s$ taking the role of $V_s\,$. We also see immediately that the
Lie algebras $\mathrm{Lie}(G_\mathbb{R})$ coincide with the real
forms described in Propositions \ref{prop:2.4} and \ref{prop:2.5}.

Since $\chi$ is invariant under the $G_\mathbb{R}$-action by conjugation, its restriction to a real toral semigroup in $T^+$ is invariant under the action of the Weyl group $W$ defined by $G_\mathbb{R}\,$. Since $\chi$ is holomorphic, its restriction to the complexification $T^+$ is likewise $W$-invariant. Now $G_\mathbb{R}$ decomposes as a direct product of two factors and so $W$ also decomposes in this way. For both cases ($K =
\mathrm{O}_N$, $\mathrm{USp}_N$) the second factor of the Weyl group
$W$ is just the permutation group $\mathrm{S}_n\,$. As a matter of fact, conjugation of a diagonal element $t_0 \in M_\mathrm{Sp}$ or $t_0 \in M_\mathrm{SO}$ by $g \in \mathrm{Mp}((U_0^{\vphantom{\ast}} \oplus U_0^\ast)_\mathbb{R})$ or $g \in \mathrm{SO}^\ast( U_0^{ \vphantom{\ast}} \oplus U_0^\ast)$ can return another diagonal element only by permutation of the eigenvalues $\mathrm{e}^{\phi_1} , \ldots, \mathrm{e}^{\phi_n}$ of $t_0\,$. (No inversion $\mathrm{e}^{\phi_j} \to \mathrm{e}^{- \phi_j}$ is possible, as this would mean transgressing the oscillator semigroup.) This factor $\mathrm{S}_n$ of $W$ will play no important role in the following, as the expressions we will encounter are automatically invariant under such permutations.

The first factors of $W$ are of greater significance. For the two
cases of $K = \mathrm{O}_N$ and $K = \mathrm{USp}_N$ these are the
Weyl groups $W_{\mathrm{SO}_{2n}}$ and $W_{\mathrm{Sp}_{2n}}$
respectively. An explicit description of these groups is as follows.
Let $\{ e_1 , \ldots , e_n \}$ be an orthonormal basis of $U$ and
decompose $U \oplus U^*$ into a direct sum of 2-planes,
\begin{displaymath}
    U \oplus U^* = P_1 \oplus \ldots \oplus P_n \;,
\end{displaymath}
where $P_j$ is spanned by the vector $e_j$ and the linear form $c e_j
= \langle e_j \, , \cdot \rangle$ ($j = 1, \ldots,n$). In both cases
at hand, i.e., for the symmetric form $S$ as well the alternating
form $A$, this is an orthogonal decomposition. The real torus under
consideration is parameterized by $(\mathrm{e}^{\mathrm{i}\psi_1} ,
\ldots , \mathrm{e}^{\mathrm{i}\psi_n}) \in (\mathrm{U}_1)^n$ acting
by $\mathrm{e}^{\mathrm{i}\psi_j} . (e_j) = \mathrm{e}^{\mathrm{i}
\psi_j} \, e_j$ and $\mathrm{e}^{\mathrm{i} \psi_j}. (c e_j) =
\mathrm{e}^{-\mathrm{i}\psi_j} ce_j\,$.

The Weyl group $W_{\mathrm{Sp}}$ is generated by the permutations of
these planes and the involutions which are defined by conjugation by
the mapping that sends $e_j \mapsto c e_j$ and $c e_j \mapsto -
e_j\,$. The Weyl group $W_\mathrm{SO}$ is generated by the
permutations together with the involutions which are induced by the
mappings that simply exchange $e_j$ with $c e_j\,$. Since we are in
the special orthogonal group and the determinant for a single
exchange $e_j \leftrightarrow c e_j$ is $-1$, the number of
involutions in any word in $W_\mathrm{SO}$ has to be even.

In summary, the $W$-action on our standard bases of linear functions,
$\{ \mathrm{i}\psi_j\}$ and $\{\phi_j \}$, is given by the respective
permutations together with the action of the involutions defined by
sign change, $\mathrm{i} \psi_j \mapsto - \mathrm{i} \psi_j\,$. In
the sequel, the Weyl group action will be understood to be either
this standard action or alternatively, depending on the context, the
corresponding action on the exponentiated functions $\{ \mathrm{e}^{
\mathrm{i} \psi_j} \}$ and $\{ \mathrm{e}^{\phi_j} \}$. As a final
remark, let us note that the Weyl-group symmetries of the function
$\chi_{T^+}$ can also be read off directly from the explicit expression (\ref{eq:4.3mrz}). In particular, the absence of reflections $\phi_j \to -\phi_j$ is clear from the conditions $\mathfrak{Re}\, \phi_j > 0\,$.

\subsection{Formula for $\chi_{T^+}$}
\label{unicity theorem}

Recall that the main goal of this paper is to compute the restriction $\chi_{T^+}$ to $T^+$ of the character $\chi$ which is defined on $\widetilde{H}$ and plays the role of a character of the $\widetilde{H}$-representation on the space of invariants $\mathcal{A}_V^K$ in the spinor-oscillator module. Here
$\widetilde{H}$ is the $2:1$ cover of an open semigroup in the complex
Lie group $G$ of the Howe partner supergroup of $K$. From now on we
will only deal with the restriction of this numerical part and
therefore we simplify notation by denoting it by $\chi_{T^+} \equiv \chi$.

We have restricted ourselves to the cases where $K$ is either
$\mathrm{O}_N$ or $\mathrm{USp}_N\,$. The representation on
$\mathfrak{a}(V)^K$ is defined at the infinitesimal level on the full
complex Lie superalgebra $\mathfrak{g}$ which is the Howe partner of $K$ in the canonical realization of $\mathfrak{osp}$ in the Clifford-Weyl algebra of $V \oplus V^\ast$. It has been shown that $\chi : \, T^+ \to \mathbb{C}$ satisfies the differential equations $D_\ell\, J\, \chi = 0\,$. We now recall that the nonzero weights of the Fourier expansion of $\chi$ are constrained to a certain region and show that $\chi$ is the unique holomorphic function satisfying both the weight constraints and the differential equations.

\subsubsection{Uniqueness}

Recall that $\Gamma_\lambda$ denotes the set of weights of the
$\mathfrak{g}$-representation on $\mathfrak{a}(V)^K$. From Corollary
\ref{cor:weights} we know that the weights $\gamma = \sum_{j = 1}^n
(\mathrm{i} m_j \psi_j - n_j \phi_j) \in \Gamma_\lambda$ satisfy the
weight constraints $- \frac{N}{2} \le m_j \le \frac{N}{2} \le n_j\,$.
The highest weight is $\lambda = \frac{N}{2} \sum (\mathrm{i} \psi_j
- \phi_j)$. By the definition of the torus $T^+$ the weights $\gamma \in
\Gamma_\lambda$ are analytically integrable and we now view
$\mathrm{e}^\gamma$ as a function on $T^+$.
\begin{theorem}\label{unicity}
The character $\chi : \, T^+ \to \mathbb{C}$ is annihilated by all
differential operators $D_\ell \circ J$ for $\ell \in \mathbb{N}\,$,
and it has a convergent expansion $\chi = \sum B_\gamma \,
\mathrm{e}^{\gamma}$ where the sum runs over weights $\gamma =
\sum_{j = 1}^n (\mathrm{i} m_j \psi_j - n_j \phi_j)$ satisfying the
constraints $- \frac{N}{2} \le m_j \le \frac{N}{2} \le n_j\,$. For
the case of $K = \mathrm{USp}_N$ it is the unique $W$-invariant function on $T^+$ with these two properties and $B_\lambda = 1$. For $K = \mathrm{O}_N$ it is the unique $W$-invariant function on $T^+$ with these two properties and $B_\lambda = 1\,$, $B_{\lambda - \mathrm{i}N\psi_n} =
0\,$.
\end{theorem}
\begin{remark}
To verify the property $B_{\lambda - \mathrm{i}N\psi_n} = 0$ which
holds for the case of $K = \mathrm{O}_N\,$, look at the right-hand
side of the formula of Corollary \ref{cor:weightexpansion}: in order
to generate a term $\mathrm{e}^\gamma = \mathrm{e}^{\lambda -
\mathrm{i}N\psi_n}$ in the weight expansion, you must pick the term
$\mathrm{e}^{- \mathrm{i} N\psi_n}$ in the expansion of the
determinant for $j = n$ in the numerator; but the latter term depends
on $k$ as $\mathrm{Det}(-k)$ which vanishes upon taking the Haar
average for $K = \mathrm{O}_N \,$. By $W$-invariance the property
$B_{\lambda - \mathrm{i}N\psi_n} = 0$ is equivalent to $B_{\lambda -
\mathrm{i} N\psi_j} = 0$ for all $j\,$.
\end{remark}
In view of this Remark and Corollaries \ref{cor:weights} and
\ref{cor:DlJchi}, it is only the uniqueness statement of Theorem
\ref{unicity} that remains to be proved here. This requires a bit of
preparation, in particular to appropriately formulate the condition
$D_\ell\, J \chi = 0\,$. For that we develop $J \chi$ in a series $J
\chi = \sum_\tau a_\tau f_\tau$ where the $f_\tau$ are
$D_\ell$-eigenfunctions for \emph{every} $\ell \in \mathbb{N}\,$.

The first step is to determine an appropriate expansion for $J$.
Recall that
\begin{displaymath}
    J = \frac{\prod_{\alpha \in \Delta ^+_0} (\mathrm{e}^{
    \frac{\alpha}{2}}- \mathrm{e}^{- \frac{\alpha}{2}})}
    {\prod_{\beta \in \Delta ^+_1} (\mathrm{e}^{\frac{\beta}{2}}
    - \mathrm{e}^{-\frac{\beta}{2}})} \;.
\end{displaymath}
Given a factor in the denominator of this representation, we wish to
factor out, e.g., $\mathrm{e}^{-\frac{ \beta}{2}}$ to obtain a term
$(1 - \mathrm{e}^{-\beta})^{-1}$ which we will attempt to develop in
a geometric series. In order for this to converge uniformly on
compact subsets of $T^+$ it is necessary and sufficient for
$\mathfrak{Re}\, \beta$ to be positive on $\mathfrak{t} = \mathrm{Lie}(T^+)$. This of course depends on the root $\beta\,$.
Fortunately, the sets of odd positive roots for our two cases of $K=
\mathrm{O}_N$ and $K= \mathrm{USp}_N$ are the same (see
$\S$\ref{sect:simple-roots}):
\begin{displaymath}
    \Delta_1^+ =\{\phi_j\pm\mathrm{i}\psi_k \mid j\,,k=1,\ldots,n\}\;.
\end{displaymath}
So indeed, if we factor out $\mathrm{e}^{-\frac{\beta}{2}}$ from each
term in the denominator and do the same in the numerator, we obtain
the expression
\begin{displaymath}
    J = \mathrm{e}^{\delta}\, \frac{\prod_{\alpha \in \Delta^+_0}
    (1 - \mathrm{e}^{-\alpha})} {\prod_{\beta \in \Delta ^+_1}
    {(1 - \mathrm{e}^{-\beta})}} \;,
\end{displaymath}
and it is possible to expand each term of the denominator in a
geometric series. Here
\begin{displaymath}
    \delta = \frac{1}{2}\sum_{\alpha \in \Delta ^+_0}\alpha
    - \frac{1}{2} \sum_{\beta \in \Delta^+_1} \beta
\end{displaymath}
is half the \emph{graded sum} of the positive roots.

Now let $\{\sigma_1, \ldots , \sigma _r\}$ be a basis of simple
positive roots (cf.\ $\S$\ref{sect:simple-roots}) and expand the
terms $(1- \mathrm{e}^{-\beta})^{-1}$ in geometric series to obtain
\begin{displaymath}
    J = \mathrm{e}^{\delta} \, \sum_{b\ge 0} A_b
    \, \mathrm{e}^{b\sigma} \;,
\end{displaymath}
which converges uniformly on compact subsets of $T^+$. In this
expression $b$ and $\sigma$ denote the vectors $b = (b_1, \ldots ,
b_r)$ and $\sigma = (\sigma_1 , \ldots , \sigma_r)$, respectively,
and $b \sigma := \sum b_i \sigma_i\,$. Following the usual
multi-index notation, $b\ge 0$ means $b_i \ge 0$ for all $i\,$. Note
$A_0 = 1$.

Now we know that the character has a convergent series representation
\begin{displaymath}
    \chi = \sum _{\gamma \in \Gamma _\lambda} B_\gamma
    \, \mathrm{e}^{\gamma} \;.
\end{displaymath}
Thus we may write
\begin{displaymath}
    J\chi = \sum_{\gamma \in \Gamma_\lambda} B_\gamma \sum_{b\ge 0}
    A_b \,\mathrm{e}^{\delta + \gamma + b\sigma}\;.
\end{displaymath}
For convenience of the discussion we let $\tilde{\gamma} := \gamma +
b \sigma$ and reorganize the sums as
\begin{equation}\label{ready for recursion}
    J\chi =\sum _{\tilde \gamma} \Big( \sum A_b\, B_{\tilde{\gamma}
    -b \sigma} \Big)\, \mathrm{e}^{\delta + \tilde\gamma} \;,
\end{equation}
where the inner sum is a finite sum which runs over all $b\ge 0$ such
that $\tilde{\gamma}- b\sigma \in \Gamma_\lambda\,$.

We are now in a position to explain the recursion procedure which
shows that $\chi$ is unique. Start by applying $D_\ell $ to $J\chi $
as represented in the expression (\ref{ready for recursion}). Since
$\delta +\tilde\gamma$ is of the form $\sum (\mathrm{i} m_k \psi _k -
n_k \phi_k)$, we immediately see that it is an eigenfunction with
eigenvalue $E(\ell ,\tilde \gamma) := (-1)^\ell \sum (m_k^{2\ell} -
n_k^{2\ell})$. The functions $\mathrm{e}^{\delta +\tilde \gamma}$ in
the sum are independent eigenfunctions. Hence it follows that
\begin{equation}\label{recursion equations}
    0 = E(\ell ,\tilde\gamma) \sum A_b \, B_{\tilde\gamma - b\sigma}
\end{equation}
for all $\tilde\gamma$ fixed and then for all $\ell\in\mathbb{N}\,$.

From now on we consider the equations (\ref{recursion equations})
only in those cases where $\tilde\gamma$ is itself a weight of our
representation. (We have license to do so as only the uniqueness part
of Theorem \ref{unicity} is at stake here.) In this case we have the
following key fact.
\begin{lemma}\label{vanishing eigenvalues}
If $\gamma \in \Gamma_\lambda$ and the eigenvalue $E(\ell , \gamma)$
vanishes for all $\ell \in \mathbb{N}\,$, then $\gamma$ is the
highest weight $\lambda$.
\end{lemma}
\begin{proof}
Our first job is to compute $\delta$. For the list of even and odd
positive roots we refer the reader to $\S\ref{sect:simple-roots}$.
Direct computation shows that if $K = \mathrm{O}_N\,$, then
\begin{equation*}
    \delta = \sum_{k=1}^n (k-1) (\mathrm{i} \psi_{n-k+1} - \phi_k)\;.
\end{equation*}
The same computation for the case of $K = \mathrm{USp}_N$ shows that
\begin{equation*}
    \delta = \sum_{k=1}^n k\, (\mathrm{i} \psi_k - \phi_{n-k+1}) \;.
\end{equation*}
Now we write $\gamma = \sum_k (\mathrm{i} m_k \psi_k - n_k \phi_k)$
with the weight constraints $-\frac{N}{2} \le m_k \le \frac{N}{2} \le
n_k\,$. The assumption that $E(\ell, \gamma)$ vanishes for all $\ell$
means that
\begin{equation*}
    \sum _k (m_{n-k+1} + k - 1)^{2\ell} = \sum_k (n_k + k - 1)^{2\ell}
    \quad \text{for all}\ \ell
\end{equation*}
in the case of $K = \mathrm{O}_N\,$. In the case of $K =
\mathrm{USp}_N$ it means that
\begin{equation*}
    \sum_k (m_{n-k+1} + k)^{2\ell} = \sum_k (n_k + k)^{2\ell}
    \quad \text{for all}\ \ell \;.
\end{equation*}
In the second case the only solution for $m_k$ and $n_k$ satisfying
the weight constraints is the highest weight $\lambda$ itself. In the
first case there is one other solution, namely that which is obtained
from the highest weight by replacing $m_n = \frac{N}{2}$ by $m_n = -
\frac{N}{2}\,$. However, one directly checks that in the
$\mathrm{O}_N$ case, where $2\mathrm{i}\psi_n$ is not a root, it is
not possible to obtain such a $\gamma$ by adding some combination of
roots from $\mathfrak{g}^{(2)}$ to $\lambda\,$.
\end{proof}
We are now able to give the proof of the uniqueness statement of
Theorem \ref{unicity}.
\begin{proof}
We will determine $B_\gamma$ recursively, starting from $B_\lambda =
1$. Let $\gamma \not= \lambda$ be a weight that satisfies the weight
constraints. Then if $K = \mathrm{USp}_N$ we know that $E(\ell ,
\gamma)$ is non-zero for some $\ell\,$. It therefore follows from
equation (\ref{recursion equations}) and $A_0 = 1$ that
\begin{equation}\label{basic equation}
    0 = B_\gamma + \sum A_b \, B_{\gamma - b\sigma}\;,
\end{equation}
where the sum runs over all $b \ne 0$ (recall that $b \ge 0$ is
always the case) such that $\gamma - b\sigma \in \Gamma_\lambda\,$.
Since the weights $\gamma - b\sigma $ involved in the sum are smaller
than $\gamma$ in the natural partial order defined by the basis of
simple roots, equation (\ref{basic equation}) defines a recursion
procedure for determining all coefficients $B_\gamma\,$.

In the case of $K = \mathrm{O}_N$ we are confronted with the fact
that the weight $\gamma = \lambda - \mathrm{i} N \psi_n$ satisfies
the weight constraints and yet gives $E(\ell,\gamma) = 0$ for all
$\ell$. However, in this exceptional situation the conditions of
Theorem \ref{unicity} provide that $B_{\lambda - \mathrm{i}N\psi_n} =
0$. Thus the expansion coefficients $B_\gamma$ are still uniquely
determined by our recursion procedure.
\end{proof}

\subsubsection {Explicit solution of the differential equations}
\label{explicit solution}

As before let $\Delta^+$ be a set of positive roots of $\mathfrak{g} =
\mathfrak{osp} (U\oplus U^*)$ or $\mathfrak{osp}(\widetilde{U} \oplus
\widetilde{U}^*)$. We now decompose these sets as
\begin{displaymath}
    \Delta^+ = \Delta_\lambda^+ \cup (\Delta^+ \setminus
    \Delta_\lambda^+)\;,\quad \Delta_\lambda^+ := \{\alpha\in\Delta^+
    \mid \mathfrak{g}_\alpha \subset \mathfrak{g}^{(-2)}\} \;,
\end{displaymath}
which means that $\Delta^+ \setminus \Delta_\lambda^+$ is a set of
positive roots of $\mathfrak{g}^{(0)}$. Let $\Delta_\lambda^+$ be
further decomposed as $\Delta_\lambda^+ = \Delta_{\lambda,0}^+ \cup
\Delta_{\lambda,1}^+$ where $\Delta_{\lambda,0}^+$ and
$\Delta_{\lambda,1}^+$ are the subsets of even and odd
$\lambda$-positive roots. Then the function $J$ has a factorization
as $J = J_0\, Z^{-1} \mathrm{e}^{\delta^\prime}$ with
\begin{displaymath}
    J_0 = \frac{\prod_{\alpha \in \Delta_0^+ \setminus
    \Delta_{\lambda,0}^+} 2 \sinh\frac{\alpha}{2}}{\prod_{\beta
    \in \Delta_1^+ \setminus \Delta_{\lambda,1}^+} 2 \sinh
    \frac{\beta}{2}}\ ,\quad Z = \frac{\prod_{\beta \in
    \Delta_{\lambda,1}^+} (1 - \mathrm{e}^{-\beta})}{\prod_{
    \alpha\in\Delta_{\lambda,0}^+} (1-\mathrm{e}^{-\alpha})}\;,
\end{displaymath}
and $\delta^\prime = \frac{1}{2} ( \sum \alpha - \sum \beta)$ is half
the graded sum of $\lambda$-positive roots. For the case of $K =
\mathrm{O}_N$ one finds $\delta^\prime = - \frac{1}{2} \sum
(\mathrm{i}\psi_j - \phi_j) = - \lambda_{N=1}\,$, while for $K =
\mathrm{USp}_N$ one has $\delta^\prime = \lambda_1\,$.

The Weyl group $W$ acts on $T^+$ and therefore on functions on $T^+$.
Let $W_\lambda \subset W$ be the subgroup which stabilizes the
highest weight $\lambda = \lambda_N$ and thus the corresponding
function $\mathrm{e}^\lambda$ on $T^+$. Note that $W_\lambda$ is the
direct product of the permutations of the set $\{ \mathrm{e}^{\phi_j}
\}$ and the permutations of the set $\{ \mathrm{e}^{\mathrm{i}\psi_j}
\}$. The symmetrizing operator $S_W$ from $W_\lambda$-invariant
analytic functions to $W$-invariant analytic functions on $T^+$ is
given by
\begin{displaymath}
    S_W (f) := \sum_{[w] \in W/W_\lambda} w(f) \;.
\end{displaymath}

Notice that the function $\mathrm{e}^\lambda Z$ is $W_\lambda
$-invariant; the symmetrized function $S_W( \mathrm{e}^\lambda Z)$
then is $W$-invariant. We now wish to show that this function
coincides with our character $\chi$. In this endeavor, an obstacle
appears to be that $\mathrm{e}^\lambda \chi$ by Corollary
\ref{cor:weightexpansion} is a polynomial in the variables $\mathrm
{e}^{\mathrm{i}\psi_1}, \ldots, \mathrm{e}^{ \mathrm{i} \psi_n }$,
whereas the function $Z$ has poles at $\mathrm{e}^{ \mathrm{i}(\psi_j
+ \psi_k)} = 1$. Hence our next step is to show that these poles are
actually canceled by the process of $W$-symmetrization.
\begin{lemma}
The function $S_W (\mathrm{e}^{\lambda} Z)$ is holomorphic on
$\cap_{j=1}^n \{ \mathfrak{Re}\, \phi_j > 0 \}$.
\end{lemma}
\begin{proof}
An even root $\alpha \in \Delta_0^+$ is some linear combination of
either the functions $\phi_j$ or the functions $\psi_j\,$. Denoting
the latter subset of even roots by $\Delta_0^+(\psi) \subset
\Delta_0^+ \,$, let $\Sigma_\alpha \subset T$ for $\alpha \in
\Delta_0^+(\psi)$ be the complex submanifold
\begin{displaymath}
    \Sigma_\alpha := \{ t \in T^+ \mid \mathrm{e}^{\alpha}(t) = 1 \}.
\end{displaymath}
By definition, the function $S_W (\mathrm{e}^\lambda Z)$ is
holomorphic on
\begin{displaymath}
    \Big( \cap_{j=1}^n\{\mathfrak{Re}\, \phi_j > 0 \} \Big)
    \setminus \left( \cup_{\alpha \in \Delta_0^+(\psi)}
    \Sigma_\alpha \right) \;.
\end{displaymath}
Now it is a theorem of complex analysis that if a function is
holomorphic outside an analytic set of complex codimension at least
two, then this function is everywhere holomorphic. Therefore, since
the intersection of two or more of the submanifolds $\Sigma_\alpha$
is of codimension at least two in $T^+$, it suffices to show that for
any $\alpha \in \Delta_0^+(\psi)$ our function $S_W (\mathrm{e}
^\lambda Z)$ extends holomorphically to
\begin{displaymath}
    D_\alpha := \Sigma_\alpha \setminus \left(\cup_{\Delta_0^+(\psi)
    \ni \alpha^\prime \not=\alpha}\,\Sigma_{\alpha^\prime}\right)\;.
\end{displaymath}

Hence let $\alpha$ be some fixed root in $\Delta_0^+(\psi)$. There
exists a Weyl-group element $w \in W$ and a $w$-invariant
neighborhood $U$ of $D_\alpha$ such that $w : \, U \to U$ is a
reflection fixing the points of $D_\alpha\,$. Let $z_\alpha : \, U
\to \mathbb{C}$ be a complex coordinate which is transverse to
$D_\alpha$ in the sense that $w(z_\alpha) = - z_\alpha\,$. Because
the root $\alpha$ occurs at most once in the product $Z$, the
function $S_W(\mathrm{e}^\lambda Z)$ has at most a simple pole in
$z_\alpha\,$. We may choose $U$ in such a way that $S_W(\mathrm{e}
^\lambda Z)$ is holomorphic on $U \setminus D_\alpha\,$. Doing so we
have a unique decomposition
\begin{displaymath}
    S_W (\mathrm{e}^\lambda Z) = \frac{A}{z_\alpha} + B
\end{displaymath}
where $A$ and $B$ are holomorphic on $U$. Since $S_W(\mathrm{e}
^\lambda Z)$ is $W$-invariant, we conclude that $w(A) = - A$ and
hence $A = 0$ along $D_\alpha\,$.
\end{proof}
\begin{lemma}
For all $\ell\, , N \in \mathbb{N}$ the function $\varphi :\, T^+ \to
\mathbb{C}$ defined by $\varphi = S_W(\mathrm{e}^{\lambda_N} Z)$ is a
solution of the differential equation $D_\ell \, J \varphi = 0\,$.
\end{lemma}
\begin{proof}
Using $Z = \mathrm{e}^{\delta^\prime} J_0 / J$ we write $\varphi =
S_W(\mathrm{e}^{\lambda_N + \delta^\prime} J_0/J)$. Then, lifting the
sum over cosets $[w] \in W/W_\lambda$ to a sum over Weyl-group
elements $w \in W$ we obtain
\begin{displaymath}
    \mathrm{ord}(W_\lambda) \, J \varphi = J \sum_{w \in W}
    w(\mathrm{e}^{\lambda_N + \delta^\prime} J_0 / J) =\sum_{w \in W}
    \mathrm{sgn}(w)\, w(J_0\,\mathrm{e}^{\lambda_N+\delta^\prime})\;,
\end{displaymath}
where $w \mapsto \mathrm{sgn}(w) \in \mathbb{Z}_2 = \{ \pm 1 \}$ is
the determinant of $w \in W \subset \mathrm{O}(\mathfrak{t}) =
\mathrm{O}_{2n}\,$.

The factor $\mathrm{e}^{\lambda_N + \delta^\prime}$ is the character
of the representation $(\frac{N \mp 1}{2} \mathrm{STr}\, ,
\mathrm{SDet}^{\frac{N \mp 1}{2}})$ of (a double cover of) the Lie
supergroup $(\mathfrak{g}^{(0)} , \mathrm{GL}(U_0) \times
\mathrm{GL}(U_1))$. This representation is one-dimensional, and from
Corollary \ref{cor:4.4} we have $D_\ell (J_0\, \mathrm{e}^{ \lambda_N
+ \delta^\prime}) = 0$ for all $\ell , N \in \mathbb{N}\,$.

The statement of the lemma now follows by applying the $W$-invariant
differential operator $D_\ell$ to the formula for
$\mathrm{ord}(W_\lambda) J \varphi$ above.
\end{proof}

\subsubsection{Weight constraints}

Here we carry out the final step in proving the explicit formula for
the character $\chi$ of our representation. Since the formula in the
case of $K = \mathrm{SO}_N$ follows directly from that for $K =
\mathrm{O}_N$ (see $\S\ref{varying}$) and the case of $K =
\mathrm{U}_N$ has been handled in \cite{CFZ}, we need only discuss
the cases of $K = \mathrm{O}_N$ and $K = \mathrm{USp}_N\,$.

In order to show that the character is indeed given by $\chi =
\varphi$ with $\varphi = S_W(\mathrm{e}^\lambda Z)$, it remains to
prove that in the series development $\varphi = \sum B_{\gamma}\,
\mathrm{e}^\gamma$ of the function defined by
\begin{equation}\label{key expression}
    \varphi = \sum_{[w] \in W/W_\lambda} \varphi_{[w]} \;, \quad
    \varphi_{[w]} := \mathrm{e}^{w(\lambda _N)}
    \frac{\prod_{\beta \in \Delta _{\lambda,1}^+} (1 - \mathrm{e}^{
    -w(\beta)})} {\prod_{\alpha \in \Delta_{\lambda,0}^+}
    (1 - \mathrm{e}^{-w(\alpha)})} \;,
\end{equation}
the only non-zero coefficients $B_\gamma$ are those where the linear
functions $\gamma$ are of the form $\gamma = \sum (\mathrm{i} m_k
\psi_k - n_k \phi_k)$ with $- \frac{N}{2} \le m_k \le \frac{N}{2} \le
n_k\, $. We also have to show that $B_\gamma = 0$ in the case of the
exceptional weight $\gamma = \lambda - \mathrm{i}N\psi_n$ occurring
for $K = \mathrm{O}_N\,$.

We have shown above that $\varphi = S_W(\mathrm{e}^\lambda Z)$ is
holomorphic on the product of the full complex torus of the variables
$\mathrm{e}^{\mathrm{i} \psi_k}$ with the domain defined by
$\mathfrak{Re}\, \phi_k > \,0\,$. Although the individual terms
$\varphi_{ [w]}$ in the representation of $\varphi$ have poles (which
cancel in the Weyl-group averaging process) we may still develop each
term of $\varphi$ in a series expansion; this will in fact yield the
desired weight constraints.

We begin with the situation where $K = \mathrm{O}_N\,$. In this case
$\Delta ^+_{\lambda,0}$ consists of the roots $\mathrm{i} \psi_j +
\mathrm{i} \psi_k\,$ ($j<k$) and $\phi_j + \phi_k$ ($j \le k)$, and
$\Delta^+_{ \lambda,1}$ is the set of roots of the form $\mathrm{i}
\psi_j + \phi_k$ ($j,\,k = 1, \ldots, n$). Let us first consider the
term of $\varphi$ where $[w] = W_\lambda\,$. Its denominator can be
developed in a geometric series on the region corresponding to
$\mathfrak{Re}\, \phi_k > 0$ for all $k\,$. There we may write this
term as
\begin{displaymath}
    \varphi_{[\mathrm{Id}]}=\mathrm{e}^{\lambda_N}\prod_{\beta}
    (1-\mathrm{e}^{-\beta}) \prod_\alpha \sum_{n \ge 0}
    \mathrm{e}^{-n\alpha} \;.
\end{displaymath}
Here and for the remainder of this paragraph $\alpha$ runs through
the $\lambda$-positive even roots and $\beta$ through the
$\lambda$-positive odd roots.

Recall that $\lambda_N = \frac{N}{2}\sum (\mathrm{i} \psi_k - \phi_k
)$, and note that all powers of $\mathrm{e}^{\mathrm{i} \psi_k}$ and
$\mathrm{e}^{\phi_k}$ occurring in the series expansion of
\begin{equation*}
\prod_\beta (1 - \mathrm{e}^{-\beta}) \prod_\alpha \sum_{n \ge 0}
    \mathrm{e}^{- n\alpha}
\end{equation*}
are non-positive. Thus, if $\gamma = \sum (\mathrm{i} m_k \psi_k -
n_k \phi_k)$ is a weight which arises in $\varphi_{[\mathrm{Id}]}\,$,
then $n_k \ge \frac{N}{2}$ and $m_k \le \frac{N}{2}$.  In the case of
the $m_k$ this is a statement only about the term $\varphi_{
[\mathrm{Id}]}\,$, but, since the action of the Weyl group on the
variables $\phi_k$ is just by permutation of the indices, it follows
that $n_k \ge \frac{N}{2}$ holds always, independent of the term
$\varphi_{[w]}$ under consideration. Hence, we ignore the $\phi_k$ in
our further discussion and only analyze the powers of the
exponentials $\mathrm{e}^{\mathrm{i} \psi_k}$ which arise in the
other terms $\varphi_{[w]}\,$.

Given a fixed index $k \in \{1, \ldots, n\}$ we will develop every
term $\varphi_{[w]}$ on a region $R = R(k)$ defined by certain
inequalities which in the case of $k = 1$ are
\begin{displaymath}
    \mathfrak{Re}(\mathfrak{i} \psi_1) > \ldots > \mathfrak{Re}
    (\mathrm{i} \psi_n) > 0 \;.
\end{displaymath}
We now discuss this case in detail.

Recall that $\mathrm{i} \psi_1$ occurs in the denominator in factors
of the form
\begin{displaymath}
    (1-\mathrm{e}^{-w(\alpha)})^{-1} = (1 - \mathrm{e}^{-
    w(\mathrm{i}\psi_1) - w(\mathrm{i} \psi_j)})^{-1}
\end{displaymath}
for $j > 1\,$. If $w(\mathrm{i}\psi_1) = \mathrm{i}\psi_1\,$, then we
expand these factors just as in the case of $\varphi_{ [\mathrm{Id}]}
\,$. Convergence of the resulting series is guaranteed no matter what
$w$ does to $\psi_j\,$.

In the situation where $w(\mathrm{i} \psi_1) = - \mathrm{i} \psi_1$
we rewrite the factors in the denominator as $(1-\mathrm{e}^{-w
(\alpha)})^{-1} = -\mathrm{e}^{w(\alpha)} (1 - \mathrm{e}^{w(\alpha)
})^{-1}$ and expand, and convergence in $R$ is again guaranteed.
Adding these series we obtain a series representation
\begin{displaymath}
    \varphi = \sum_{[w]} \varphi_{[w]} = \sum_\gamma B_\gamma \,
    \mathrm{e}^\gamma \;,
\end{displaymath}
which is convergent on $R\,$.
\begin{lemma}\label{m_1 estimate}
If $\gamma = \sum (\mathrm{i} m_k \psi_k - n_k \phi_k)$ and $B_\gamma
\not= 0\,$, then $m_1 \le \frac{N}{2}\, $.
\end{lemma}
\begin{proof}
If $w(\mathrm{i} \psi_1) = \mathrm{i} \psi_1\,$, then by the same
argument as in the case of $[w] = [\mathrm{Id}]$ we see that
$\mathrm{e}^{ \mathrm{i} \psi_1}$ occurs in the series development of
$\varphi_{[w]}$ with a power $m_1$ of at most $\frac{N}{2}\,$.

Now suppose that $w(\mathrm{i}\psi_1) = - \mathrm{i} \psi_1\,$. Then,
following the prescription above we rewrite the $\psi_1$-dependent
factors in $\varphi_{[w]}$ as
\begin{equation}\label{eq:4.7mrz}
    \mathrm{e}^{\frac{N}{2} w(\mathrm{i}\psi_1)} \frac{\prod_{j \ge 1}
    (1 - \mathrm{e}^{-w(\mathrm{i}\psi_1 + \phi_j)})} {\prod_{j \ge 2}
    (1 - \mathrm{e}^{-w(\mathrm{i}\psi_1 + \mathrm{i}\psi_j)})} =
    \mathrm{e}^{(-\frac{N}{2}+1) \mathrm{i}\psi_1}\frac{\prod_{j\ge 1}
    (\mathrm{e}^{-\mathrm{i}\psi_1} - \mathrm{e}^{-\phi_j})}{\prod_{j
    \ge 2}(\mathrm{e}^{- \mathrm{i}\psi_1}- \mathrm{e}^{-w(\mathrm{i}
    \psi_j)})} \;,
\end{equation}
and expand the r.h.s.\ in powers of $\mathrm{e}^{- \mathrm{i}
\psi_1}$. It follows that in this case $m_1 \le - \frac{N}{2} + 1\,$,
which for any positive integer $N$ implies that $m_1 \le
\frac{N}{2}\,$.
\end{proof}
Using Weyl-group invariance, this estimate for $m_1$ will now yield
the desired result.
\begin{lemma}\label{constraints fulfilled}
Suppose that $K = \mathrm{O}_N$ and let $\varphi = \sum B_\gamma\,
\mathrm{e}^\gamma$ be the globally convergent series expansion of the
proposed character $\varphi = S_W(\mathrm{e}^\lambda Z)$. Then for
every weight $\gamma = \sum (\mathrm{i} m_k \psi_k - n_k \phi_k)$
with $B_\gamma \not= 0$ it follows that $- \frac{N}{2} \le m_k \le
\frac{N}{2} \le n_k\,$.
\end{lemma}
\begin{proof}
The inequality $n_k \ge \frac{N}{2}$ was proved above as an immediate
consequence of the fact that the Weyl group $W$ effectively acts only
on the $\psi_j\,$.

Above we showed that on the region $R$ the proposed character
$\varphi$ has a series development where in every $\gamma$ the
coefficient $m_1$ of $\mathrm{i} \psi_1$ is at most $\frac{N}{2}\,$.
Recalling the fact that the function $\varphi$ is holomorphic on $T^+$,
we infer that $m_1 \le \frac{N}{2}$ also holds true for the globally
convergent series development $\sum B_\gamma\, \mathrm{e}^\gamma$.

To get the same statement for $\mathrm{i} \psi_k$ with $k \not= 1$ we
just change the definition of $R$ to $R(k)$ defined by the
inequalities $\mathfrak{Re} (\mathrm{i} \psi_k) > \mathfrak{Re}
(\mathrm{i} \psi_1) > \ldots > 0 \, $. Arguing for general $k$ as we
did for $k = 1$ in the above lemma, we show that the coefficient
$m_k$ of $\mathrm{i} \psi_k$ in every $\gamma$ in the series
expansion of every $\varphi_{[w]}$ on $R(k)$ is at most
$\frac{N}{2}\,$. By the holomorphic property, the same is true for
the global series expansion of the proposed character $\varphi$.

Hence, to complete the proof we need only show the inequality $m_k
\ge -\frac{N}{2}\,$. But for this it suffices to note that for every
$k$ there is an element $w$ of the Weyl group with $w(\mathrm{i}
\psi_k) = - \mathrm{i}\psi_k\,$. Indeed, using the Weyl invariance of
$\varphi$, if there was some $\gamma$ where $m_k < - \frac{N}{2}\,$,
then the coefficient of $\mathrm{i} \psi_k$ in $w(\gamma )$ would be
larger than $\frac{N}{2}\,$.
\end{proof}
To complete our work, we must prove Lemma \ref{constraints fulfilled}
for the case $K = \mathrm{USp}_N\,$. For this we use the same
notation as above for the basic linear functions, namely $\mathrm{i}
\psi_k$ and $\phi_k\,$. Here, compared to the $\mathrm{O}_N$ case,
there are only slight differences in the $\lambda$-positive roots and
the Weyl group. The only difference in the roots is in
$\Delta^+_{\lambda,0}$ where $\mathrm{i} \psi_j + \mathrm{i} \psi_k$
occurs in the larger range $j \le k$ and $\phi_j + \phi_k$ in the
smaller range $j < k\,$. The Weyl group acts by permutation of
indices on both the $\mathrm{i} \psi_j$ and $\phi_j$ and by sign
reversal on the $\mathrm{i} \psi_j\, $. In this case, as opposed to
the case above where only an even number of sign reversals were
allowed, every sign reversal transformation is in the Weyl group.

In order to prove Lemma \ref{constraints fulfilled} in this case, we
need only go through the argument in the $\mathrm{O}_N$ case and make
minor adjustments. In fact, the main step is to prove Lemma \ref{m_1
estimate} and, there, the only change is that the range of $j$ for
the factor $1 - \mathrm{e}^{-\mathrm{i}(\psi_1 + \psi_j)}$ is larger.
This is only relevant in the case $w(\mathrm{i} \psi_1) = -
\mathrm{i} \psi_1\,$, where we rewrite the additional denominator
term $(1 - \mathrm{e}^{-w(2\mathrm{i} \psi_1)})^{-1}$ as $-
\mathrm{e}^{-2\mathrm{i} \psi_1}(1 - \mathrm{e}^{- 2\mathrm{i} \psi
_1})^{-1}$. Hence the factor in front of the ratio of products on the
r.h.s.\ of equation (\ref{eq:4.7mrz}) gets an additional factor of
$\mathrm{e}^{-2\mathrm{i}\psi_1}$ and now is $\mathrm{e}^{-\mathrm{i}
(\frac{N}{2} + 1)\psi_1}$. Thus $m_1 \le -\frac{N}{2} -1$ which
certainly implies $m_1 \le \frac{N}{2}\,$.

Let us summarize this discussion.
\begin{theorem}
For both $K = \mathrm{O}_N$ and $K = \mathrm{USp}_N$ every weight
$\gamma = \sum (\mathrm{i} m_k \psi_k - n_k \phi_k)$ occurring in the
series expansion $S_W(\mathrm{e}^\lambda Z) = \sum B_\gamma\,
\mathrm{e}^\gamma $ obeys the weight constraints
\begin{displaymath}
    - \frac{N}{2} \le m_k \le \frac{N}{2} \le n_k \quad (k = 1,
    \ldots, n)\;.
\end{displaymath}
\end{theorem}
Moreover, using the fact that the Weyl-group transformations for $K =
\mathrm{O}_N$ always involve an even number of sign changes, one sees
that $B_{\lambda - \mathrm{i}N\psi_n} = 0$ in that case. As a
consequence of the uniqueness theorem (Theorem \ref{unicity}) we
therefore have
\begin{displaymath}
    \chi = S_W(\mathrm{e}^\lambda Z) = \sum_{[w] \in W/W_\lambda}
    \mathrm{e}^{w(\lambda_N)}\frac{\prod_{\beta\in\Delta_{\lambda,1}^+}
    (1 - \mathrm{e}^{-w(\beta)})}{\prod_{\alpha\in\Delta_{\lambda,0}^+}
    (1 - \mathrm{e}^{-w(\alpha)})}
\end{displaymath}
in both the $\mathrm{O}_N$ and $\mathrm{USp}_N$ cases. Since the
$\mathrm{SO}_N$ case has been handled as a consequence of the result
for $\mathrm{O}_N\,$, our work is now complete.

\bigskip\noindent
\textbf{Acknowledgement.} The work for this project was carried out with the support
of SFB/Tr 12, \emph{Symmetries and Universality in Mesoscopic Systems}, of the Deutsche
Forschungsgemeinschaft.

\end{document}